\documentclass[10pt]{article}

\usepackage[margin=1in,letterpaper]{geometry}
\usepackage[sorting=none,backend=biber]{biblatex}
\usepackage{amsmath, amsthm, amssymb, thmtools, thm-restate}
\usepackage{braket}
\usepackage{graphicx}
\usepackage{xspace}
\usepackage{xcolor}
\usepackage{mathtools}
\usepackage{complexity}
\usepackage{scalerel}
\usepackage{array}
\usepackage{soul}
\usepackage{tikz}
\usetikzlibrary{decorations.pathreplacing}
\usetikzlibrary{arrows.meta}
\usepackage{enumitem}
\usepackage{multicol}
\usepackage{subcaption}
\usepackage{qcircuit}
\usepackage{pdfpages}
\usepackage[titletoc]{appendix}
\usepackage[unicode, pdfencoding=auto]{hyperref}
\usepackage[capitalize]{cleveref}
\hypersetup{
  colorlinks = true,
  linkcolor = blue,
  anchorcolor = blue,
  citecolor = blue,
  urlcolor = blue
  }

\newcommand{\qvdots}{% \vdots for qcircuit
  \raisebox{0.3em}{\ensuremath{\vdots}}%
}
\newcommand{\qvdotss}{% \vdots for qcircuit
  \raisebox{0.7em}{\ensuremath{\vdots}}%
}

\definecolor{myblue}{HTML}{60B3E4}

\DeclareMathOperator{\langyes}{\mathcal{L}_{yes}}
\DeclareMathOperator{\langno}{\mathcal{L}_{no}}
\DeclareMathOperator{\Span}{span}
\DeclareMathOperator{\piinitpair}{\Pi_{\it init}^{\ket{\Phi^{\scaleto{+}{3.5pt}}}}}
\DeclareMathOperator{\piinitzero}{\Pi_{\it init}^{\ket{0}}}
\DeclareMathOperator{\piinitdot}{\Pi_{\it init}^{\ket{\cdot}}}
\DeclareMathOperator{\piinitcopy}{\Pi_{\it init}^{\scaleto{\ket{00,11}}{7pt}}}
\DeclarePairedDelimiter{\norm}{\lvert}{\rvert}

\newcommand{\etal}{\emph{et~al.}\xspace} 
\newcommand{\cupdot}{\mathbin{\mathaccent\cdot\cup}}

\crefname{thm}{Theorem}{Theorems}
\crefname{defn}{Definition}{Definitions}

\newtheorem{thm}{Theorem}
\newtheorem{defn}{Definition}[section]
\newtheorem{remark}{Remark}
\newtheorem{prop}{Proposition}[section]
\newtheorem{claim}{Claim}[section]
\newtheorem{coro}{Corollary}[section]
\newtheorem{lemma}{Lemma}[section]

\newcolumntype{M}[1]{>{\centering\arraybackslash}m{#1}}

\title{Quantum SAT Problems with Finite Sets of Projectors are Complete for a Plethora of Classes}

\author{Ricardo Rivera Cardoso$^{1,}$\thanks{ricardo.rivera@savba.sk} \,, Alex Meiburg$^{2,3, }$\thanks{bqp@ohaithe.re} \,, and Daniel Nagaj$^{1, }$\thanks{daniel.nagaj@savba.sk} \ \\
\hfill \break \\
\normalsize
  \textit{$^1$RCQI, Institute of Physics, Slovak Academy of Sciences, Bratislava, Slovakia} \\ \normalsize
  \textit{$^2$Perimeter Institute for Theoretical Physics, Waterloo, Canada} \\ \normalsize
  \textit{$^3$Institute for Quantum Computing, University of Waterloo}
}

\date{}

\begin{document}

\maketitle

\begin{abstract}
  Previously, all known variants of the Quantum Satisfiability (QSAT) problem---consisting of determining whether a $k$-local ($k$-body) Hamiltonian is frustration-free---could be classified as being either in $\mathsf{P}$; or complete for $\mathsf{NP}$, $\mathsf{MA}$, or $\mathsf{QMA_1}$. Here, we demonstrate new qubit variants of this problem that are complete for $\mathsf{BQP_1}$, $\mathsf{coRP}$, $\mathsf{QCMA}$, $\mathsf{PI(coRP,NP)}$, $\mathsf{PI(BQP_1,NP)}$, $\mathsf{PI(BQP_1,MA)}$, $\mathsf{SoPU(coRP,NP)}$, $\mathsf{SoPU(BQP_1,NP)}$, and $\mathsf{SoPU(BQP_1,MA)}$. Our result implies that a complete classification of quantum constraint satisfaction problems (QCSPs), analogous to Schaefer's dichotomy theorem for classical CSPs, must either include these 13 classes, or otherwise show that some are equal. Additionally, our result showcases two new types of QSAT problems that can be decided efficiently, as well as the first nontrivial $\mathsf{BQP_1}$-complete problem.

  We first show there are QSAT problems on qudits that are complete for $\BQP_1$, $\coRP$, and $\QCMA$. We construct these problems by restricting the finite set of Hamiltonians to consist of elements similar to $H_{init}$, $H_{prop}$, and $H_{out}$, seen in the circuit-to-Hamiltonian transformation. Usually, these are used to demonstrate hardness of QSAT and \textit{Local Hamiltonian} problems, and so our proofs of hardness are simple. We modify these terms to involve high-dimensional data and clock qudits, ternary logic, and either monogamy of entanglement or specific clock encodings to ensure that all Hamiltonians generated with these three elements can be decided in their respective classes. These problems can then expressed in terms of qubits, as we prove the non-trivial fact that any QCSP can be reduced to a qubit problem while maintaining the same complexity---something believed not to be possible classically. The remaining six problems are obtained by considering ``sums'' and ``products'' of the previous seven QSAT problems mentioned here. Before this work, the QSAT problems generated in this way resulted in complete problems for $\mathsf{PI}$ and $\mathsf{SoPU}$ classes that were trivially equal to $\NP$, $\MA$, or $\QMA_1$. We thus commence the study of these new and seemingly nontrivial classes.

  While [Meiburg, 2021] first sought to prove completeness for $\coRP$, $\BQP_1$, and $\QCMA$, we note that his constructions are flawed, leading to incorrect proofs of these statements. Here, we rework these constructions, as well as obtain improvements on the required qudit dimensionality.
\end{abstract}

\newpage

\section{Introduction}

Many of the interesting and puzzling phenomena in many-body physics occurs at the ground state of materials. One way to study quantum systems in this state is through their ground state energy, as this quantity can be used to provide information about some of the physical and chemical properties of the system. It is thus of great interest to calculate or even estimate this quantity. This task is embodied by the $k$-\textsc{Local Hamiltonian} ($k$-LH) problem. Specifically, given a \textit{k-local} ($k$-body) Hamiltonian---an operator of the form $H = \sum_{i} h_i$ where each $h_i$ acts on at most $k$ qubits---and two numbers $a,b \in \mathbb{R}$ with $b - a \geq 1/poly(n)$, this problem consists of distinguishing between the cases where $H$ has an eigenvalue less than $a$ or greater than $b$. Kitaev \cite{kitaev2002classical} showed that $k$-LH with $k \geq 5$ (and later improved to $k \geq 2$ \cite{2lh}) is unlikely to be decided efficiently with a classical or quantum computer. In complexity theory terms, $k$-LH with $k \geq 2$ is $\QMA$-complete.\footnote{The class $\QMA$ can be thought of as the quantum analog of $\NP$, or more accurately $\MA$ since the class has probabilistic acceptance and rejection.}

The LH problem is considered a ``weak'' quantum constraint satisfaction problem (QCSP) as states with energy less than $a$ do not necessarily minimize the energy of each $h_i$. For this reason, LH is often compared to MAX-$k$-SAT instead of the ``strong'' CSP $k$-SAT. Due to the immense importance of SAT in classical complexity and other hard sciences, Bravyi \cite{bravyi2006efficient} defined the \textsc{Quantum} $k$-SAT ($k$-QSAT) problem. Given a set of $k$-local projectors (also referred as \textit{clauses} or \textit{constraints}) and a number $b \in \mathbb{R}$, this problem consists of distinguishing between the cases where there exists a state that simultaneously lies in the null space of all projectors, or for all states, the penalty incurred by violations of the constraints is greater than $b$.\footnote{Alternatively, this problem can be defined with local Hamiltonians instead of projectors, in which case, the problem is equivalent to determining whether the Hamiltonian is frustration-free.} Bravyi showed that $2$-QSAT on qubits is in $\P$ while $k$-QSAT with $k \geq 4$ (and later improved to $k \geq 3$ \cite{gosset2016quantum}) is $\QMA_1$-complete when using the Clifford+T gate set $\mathcal{G}_8 = \{H, {\rm CNOT}, T\}$.\footnote{$\QMA_1$ is the one-sided error variation of $\QMA$ with perfect completeness, i.e.\ instances for which the answer is ``yes'' (in this case frustration-free Hamiltonians) are accepted with certainty. The notation $\mathcal{G}_8$ stems from Ref.\ \cite{cyclotomicset} and denotes the Clifford-cyclotomic gate set of degree of $8$. The reason why it is necessary to specify the gate set for classes with perfect completeness is discussed in \cref{subsubsection:perfect}.}

Interestingly, these two problems have in common that they are in $\P$ for a certain $k$ but appear to become much harder for $k+1$: LH is in $\P$ for $k \leq 1$ and becomes $\QMA$-complete for $k > 1$, while QSAT is in $\P$ for $k \leq 2$ and $\QMA_1^{\mathcal{G}_8}$-complete for $k > 2$. This is not entirely surprising since the Hamiltonians considered in the problems have no restriction other than their locality, and perhaps the difficulty lies in deciding ``unphysical'' Hamiltonians. Following this line of thought, others have considered variations of these problems where the $h_i$ are drawn from more realistic and relevant sets that satisfy some property or correspond to a physical model. To name a few, these may be stoquastic \cite{stoqma}, commuting \cite{lhcommuting}, fermionic \cite{lhfermionic}, bosonic \cite{lhbosonic}, or from models like the Heisenberg \cite{lhheisenberg} and Bose-Hubbard \cite{lhbosehubbard}. In addition, one might also consider placing restrictions on the geometry of the problem \cite{lh2dlattice,lhline8,2dclh,2dclhimprovement,2dstoquasticlh}.

In a landmark result, Cubitt and Montanaro \cite{cubitt2016complexity} showed that any LH problem where the $h_i$ are drawn from a finite set of at most $2$-local qubit Hermitian matrices can be classified as being either in $\P$, $\NP$-complete, $\mathsf{StoqMA}$-complete, or $\QMA$-complete.\footnote{$\mathsf{StoqMA}$ is the class of problems equivalent to estimating the ground state energy of the transverse-field Ising model \cite{timstoqma}.} As decision problems in the latter three classes are not known to be efficiently solvable in either classical or quantum computers, they showed that the only Hamiltonians of this type for which the LH problem can be solved efficiently are those with only $1$-local terms. This is significant, as many relevant Hamiltonians in nature can be approximated by $2$-local Hamiltonians of this type (e.g.\ all those supported on Pauli operators like Heisenberg and Ising spin glass models), and it is then likely that estimating their ground state energy efficiently lies outside of reach. Moreover, their result has led to a much larger repertoire of problems from which to construct reductions and potentially show the complexity of other computational problems.

Prior to our work, all known QSAT problems with finite or infinite sets of local interactions could be classified as being either in $\P$, $\NP$-complete, $\MA$-complete, or $\QMA_1$-complete, but this list is not known to be exhaustive in either case. The fact that QSAT has resisted classification can be attributed to two factors. First, is that since most relevant instances of QSAT can be decided classically (2-QSAT is in $\P$), there is a lack of interest to search for a classification of QSAT problems with $k > 2$. This is unlike in the LH problem where most relevant instances were hard (2-LH is $\QMA$-complete), motivating the study of Cubitt and Montanaro. Second, is the fact that QSAT problems are usually complete for classes that are harder to work with as they seem to depend on gate sets. In this work, it is our goal to concretize the implications that such a theorem may have, and hence motivate its study.

In this work, we show there are $9$ new QSAT problems with a finite set of $\mathcal{O}(1)$-local qubit clauses that are complete for $\BQP_1$, $\coRP$, and $\QCMA$, as well as six new classes that we introduce: $\mathsf{PI(coRP,NP)}$, $\mathsf{PI(BQP_1,NP)}$, $\mathsf{PI(BQP_1,MA)}$, $\mathsf{SoPU(coRP,NP)}$, $\mathsf{SoPU(BQP_1,NP)}$, and $\mathsf{SoPU(BQP_1,MA)}$. Together with the problems 2-SAT, 3-SAT, \textsc{Stoquastic 6-SAT}, and 3-QSAT that are complete for $\P$, $\NP$, $\MA$, and $\QMA_1$-complete, our results imply that a complete classification for strong QCSPs must either include these $13$ classes, or otherwise demonstrate that some of these are equal.\footnote{We refer to $k$-SAT as a QSAT problem since all of its instances can be realized using diagonal projectors.}$^,$\footnote{As defined, \textsc{Stoquastic 6-SAT} and 3-QSAT allow an infinite set of projectors. However, their proofs of hardness use a finite set of projectors, and hence there is a finite set of projectors for which the problems are $\MA$- and $\QMA_1$-complete.} We show that there are several interesting nontrivial relationships between these classes, and so a classification theorem demonstrating that there are fewer than $13$ classes could present exciting results. We thus motivate the search for such a theorem. As a corollary of our result, we show that there are new types of QSAT problems that can be solved efficiently, as well as the first nontrivial problem known to be complete for $\BQP_1$. Finally, as we discuss further in the text, the classes $\mathsf{PI(A,B)}$ and $\mathsf{SoPU(A,B)}$ are rarely mentioned in literature, and on the few occasions they have been considered, it has been for different $\mathsf{A}$ and $\mathsf{B}$ from those shown here \cite{Papadimitriou1984,Cai1988}. Our work then initiates the study of these classes, providing complete problems for each.

\subsection{Summary of results}

The notation used here is given in \cref{subsection:notation}. Our main result establishes that the QSAT problem SLCT-QSAT is $\BQP_1^{\mathcal{G}_8}$-complete. However, as the construction and analysis of this problem is contrived, we first show that the simpler and less optimized version of this problem, LCT-QSAT, is also complete for this class.

\begin{thm} \label{thm:bqplarge}
  The problem {\rm \textsc{Linear-Clock-Ternary-QSAT}} \textnormal{(LCT-QSAT; \cref{defn:lctqsat})} with $4$-local clauses acting on $17$-dimensional qudits is $\BQP_1^{\mathcal{G}_8}$-complete.
\end{thm}
\noindent
An interesting feature of this problem, and one that may be of independent interest, is that this problem makes clever use of the principle of monogamy of entanglement to strongly constrain the structure of input instances, facilitating the task of deciding whether they are frustration-free.\footnote{This construction is the most faithful to those considered by Meiburg in Ref.\ \cite{meiburg2021quantum}.} Unfortunately, this trick comes at a price of high qudit dimensionality. Our main result shows that by relaxing the constraint on the instance's structure and instead study the instances more closely, we can obtain a similar problem with the same complexity but with reduced qudit dimensionality.

\begin{thm} \label{thm:bqpsmall}
  The problem {\rm \textsc{Semilinear-Clock-Ternary-QSAT}} \textnormal{(SLCT-QSAT; \cref{defn:slctqsat})} with $4$-local clauses acting on $6$-dimensional qudits is $\BQP_1^{\mathcal{G}_8}$-complete.
\end{thm}
\indent
Recently, among many other interesting results, Rudolph \cite{rudolph2024onesided} demonstrated that $\BQP_1^{\mathcal{G}_{2^i}} = \BQP_1^{\mathcal{G}_{2^j}}$ for any $i,j \in \mathbb{N}$ (Theorem 3.4). In other words, any problem in $\BQP_1$ using a Clifford-cyclotomic gate set of degree $2^i$ can be perfectly simulated with one of degree $2^j$ for all $i,j \in \mathbb{N}$.

\begin{coro}
  The problems {\rm LCT-QSAT} and {\rm SLCT-QSAT} are $\BQP_1$-complete with any gate set $\mathcal{G}_{2^l}$ with $l \in \mathbb{N}$.
\end{coro}
\indent
Subsequently, by performing slight modifications to the clauses of SLCT-QSAT, we also obtain $\QCMA$-complete and $\coRP$-complete problems:

\begin{thm} \label{thm:qcma}
  The problem {\rm \textsc{Witnessed SLCT-QSAT}} \textnormal{(\cref{defn:wslctqsat})} with $4$-local clauses acting on $8$-dimensional qudits is $\QCMA$-complete.
\end{thm}

\begin{thm} \label{thm:corp}
  The problem {\rm \textsc{Classical SLCT-QSAT}} \textnormal{(\cref{defn:cslctqsat})} with $5$-local clauses acting on $8$-dimensional qudits is $\coRP$-complete.
\end{thm}
\indent
Then, using a similar application of monogamy of entanglement as in LCT-QSAT, we demonstrate that we can reduce any QCSP on qudits to another one on qubits.

\begin{thm}[informal] \label{thm:dtob}
  Every QCSP $\mathcal{C}$ on qudits is equivalent in difficulty to some other QCSP $\mathcal{C}'$ on qubits.
\end{thm}

\begin{coro} \label{coro:dtob}
  Together, \cref{thm:bqpsmall,thm:qcma,thm:corp} and \cref{thm:dtob} imply:
  \begin{enumerate}
    \item There is a $\BQP_1^{\mathcal{G}_8}$-complete {\rm QSAT} problem on qubits with $48$-local interactions.
    \item There is a $\QCMA$-complete {\rm QSAT} problem on qubits with $48$-local interactions.
    \item There is a $\coRP$-complete {\rm QSAT} problem on qubits with $60$-local interactions.
  \end{enumerate}
\end{coro}
\noindent
We refer to these problems by the same name as before, except that we now add a subindex to represent that the problem refers to the qubit version, e.g.\ SLCT-QSAT$_2$ is the QSAT problem that results from the reduction of SLCT-QSAT.

Finally, there is a notion of ``products'' and ``sums'' for both CSPs and QCSPs. Specifically, these are the \textit{direct product} ``$\otimes$'' and \textit{direct sum} ``$\oplus$'' (\cref{defn:prodqcsps,defn:sumqcsps}). Using this, we show that there are six new QSAT problems that are complete for classes $\mathsf{PI(A,B)}$ and $\mathsf{SoPU(A,B)}$, where $\mathsf{A}$ and $\mathsf{B}$ are themselves complexity classes. $\mathsf{PI(A,B)}$ stands for the \textit{pairwise intersection of classes} (\cref{defn:PI}), and $\mathsf{SoPU(A, B)}$ for the \textit{star of pairwise union of classes} (\cref{defn:SoPU}). Roughly, these two classes correspond to the sets of problems that can be expressed as the intersection and union (respectively) of a problem in $\mathsf{A}$ and a problem in $\mathsf{B}$.\footnote{These classes are not to be confused with $\mathsf{A} \cap \mathsf{B}$ and $\mathsf{A} \cup \mathsf{B}$. $\mathsf{A} \cap \mathsf{B}$ corresponds to the set of problems that are in both $\mathsf{A}$ and $\mathsf{B}$, while $\mathsf{A} \cup \mathsf{B}$ corresponds to those that are in either $\mathsf{A}$ or $\mathsf{B}$.} We show:

\begin{thm} \label{thm:pisopu}
  Let {\rm ``$\otimes$''} and {\rm ``$\oplus$''} denote the {\rm direct product} and {\rm direct sum} for quantum constraint satisfaction problems. Pairwise combinations of the four {\rm QSAT} problems---{\rm \textsc{3-SAT}}, {\rm \textsc{Classical SLCT-QSAT}$_2$}, {\rm \textsc{SLCT-QSAT}$_2$}, and {\rm \textsc{Stoquastic} $6$-SAT}---yield the following complete problems:
  \begin{enumerate}
    \item {\rm \textsc{Classical SLCT-QSAT$_2$}} $\otimes$ {\rm \textsc{3-SAT}} is $\mathsf{PI(coRP,NP)}$-complete.
    \item {\rm \textsc{Classical SLCT-QSAT$_2$}} $\oplus$ {\rm \textsc{3-SAT}} is $\mathsf{SoPU(coRP,NP)}$-complete.
    \item {\rm \textsc{SLCT-QSAT$_2$}} $\otimes$ {\rm \textsc{3-SAT}} is $\mathsf{PI(BQP_1^{\mathcal{G}_8},NP)}$-complete.
    \item {\rm \textsc{SLCT-QSAT$_2$}} $\oplus$ {\rm \textsc{3-SAT}} is $\mathsf{SoPU(BQP_1^{\mathcal{G}_8},NP)}$-complete.
    \item {\rm \textsc{SLCT-QSAT$_2$}} $\otimes$ {\rm \textsc{Stoquastic 6-SAT}} is $\mathsf{PI(BQP_1^{\mathcal{G}_8},MA)}$-complete.
    \item {\rm \textsc{SLCT-QSAT$_2$}} $\oplus$ {\rm \textsc{Stoquastic 6-SAT}} is $\mathsf{SoPU(BQP_1^{\mathcal{G}_8},MA)}$-complete.
  \end{enumerate}
\end{thm}
\indent
Finally, given that the QSAT problems in \cref{coro:dtob,thm:pisopu} consist of finite sets of projects with $\mathcal{O}(1)$-local qubit clauses, and similarly $2$-SAT, $3$-SAT, \textsc{Stoquastic} 6-SAT, and $3$-QSAT (which are respectively in $\P$, $\NP$-complete, $\MA$-complete and $\QMA_1^{\mathcal{G}_8}$-complete), our results imply that:

\begin{coro} \label{coro:classification}
  A complete classification theorem for strong QCSPs with $\mathcal{O}(1)$-local clauses acting on qubits must either include $13$ classes, or otherwise indicate that some of these are equal.
\end{coro}
\indent
The relationship between the $13$ classes mentioned here is shown in \cref{fig:inclusions}.

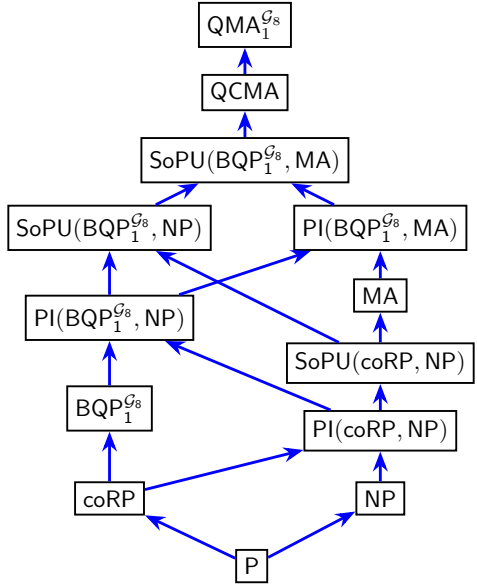
\begin{figure}
  \begin{center}
    \scalebox{0.9}{
      \begin{tikzpicture}
        \begin{scope}[every node/.style={thick,draw}]
          \node (P) at (2.1,0) {$\mathsf{P}$};
          \node (coRP) at (0,1) {$\mathsf{coRP}$};
          \node (NP) at (4,1) {$\mathsf{NP}$};
          \node (BQP) at (0,2.33) {$\mathsf{BQP_1^{\mathcal{G}_8}}$};
          \node (PIcoRPNP) at (4,2) {$\mathsf{PI(coRP,NP)}$};
          \node (SoPUcoRPNP) at (4,3) {$\mathsf{SoPU(coRP,NP)}$};
          \node (MA) at (4,4) {$\mathsf{MA}$};
          \node (PIBQPNP) at (0,3.67) {$\mathsf{PI(BQP_1^{\mathcal{G}_8},NP)}$};
          \node (SoPUBQPNP) at (0,5) {$\mathsf{SoPU(BQP_1^{\mathcal{G}_8},NP)}$};
          \node (PIBQPMA) at (4,5) {$\mathsf{PI(BQP_1^{\mathcal{G}_8},MA)}$};
          \node (SoPUBQPMA) at (2,6) {$\mathsf{SoPU(BQP_1^{\mathcal{G}_8},MA)}$};
          \node (QCMA) at (2,7) {$\mathsf{QCMA}$};
          \node (QMA) at (2,8) {$\mathsf{QMA}_1^{\mathcal{G}_8}$};
        \end{scope}

        \begin{scope}[>={Stealth[blue]},
          every edge/.style={draw=blue,very thick}]
          \path [->] (P) edge (coRP);
          \path [->] (P) edge (NP);
          \path [->] (coRP) edge (BQP);
          \path [->] (coRP) edge (PIcoRPNP);
          \path [->] (NP) edge (PIcoRPNP);
          \path [->] (PIcoRPNP) edge (SoPUcoRPNP);
          \path [->] (SoPUcoRPNP) edge (MA);
          \path [->] (BQP) edge (PIBQPNP);
          \path [->] (PIcoRPNP) edge (PIBQPNP);
          \path [->] (SoPUcoRPNP) edge (SoPUBQPNP);
          \path [->] (PIBQPNP) edge (SoPUBQPNP);
          \path [->] (SoPUBQPNP) edge (SoPUBQPMA);
          \path [->] (MA) edge (PIBQPMA);
          \path [->] (PIBQPNP) edge (PIBQPMA);
          \path [->] (PIBQPMA) edge (SoPUBQPMA);
          \path [->] (SoPUBQPMA) edge (QCMA);
          \path [->] (QCMA) edge (QMA);
        \end{scope}
      \end{tikzpicture}
    }
  \end{center}
  \caption{The classes for which we now have a complete strong QCSP, and their corresponding inclusions. In this work, we show completeness for quantum complexity classes with perfect completeness using the Clifford+T gate set $\mathcal{G}_8 = \{H, T, {\rm CNOT}\}$. Rudolph's \cite{rudolph2024onesided} result further strengthens ours by showing that $\BQP_1^{\mathcal{G}_8} = \BQP_1^{\mathcal{G}_{2^l}}$ for all $l \geq 1$. We discuss some of the inclusions in this figure in \cref{subsection:complexity} and \cref{section:sumproduct}.}
  \label{fig:inclusions}
\end{figure}

\subsection{Proof techniques}
\textbf{QSAT problems for $\BQP_1$, $\coRP$, and $\QCMA$} \\
\noindent
We keep the discussion general since each QSAT problem differs slightly from the others, and we think that a more detailed explanation should be delayed after the preliminary section. The first subsection of each section discussing a QSAT problem provides this more detailed and complete overview.

In summary, we prove \cref{thm:bqplarge,thm:bqpsmall,thm:qcma,thm:corp} by constructing these problems to consist of elements similar to $H_{init}$, $H_{prop}$, and $H_{out}$  seen in the circuit-to-Hamiltonian transformation originating in Refs.\ \cite{kitaev2002classical,feynman1986quantum}. Typically, these terms are local operators acting on a ``data'' and ``clock'' qubit register to encode the initialization, (unitary) propagation, and readout steps of a quantum circuit. In other words, one can encode a quantum circuit into the ground state of the local Hamiltonian $H = \sum_i (H_{init})_i + \sum_j (H_{prop})_j + \sum_k (H_{out})_k$. Typically, the quantum circuit decides a computational problem. If the circuit always accepts positive instances, i.e.\ measures the ``answer'' data qubit an obtains outcome ``1'' with certainty, $H$ is frustration-free. This transformation is hence useful to show hardness of QSAT problems. Because our operators, albeit slightly different, still implement these ideas, it follows that our QSAT problems are hard for $\BQP_1$, $\coRP$ and $\QCMA$. The difficulty---and the reason why modifications are necessary---is to show that these problems are contained within these classes.

Ideally, any instance formed from a polynomial number of these interaction terms should encode circuit evaluation, making frustration-freeness dependent solely on the readout. Then, the algorithm that decides the instance simply prepares the initial state, evaluates the circuit, and measures the answer qubit, accepting only if the outcome is ``1''. This procedure shows that the problem is in $\BQP_1$. If the circuit is instead a reversible classical circuit, the problem is in $\coRP$. In practice, however, the input instances generally do not encode a circuit, but rather form complex webs of interactions, making frustration-freeness hard to decide.

To address this, we redefine $H_{init}$, $H_{prop}$, and $H_{out}$ to act on high-dimensional data and clock \textit{qudit} registers, allowing us to embed ternary logic onto the data particles, as well as enforce a linear clock structure within the input instances. For example, in our first QSAT problem, we allow the clock particles to have two additional 2-dimensional subspaces and demand that neighboring clock qudits should be maximally entangled. By monogamy of entanglement, it is evident that each clock qudit can be entangled with at most two other clock qudits. Instances that violate monogamy are then not frustration-free, and so the other instance must have the linear structure mentioned, and hence express the evaluation of a quantum circuit.\\

\medbreak
\noindent
\textbf{Qubit QCSPs} \\
\noindent
We prove the non-trivial fact that any QCSP on qudits can be reduced to a QCSP of equal difficulty on qubits. While the standard mapping of decomposing a $d$-qudit into $\lceil \log_2 d \rceil$ qubits works for demonstrating that the resulting qubit instance does have the same satisfiability status as its parent qudit instance, it is not clear that the opposite holds. This is because any generic qubit instance can now have constraints that in the qudit instance would not exist. In short, the standard qudit-to-qubit mapping is not surjective in the set of problem instances, and the resulting QCSP includes some much harder instances. We show how to address this with a more careful mapping.\\

\medbreak
\noindent
\textbf{Product and sum QSAT problems} \\
\noindent
Our product and sum constructions ultimately derive from tensor products and tensor sums of Hilbert spaces. To prove that the resulting problems have the correct complexity classes, we show that satisfying states always respect the product (resp.\ sum) structure, and that conversely we can construct states in the product (resp.\ sum) problems from solutions in the original problems. There are some mild technical conditions used to ensure that we can convert between the product (resp.\ sum) QCSP and the originals efficiently, i.e.\ the mapping can be performed in $\P$.

\subsection{Discussion and open questions}

In \cref{thm:bqpsmall,thm:corp}, we show that there are two new types of QSAT problems that can be decided efficiently with a quantum or probabilistic classical computer. Unfortunately, the Hamiltonians used in these problems are artifacts built to achieve these results and do not immediately correspond to Hamiltonians of interest, even in the qubit case. Recently, interesting developments in the fields of quantum chemistry \cite{baiardi2022explicitly}, high-energy physics \cite{petiziol2021quantum} and nuclear physics \cite{busnaina2025nativethreebody,chuang2024halo,cruz2020probing} have shown that $3$- or $4$-local Hamiltonians are sometimes necessary to explain emergent physics. The QSAT problems for these Hamiltonians are not immediately tractable as they have locality $k >3$. It would thus be exciting to determine if these Hamiltonians, or others, fall within these complexity classes. We hope that having demonstrated that such problems exist, our results inspire others to search for more relevant cases.

Another interesting observation about \textsc{Classical SLCT-QSAT} is that while the problem is defined with quantum Hamiltonians and has a highly entangled zero-energy ground state, it is complete for a classical complexity class. It would then be interesting to try and reformulate this problem in purely classical terms. Aharanov and Grilo \cite{twocombinatorial} recently undertook a similar task for the $\MA$-complete \textsc{Stoquastic QSAT} problem and showed that it is equivalent to a combinatorial optimization problem.

\cref{thm:bqpsmall,thm:qcma,thm:corp} shows that there are seven complexity classes that have strong QCSPs as complete problems. Specifically, we introduce three new such problems to the already existing four. These seven classes belong to a larger set classes corresponding to polynomial-time computation and verification. Arguably, this set is sometimes considered the ``most natural'' or ``relevant'' set of complexity classes. We can differentiate the classes within this set by three defining properties: (1) type of circuit---classical or quantum; (2) type of witness state---no witness, classical witness, quantum witness; (3) type of error---no error, one-sided error, and two-sided error. Ranked in order of apparent difficulty, the seven classes mentioned here are (ignoring the gate set issue for the moment):

\begin{enumerate}
  \item $\P$: Classical, no witness, no error.
  \item $\mathsf{coRP}$: Classical, no witness, one-sided error.
  \item $\NP$: Classical, classical witness, no error.
  \item $\MA$: Classical, classical witness, two-sided error (and one-sided error).
  \item $\BQP_1$: Quantum, no witness, one-sided error.
  \item $\QCMA$: Quantum, classical witness, two-sided error (and one-sided error).
  \item $\QMA_1$: Quantum, quantum witness, one-sided error.
\end{enumerate}
\noindent
Are there any obvious omissions in this list for which we can expect another complete QCSP problem? By considering combinations of their three defining properties, there are a total of $18$ classes. Without the seven (technically 9 as $\MA = \MA_1$ and $\QCMA = \QCMA_1$) included here, the remaining $9$ classes without known QSAT complete problems are: (1-3) $\{c\}_{\rm circuit} \times \{q\}_{\rm witness} \times \{\varnothing, 1, 2\}_{\rm error}$, (4-6)
$\{q\}_{\rm circuit} \times \{\varnothing,c,q\}_{\rm witness} \times \{\varnothing\}_{\rm error}$, (7-8) $\{c,q\}_{\rm circuit} \times \{\varnothing\}_{\rm witness} \times \{\rm 2\}_{\rm error}$, and (9) $\{q\}_{\rm circuit} \times \{q\}_{\rm witness} \times \{\rm 2\}_{\rm error}$. The first set corresponds to classical classes with a quantum witness, however, we cannot have such classes and can then be discarded. The second set corresponds to quantum classes with no error. Little is known about these classes as they appear to be extremely difficult to work with. This is mainly due to the perfect completeness requirement which implies that no matter which instance (and witness in the case of verification classes), the circuit is never fooled into accepting. This is also why in the classification above we let one-sided error classes default to those with perfect completeness and bounded soundness instead of also considering the possibility of perfect soundness and bounded completeness. While it may be possible that there are complete QCSPs for these classes, we conjecture that this is very unlikely. The remaining three classes are $\BPP$, $\BQP$, and $\QMA$. Since $\coRP \subseteq \BPP$, $\BQP_1 \subseteq \BQP$, and $\QMA_1 \subseteq \QMA$ there are clearly strong QCSPs in this classes. Demonstrating that a QCSP is hard for them requires encoding a probabilistic circuit into an instance of this problem. Usually, these proofs are done via the circuit-to-Hamiltonian transformation that encodes the circuit into the ground state of a Hamiltonian. However, if the circuit is probabilistic, there is no state that satisfies all $\Pi_{init}$, $\Pi_{prop}$, and $\Pi_{out}$ clauses simultaneously. Thus, other techniques are needed, and not many are known.\footnote{Another technique is to reduce an already known hard problem into an instance of the target problem. For the LH problem, this is done via perturbation theory gadgets \cite{2lh,lhheisenberg,cubitt2016complexity}. However, these gadgets rely on approximations and therefore do not preserve perfect completeness.} Another approach that may offer a positive answer to this question is if these classes admit a scheme that boosts their acceptance probabilities to $1$. Jordan \etal \cite{jordan2011qcma} showed that this was possible for $\QCMA$ (demonstrating that $\QCMA = \QCMA_1$), but whether this is possible for the other classes remains an open question.

\cref{thm:pisopu} adds an additional $6$ classes to the set of classes with strong QCSPs complete problems. As mentioned, these results are obtained by considering ``sums'' and ``products'' of pairs of QSAT problems from our previous result. The first question one might ask is: if there are $7$ QSAT complete problems for the ``more natural'' classes, why are there only $6$ and not $21$ ``sum'' and ``product'' complete QSAT problems? While this is explained in more detail in \cref{section:sumproduct}, we note that this follows from a class property stating that if $\mathsf{A} \subseteq \mathsf{B}$, then $\mathsf{PI(A,B)} = \mathsf{SoPU(A,B)} = \mathsf{B}$. This is the reason why prior to our work, these classes were not considered. Indeed, the classes with complete QSAT problems were $\P$, $\NP$, $\MA$, and $\QMA_1$, but $\P \subseteq \NP \subseteq \MA \subseteq \QMA_1$, and so the classes $\mathsf{PI}$ and $\mathsf{SoPU}$ that resulted from considering pairs of these classes were trivial, i.e.\ they were equal to the largest of the two classes. With our results, there are now QSAT problems that are complete for classes not known to be contained within each other. These are: $\coRP \overset{?}\subseteq \NP$, $\NP \overset{?}\subseteq  \BQP_1$, and $\MA \overset{?}\subseteq  \BQP_1$. Hence, we only obtain $6$ problems that are complete for new and seemingly nontrivial classes. Besides the inclusions shown in \cref{fig:inclusions}, there is very little we know for certain about this classes. However, another thing to note, and one that is of importance in a future classification theorem \cref{coro:classification}, is that some of these classes may be related via derandomization conjectures. Among the complexity theory community, it is conjectured that $\P = \BPP$, implying that every problem solvable with a probabilistic classical computer can also be solved deterministically. If true, $\mathsf{PI(\coRP,\NP)} = \mathsf{PI(\P,\NP)}$ and $\mathsf{SoPU(\coRP,\NP)} = \mathsf{SoPU(\P,\NP)}$. Then, since $\P \subseteq \NP$, these two complexity classes are trivial and both are equal to $\NP$. There is also reason to believe that a weaker version of derandomization where $\NP = \MA$ may be true. Similarly, in this case, these two classes become equal to $\NP$. Now, observe that since it is not known whether $\BQP_1 \subseteq \NP$ or vice versa, and there is no strong evidence for either, the classes $\mathsf{PI(\BQP_1, \NP)}, \mathsf{SoPU(\BQP_1, \NP)}, \mathsf{PI(\BQP_1, \MA)}$, and $\mathsf{SoPU(\BQP_1, \MA)}$ would be expected to exhibit the most unique behavior.

Finally, let us consider the implications of \cref{coro:classification}. If such a theorem is proven and the classification contains fewer than $13$ classes, this could have exciting implications as the majority of these classes appear to be quite distinct from one another. Even if such a theorem proves any of the derandomization conjectures, it would be a great result since such proofs have eluded us for many decades. On the other hand, a classification showing that there are more than $13$ classes would be a stark contrast with classical strong CSPs, which can be completely classified as being either in $\P$ or $\NP$-complete \cite{schaeferdichotomy, zhukdichotomy}. This would highlight the more rich and complex panorama of strong QCSPs, and establish a larger repertoire of problems from which to construct reductions and potentially describe the complexity of other problems.

\subsection{Organization}

In \cref{section:preliminaries} we present a brief discussion of the complexity classes relevant for this paper, as well as a brief summary of Bravyi's proof that $k$-QSAT is $\QMA_1$-complete for $k \geq 4$ \cite{bravyi2006efficient}. Elements of this proof will be relevant in the next sections. In \cref{section:monogamy} we present the first and simplest QSAT problem, LCT-QSAT, and prove that it is $\BQP_1$-complete. Subsequently, in \cref{section:our}, we make several changes to the definition of this past problem and obtain SLCT-QSAT, which has better locality and particle dimensionality. We prove \cref{thm:bqpsmall}. In \cref{section:qcma,section:corp} we demonstrate that by making small changes to the definition of SLCT-QSAT, we can obtain two other QSAT problems that are complete for $\QCMA$ and $\coRP$, proving \cref{thm:corp,thm:qcma}. In \cref{section:universality}, we prove that every QCSP on qudits can be reduced to another one on qubits, while maintaining the same complexity. Finally, in \cref{section:sumproduct} we introduce the direct sum and product for QCSPs and demonstrate that these operations yield six new QSAT problems that are complete for the classes $\mathsf{PI}$ and $\mathsf{SoPU}$.

\section{Preliminaries} \label{section:preliminaries}

In this section, we briefly review some useful concepts for this work. In \cref{subsection:complexity}, we discuss some relevant classical and quantum complexity classes along with their one-sided error variation. Subsequently, in \cref{subsection:circtoh} we review the quantum satisfiability problem $k$-QSAT defined by Bravyi \cite{bravyi2006efficient} and his use of the circuit-to-Hamiltonian transformation to prove the problem is $\QMA_1$-hard. For readers already familiar with these topics, we simply refer them to the discussion about classes with perfect completeness in \cref{subsection:complexity}, and to \cref{eqn:clock} in \cref{subsection:circtoh} which points to the clock encoding used in the constructions of this paper.

\subsection{Notation} \label{subsection:notation}

For a bitstring $x$, let $\norm{x}$ denote the number of bits in $x$. For $n \in \mathbb{N}^+$, let $[n] = \{0,1,\ldots,n-1\}$.

For some complexity classes, we specify the gate set used. Here, we use the Clifford-cyclotomic gate sets $\mathcal{G}_m$ defined in Ref.\ \cite{cyclotomicset}. Specifically, we only consider those that are a power of two. These are: $\mathcal{G}_2 := \{X, {\rm CNOT}, {\rm Toffoli}, H \otimes H\}$, $\mathcal{G}_4 := \{X, {\rm CNOT}, {\rm Toffoli}, \zeta_8 H\}$, and for $l \geq 3$, $\mathcal{G}_{2^l} := \{H, {\rm CNOT}, T_{2^l}\}$. Here, $T_{2^l} = {\rm diag}(1, \zeta_{2^l})$ where $\zeta_{2^l} = e^{2 \pi i/2^l}$ is a primitive $2^l$-th root of unity.

In all quantum circuits considered here, we let $U_0 = I$ (sometimes also written as $U_{0,0}$). The same is true for classical circuits $Q$ and classical reversible circuits $R$. For circuits that decide computational problems, we let $ans$ denote the qubit that when measured provides this decision. We accept the instance if the qubit is measured and yields outcome ``1'', and reject otherwise. Usually, $ans$ is the first ancilla qubit of the circuit.

For a circuit $U_{n}$ that decides an instance $x$ with $\norm{x} = n$, we denote $U_x$ as the circuit where the instance $x$ is encoded into it and the inputs are only ancilla qubits in the $\ket{0}$ state.

\subsection{Complexity classes} \label{subsection:complexity}

A \textit{promise problem} $A = (A_{yes}, A_{no})$ is a computational problem consisting of two non-intersecting sets $A_{yes}, A_{no} \subseteq \{0,1\}^*$ where given an instance $x \in \{0,1\}^*$ (promised to be in one of the two sets), one is tasked to determine if $x \in A_{yes}$ ($x$ is a yes-instance) or $x \in A_{no}$ ($x$ is a no-instance).\footnote{The asterisk over the set is known as the \textit{Kleene star} and is used to represent strings of any finite size.} Let $n = \norm{x}$ denote the size of $x$.

The first complexity class we consider is that composed of promise problems that can be decided probabilistically on a classical computer.

\begin{defn}[\BPP]
    A promise problem $A$ is in $\BPP$ iff there exists a polynomial $p$ and a family of polynomial-time uniform classical algorithms $\{Q_n\}$ that take as input a binary string $x$ and a random bitstring $r \in \{0,1\}^{p(n)}$, such that:
    \begin{itemize}
        \item (Completeness) If $x \in A_{yes}$, $\Pr\left[Q_n \textnormal{ accepts } (x,r)\right] \geq 2/3$.
        \item (Soundness) If $x \in A_{no}$, $\Pr\left[Q_n \textnormal{ accepts } (x,r) \right] \leq 1/3$.
    \end{itemize}
\end{defn}
\noindent
The quantum analog of this class, where the classical algorithm becomes a quantum algorithm, is $\BQP$.

\begin{defn}[\BQP] \label{defn:bqp1}
    A promise problem $A$ is in $\BQP$ iff there exists a polynomial $q$ and a family of polynomial-time uniform quantum circuits $\{U_n\}$ that take as input a binary string $x$ and use at most $q(n)$ ancilla qubits, such that:
    \begin{itemize}
        \item (Completeness) If $x \in A_{yes}$, $\Pr \left[U_n \textnormal{ accepts } x \right] \geq 2/3$.
        \item (Soundness) If $x \in A_{no}$, $\Pr \left[U_n \textnormal{ accepts } x \right] \leq 1/3$.
    \end{itemize}
\end{defn}
\noindent
One can also consider a class where instead of computing the solution, one is tasked with verifying a given candidate solution (also known as \textit{witness}).

\begin{defn}[\QMA] \label{defn:qma}
    A promise problem $A$ is in $\QMA$ iff there exist polynomials $p,q$ and a family of polynomial-time uniform quantum (verifier) circuits $\{U_n\}$ that take as input a binary string $x$ and a witness state $\ket{\psi_x}$ of at most $p(n)$ qubits, and use at most $q(n)$ ancilla qubits, such that:
    \begin{itemize}
        \item (Completeness) If $x \in A_{yes}$, then there exists a state $\ket{\psi_x}$ such that $\Pr \left[ U_n \textnormal{ accepts } (x,\psi_x) \right] \geq 2/3$.
        \item (Soundness) If $x \in A_{no}$, then for any witness state $\ket{\psi_x}$, $\Pr \left[ U_n \textnormal{ accepts } (x,\psi_x) \right] \leq 1/3$.
    \end{itemize}
\end{defn}
\noindent
If we receive a classical binary bitstring $y_x$ as a witness instead of a quantum state, we obtain $\QCMA$.

\begin{defn}[\QCMA]
    $\QCMA$ is defined in a similar way as $\QMA$, except the quantum state $\ket{\psi_x}$ is replaced by the computational basis state $\ket{y_x}$.
\end{defn}

\subsubsection{Perfect completeness} \label{subsubsection:perfect}

In this paper we are interested in a variation of the classes above where the acceptance probability of yes-instances is equal to one. These classes are said to have \textit{perfect completeness}, and they are one of the two types of classes with \textit{one-sided error}. Although these classes appear to be similar to their two-sided error variation, quantum complexity classes with one-sided error require a more precise treatment as they are not known to be independent of the gate set used. Indeed, the Solovay-Kitaev theorem \cite{solovaykitaev} used to resolve this issue for quantum classes with two-sided error only works for approximate equivalence of universal gate sets and not perfect equivalence. Thus, for these classes (with some exceptions), one must specify the gate set used by the quantum circuits. This is not the case for classical complexity classes as it is known that every classical circuit using gate set $\mathcal{G}$ can be perfectly simulated by another circuit using a universal gate set $\mathcal{G}'$.

Given this discussion, we can then define one-sided error classes as follows:

\begin{defn}[Classes with perfect completeness] \label{defn:onesidederr}
    Let $\mathcal{C}$ be a complexity class with two-sided error. The variation of this class with perfect completeness is defined in a similar way to $\mathcal{C}$ except for the following differences:
    \begin{enumerate}
        \item The acceptance probability must be exactly $1$ when $x \in \mathcal{A}_{yes}$.
        \item If $\mathcal{C}$ is a quantum complexity class, the gate set $\mathcal{G}$ used by the quantum circuits $\{U_n\}$ must be specified.
    \end{enumerate}
    \noindent
    The class with one sided error is generally denoted as $\mathcal{C}_1$, or $\mathcal{C}_1^{\mathcal{G}}$ if it is a quantum complexity class.
\end{defn}
\noindent
This sensibility to the gate set in quantum complexity classes is the reason why, in \cref{thm:bqplarge,thm:bqpsmall}, we explicitly stated that LCT-QSAT and SLCT-QSAT are complete for $\BQP_1$ with the particular choice of gate set $\mathcal{G}_8$. It also presents other complications. To see this, consider $\BQP$. It is evident that $\BQP_1^\mathcal{G} \subseteq \BQP$ for any arbitrary gate set $\mathcal{G}$, and also that $\P \subseteq \BQP$. However, is it true that $\P \subseteq \BQP_1^\mathcal{G}$? Fortunately, for the Clifford+T gate set (i.e.\ $\mathcal{G}_8$) used in this paper, the class $\BQP_1^{\mathcal{G}_8}$ follows the intuitive containment of classes. Indeed, $\P$ and $\coRP$ are contained in this class. To see this, consider the fact that any classical circuit is equivalent to a reversible classical circuit, and moreover the circuit may consist only of Toffoli gates. Then, a result by Giles and Selinger \cite{giles2013exact} shows that each Toffoli can be efficiently and perfectly decomposed in terms of Clifford and $T$ gates. For any problem in $\P$ or $\coRP$, it can also be decided by a $\BQP_1^{\mathcal{G}_8}$ algorithm that runs the classical reversible circuit, replacing each Toffoli by its decomposition in terms of Clifford and $T$ gates. The circuit is guaranteed to have perfect completeness. The relationships between the eight classes mentioned here is given by $\coRP \subseteq \BQP_1^{\mathcal{G}_8} \subseteq \BQP \subseteq \QCMA_1^{\mathcal{G}_8} = \QCMA \subseteq \QMA_1 \subseteq \QMA$.

A relevant feature of these six classes is that the gap between the acceptance and rejection probabilities can be amplified. It can be shown that as long as these probabilities are separated by the inverse of a polynomial, there is a scheme (e.g.\ repetition and majority vote) that makes these probabilities exponentially close to $1$ and $0$, respectively, requiring only a polynomial number of extra steps \cite{kitaev2002classical, aharonov2002quantum}. We will make use of this statement to show that the quantum satisfiability problems mentioned in this paper meet the required soundness conditions.

Interestingly, Jordan \etal \cite{jordan2011qcma} showed that if the circuits that decide a $\QCMA$ problem consist of gates with a succinct representation (e.g.\ $\mathcal{G}_8$), the acceptance probability of yes-instances can be amplified additively to be exactly $1$. In other words, they showed that $\QCMA_1^{\mathcal{G}_8} \subseteq \QCMA$, concluding that $\QCMA_1^{\mathcal{G}_8} = \QCMA$. This is the reason why in \cref{thm:qcma}, we state that the problem \textsc{Witnessed SLCT-QSAT} is $\QCMA$-complete: in \cref{section:qcma}, we actually show that this problem is $\QCMA_1^{\mathcal{G}_8}$-complete, but by this result, it ends up being $\QCMA$-complete. To this day, it remains an open question whether perfect amplification can also work for $\BQP$ and $\QMA$. In the case of $\QMA$, it is believed that this is not the case as one can show that there exists an oracle for which $\QMA \neq \QCMA_1$  \cite{aaronsonqmacompleteness}. However, a similar claim was made about $\QCMA$ and $\QCMA_1$.

\subsection{k-QSAT \& the Circuit-to-Hamiltonian transformation} \label{subsection:circtoh}

Here, we introduce \textsc{Quantum} $k$-SAT (denoted here as $k$-QSAT) as defined by Gosset and Nagaj in Ref.\ \cite{gosset2016quantum}. We present relevant parts of the proofs showing that $k$-QSAT is contained in $\QMA_1$ for any constant $k$, and $\QMA_1$-hard for $k \geq 6$. While Bravyi's \cite{bravyi2006efficient} original work demonstrates hardness for $k \geq 4$, we choose to present this slightly weaker result for brevity, but also to introduce our clock encoding and notation useful for the rest of this paper.

As we are working to prove the inclusion and hardness of this problem for a class requiring perfect completeness, it is necessary to specify the gate set used by the quantum circuits. For now, let us simply refer to this gate set as $\mathcal{G}$, and we will later show specifically which gates should be used. In addition, we also have to be wary that all operations can be performed with perfect accuracy using gates from this set and all measurements are in the computational basis. For this purpose, Gosset and Nagaj introduce the following set of projectors.

\begin{defn}[Perfectly measurable projectors] \label{defn:P}
    Let $\mathcal{P}$ be the set of projectors such that every matrix element in the computational basis has the form
    \begin{equation*}
        \frac{1}{4}(a + ib + \sqrt{2}c + i\sqrt{2}d)
    \end{equation*}
    \noindent
    for $a,b,c,d \in \mathbb{Z}$.
\end{defn}
\indent
The (promise) problem $k$-QSAT can be defined as follows.

\begin{defn}[$k$-QSAT]\label{defn:kqsat}
    Given an integer $n$ and an instance $x$ consisting of a collection of projectors $\{\Pi_i\} \subset \mathcal{P}$ where each $\Pi_i$ acts nontrivially on at most $k$ qubits, the problem consists on deciding whether (1) there exists an $n$-qubit state $\ket{\psi_{sat}}$ such that $\Pi_i \! \ket{\psi_{sat}} = 0$ for all $i$, or (2) for every $n$-qubit state $\ket{\psi}$, $\Sigma_i \bra{\psi} \! \Pi_i \! \ket{\psi} \geq 1/poly(n)$. We are promised that these are the only two cases. We output ``YES'' if (1) is true, or ``NO'' otherwise.
\end{defn}
\indent
One can think of this problem as being presented with a list of constraints or \textit{clauses} (the projectors $\Pi_i$) and tasked with distinguishing between the following cases: (1) there exists a state a state that satisfies all constraints (a \textit{satisfying state}), or (2) any possible state induces a violation of the constraints greater than $1/poly(n)$. The promise sets the conditions for classifying instances as either $x \in \langyes$ or $x \in \langno$. Without this promise, the problem becomes seemingly harder as it requires distinguishing between the case where the projectors are satisfiable, and the case where they are not but the violation induced by some states could be exponentially close to zero. Without a promise, the problem is most likely not contained in $\QMA_1$.

\subsubsection{In \texorpdfstring{QMA\textsubscript{1}}{QMA\_1}} \label{subsubsection:inqma}

Suppose we are presented with a witness state $\ket{\psi_{\rm wit}}$ and a $k$-QSAT instance composed of projectors $\{\Pi_i\}$. The quantum algorithm that decides whether this state satisfies all projectors $\Pi_i$ consists of simply measuring the eigenvalues of all projectors on this state. Then, if all measured eigenvalues are $0$, we conclude that all projectors are satisfied by the state and output ``YES''. Otherwise, we reject.

Specifically, we measure the eigenvalue of a projector $\Pi_{i}$ by applying the unitary

\begin{equation*}
    V(\Pi_i) = \Pi_i \otimes X + (I - \Pi_i) \otimes I,
\end{equation*}
\noindent
to the witness and an additional ancilla qubit in the state $\ket{0}$, followed by a measurement of the ancilla in the computational basis. Here, $X$ denotes the Pauli-X gate. The probability that $\ket{\psi_{\rm wit}}$ does not satisfy projector $\Pi_i$ (obtain outcome ``1'') is given by

\begin{equation} \label{eqn:probmeasure}
    p_i = \bra{\psi_{\rm wit}} \Pi_i \ket{\psi_{\rm wit}}.
\end{equation}
\noindent
By defining the acceptance probability as the probability that all measurements produce outcome ``0'', and assuming $V(\Pi_i)$ can be implemented perfectly with gate set $\mathcal{G}$, one can show that this algorithm meets the completeness and soundness conditions of $\QMA_1$, and so $k$-QSAT is contained in this class. The proof of these statements is not shown here since a similar argument will be presented at the end of \cref{section:monogamy}.

As mentioned, to support this claim, it is necessary to demonstrate that $V(\Pi_i)$ can be implemented perfectly with gate set $\mathcal{G}$. In Ref.\ \cite{bravyi2006efficient}, Bravyi argued that this was in fact possible for all $k$-local projectors if these were picked from a field $\mathbb{F} \subset \mathbb{C}$ such that all of their matrix elements had an exact representation. However, this result was later withdrawn \cite{gosset2016quantum}. Gosset and Nagaj \cite{gosset2016quantum} later showed that Bravyi's proof could be fixed by restricting the projectors to be from the set $\mathcal{P}$ above. Their argument is based on Giles and Selinger's theorem, showing that any $m \times m$ unitary whose matrix elements are all of the form $1/2^s (a + ib + \sqrt{2}c + i\sqrt{2}d)$ where $s \in \mathbb{N}$ and $a,b,c,d \in \mathbb{Z}$, can be decomposed into a sequence of $\mathcal{O}(3^{m}s)$ gates from the Clifford+T gate set using a single additional ancilla qubit. Then, since $\Pi_i \subset \mathcal{P}$ and each $\Pi_i$ is a $k$-local projector with constant $k$, $V(\Pi_i)$ has the required form. Moreover, since $V(\Pi_i)$ is independent of $n$, i.e.\ a constant-sized matrix, it can be decomposed into polynomially-many gates of the Clifford+T set. This discussion then concludes that $k$-QSAT (as defined in \cref{defn:kqsat}) is in $\QMA_1$ if $\mathcal{G} = \mathcal{G}_8 = \{H, {\rm CNOT}, T\}$.

Recently, Rudolph \cite{rudolph2024onesided} generalized this result and brought it closer to Bravyi's original definition by showing that $k$-QSAT is complete for a field. Specifically, he showed that the $3$-QSAT problem, where all projectors have matrix elements that belong to the field $\mathbb{Q}(\zeta_{2^l})$, is complete for $\QMA_1^{\mathcal{G}_{{2^l}}}$ as long as $l \geq 3$.

\subsubsection{\texorpdfstring{QMA\textsubscript{1}}{QMA\_1}-hard} \label{subsubsection:qmahard}

Now, we discuss elements of the proof demonstrating that $k$-QSAT is $\QMA_1$-hard when $k \geq 6$ and for any gate set $\mathcal{G}$ that is universal for quantum computation.

To show this result, we have to prove that any instance $x$ of an arbitrary promise problem in $\QMA_1$ can be transformed or \textit{reduced} in polynomial time into an instance $x'$ of $k$-QSAT where the answer to the original problem and the transformed one is the same for all instances. Furthermore, we also need to show that all projectors of the resulting $k$-QSAT instance act on at most $6$ qubits.

Let $U_x = U_L \ldots U_1$ with $U_i \in \mathcal{G}$ and $L = poly(n)$ be the $\QMA_1$ verification circuit where given an instance $x$ of a problem $A = (A_{yes}, A_{no})$, $U_x$ decides whether $x \in A_{yes}$ or $x \in A_{no}$. The input to the circuit consists of the $p$-qubit witness state $\ket{\psi_{\rm wit}}$, and a $q$-qubit ancilla register $D$ (referred to as the \textit{data} register) initialized to the state $\ket{0}^{\otimes q}$, where $p$ and $q$ are two polynomials in $n = \norm{x}$. Additionally, let the answer be obtained by measuring one of the ancilla qubits (denoted by $ans$) in the computational basis, where outcome ``1'' means the instance is accepted, while outcome ``0'' means the instance is rejected. The goal of the reduction is to engineer a set of $6$-local projectors such that they are uniquely satisfied by the state encoding the evaluation of the circuit $U$ on $\ket{\phi_0} := \ket{0}^{\otimes q} \otimes \ket{\psi_{\rm wit}}$ at all steps of the computation. This state is appropriately known as the (computational) \textit{history state} and is given by

\begin{equation} \label{eqn:hist}
    \ket{\psi_{hist}} := \frac{1}{\sqrt{L+1}} \sum_{t = 0}^L U_t \ldots U_0 \ket{\phi_0}_D \otimes \ket{C_t}_C,
\end{equation}
\noindent
where $U_0 := I$ is a dummy unitary introduced for convenience. Here, we have introduced a \textit{clock} register $C$ acting on a new (not yet specified) Hilbert space used to keep track of the current step in the computation. Clearly, this history state can be defined in many ways depending on the implementation of the states $\ket{C_t}$. In this paper, we choose a clock encoding acting on $\mathcal{H}_{clock} = (\mathbb{C}^3)^{\otimes{L+1}}$ consisting of the \textit{ready} state $\ket{r}$, the \textit{active} state $\ket{a}$, and the \textit{dead} state $\ket{d}$, which progresses as

\begin{equation} \label{eqn:clock}
    \begin{aligned}
        \ket{C_0} & = \ket{a_0 r_1 r_2 \ldots r_L}, \\
        \ket{C_1} & = \ket{d_0 a_1 r_2 \ldots r_L}, \\
        \ket{C_2} & = \ket{d_0 d_1 a_2 \ldots r_L}, \\
        \vdots                                      \\
        \ket{C_L} & = \ket{d_0 d_1 d_2 \ldots a_L}.
    \end{aligned}
\end{equation}
\noindent
We refer to these basis states as the \textit{legal} states of the clock, and all other basis states as \textit{illegal}.

The projectors that allow us to build the required $k$-QSAT instance act on both of these Hilbert spaces and are given by

\begin{equation}\label{eqn:projs}
    \begin{aligned}
        P_{init}^{(i)}                                                                                                   & := \ket{1} \! \bra{1}_i \otimes \ket{a} \! \bra{a}_0,                                                                          \\
        P_{out}^{(i)}                                                                                                    & := \ket{0} \! \bra{0}_i \otimes \ket{a} \! \bra{a}_L,                                                                          \\
        P_{prop,U}^{(i)} := \frac{1}{2} \big[ I^{\otimes2} \otimes \ket{ar} \! \bra{ar}_{i-1,i} + I^{\otimes 2}  \otimes & \ket{da} \! \bra{da}_{i-1,i}  - U \otimes \ket{da} \! \bra{ar}_{i-1,i} - U^\dagger \otimes \ket{ar} \! \bra{da}_{i-1,i} \big],
    \end{aligned}
\end{equation}
\noindent
which receive an index to specify its action on another particle. Observe that $P_{init}$ and $P_{out}$ act on a single data and clock particle, while $P_{prop,U}$ acts on two data qubits and two clock particles. As each clock particle can be represented by two qubits, albeit a bit wastefully, it is evident that these projectors are at most $6$-local (on qubits). Other clock encodings may lead to different locality.\footnote{In Ref.\ \cite{bravyi2006efficient}, Bravyi employs a four-state clock encoding, $2L + 1$ clock basis states, and an additional propagation projector. This allows interactions between either two clock particles at a time or one clock particle and two data qubits, resulting in $4$-local projectors. However, this comes at a cost of increased clock particle dimensionality.}

Each projector in \cref{eqn:projs} penalizes states that do not meet certain requirements. (Initialization) $P_{init}$ requires that when clock particle $0$ is in the state $\ket{a}$, data qubit $i$ is initialized to $\ket{0}$. (Computational propagation) $P_{prop,U}$ requires that as clock particles $i$ and $i+1$ transition from $\ket{ar}$ to $\ket{da}$, $U$ is applied to two qubits of the data register. (Readout) Finally, $P_{out}$ requires that when clock qudit $L$ is in the state $\ket{a}$, data qubit $i$ is in the state $\ket{1}$.\footnote{Unlike Bravyi \cite{bravyi2006efficient} and Meiburg \cite{meiburg2021quantum}, we define $P_{out}$ so it is satisfied when the logical qubit is in the state $\ket{1}$, and not $\ket{0}$.} Aside from these projectors, one also has to define

\begin{equation}\label{eqn:clockprojs}
    \begin{aligned}
        P_{start}                                                                                   & := \ket{r} \! \bra{r}_0,                                                                                  \\
        P_{stop}                                                                                    & := \ket{d} \! \bra{d}_L,                                                                                  \\
        P_{clock}^{(i)} := \ket{r} \! \bra{r}_i \otimes (I - \ket{r} \! \bra{r})_{i+1} + \ket{a} \! & \bra{a}_i \otimes (I - \ket{r} \! \bra{r})_{i+1} + \ket{d} \! \bra{d}_i \otimes \ket{r} \! \bra{r}_{i+1},
    \end{aligned}
\end{equation}
\noindent
which are at most $4$-local projectors requiring that the clock states have the form described in \cref{eqn:clock}. Furthermore, these six types of projectors are of the form given in \cref{defn:P} and are hence projectors from $\mathcal{P}$ as required. Finally, using the six types of projectors of \cref{eqn:projs} and \cref{eqn:clockprojs}, the instance that encodes the verifier circuit $U=U_L \ldots U_1$ is given by

\begin{equation} \label{eqn:hamils}
    \begin{aligned}
        H_{init}  & := \sum_{b \in \textit{ancilla}} P_{init}^{(b)},          \\
        H_{prop}  & := \sum_{t=1}^L P_{prop,U_t}^{(t)}                        \\
        H_{out}   & := P_{out}^{(ans)},                                       \\
        H_{clock} & := P_{start} + P_{stop} + \sum_{c \in C} P_{clock}^{(c)}.
    \end{aligned}
\end{equation}
\noindent
We illustrate this instance in \cref{fig:reductioninstance} (for simplicity, we assume the data register consists only of four qubits). The set of projectors that define this $k$-QSAT instance are the individual terms of the sum; however, we often group them into positive semi-definite terms resembling those of the Local Hamiltonian problem \cite{kitaev2002classical}. In \cref{eqn:hamils}, the $H_{init}$ term requires that all ancilla qubits from register $\mathcal{D}$ are initialized to $\ket{0}$, leaving the data qubits for the witness state ``free'' or un-initialized. $H_{prop}$ defines a clock register of $L+1$ particles and requires that as time progresses from $t-1$ to $t$, $U_t$ is applied to the data qubits. $H_{out}$ requires that at the end of the computation $ans$ is measured to be ``1''. Finally, $H_{clock}$ requires that we obtain a running clock register and that the clock progresses as shown in \cref{eqn:clock}. Together, the terms $H_{init}, H_{prop}$, and $H_{clock}$ require that if there exists a state satisfying all of their projectors, the state must mimic the evaluation of the quantum circuit $U=U_L \ldots U_1$ on the state $\ket{\phi_0}$. This is the history state of \cref{eqn:hist} with the clock encoding of \cref{eqn:clock}. Moreover, if the verification circuit $U$ accepts with certainty, the history state also satisfies $H_{out}$ and is thus the unique ground state of the $6$-local Hamiltonian $H = H_{init} + H_{prop} + H_{out} + H_{clock}$.

This concludes the transformation of the circuit into local Hamiltonians. Completing the proof that $6$-QSAT is $\QMA_1$-hard requires showing that if $x \in A_{yes}$ then $x'$ has a frustration-free ground state, while if $x \in A_{no}$ then the ground state energy of $H$ is not too low. The former case is straightforward since the history state of \cref{eqn:hist} is, by construction, the state that satisfies all clauses of the instance. On the other hand, the latter case is significantly more challenging as it requires relating the null spaces of the four non-commuting operators $H_{init}, H_{prop}, H_{out}$, and $H_{clock}$. The key to accomplish this is through Kitaev's Geometric lemma:

\begin{lemma}[Geometric lemma \cite{kitaev2002classical}] \label{lemma:geometric}
    Let $H_1 \succeq 0$ and $H_2 \succeq 0$ be two positive semi-definite operators with null spaces $\mathcal{S}_1$ and $\mathcal{S}_2$, respectively. Suppose the null spaces have trivial intersection, i.e.\ $\mathcal{S}_1 \cap \mathcal{S}_2 = \{ \mathbf{0} \}$. Let $\gamma(H)$ denote the smallest non-zero eigenvalue of $H$. Then,

    \begin{equation}\label{eqn:geometric}
        H_1 + H_2 \succeq \min\{\gamma(H_1), \gamma(H_2)\} \cdot \frac{(1 - \alpha)}{2}
    \end{equation}
    \noindent
    where
    \begin{equation}
        \alpha = \max_{\ket{\eta} \in \mathcal{S}_1} \bra{\eta} \Pi_{\mathcal{S}_2} \ket{\eta}.
    \end{equation}
\end{lemma}
\noindent
We do not present the rest of the hardness proof here as it is beyond the scope of this preliminary section, and a similar analysis is presented in \cref{subsection:hardnessmono}.

\begin{figure}[t]
    \centering
    \includegraphics[width=0.95\textwidth]{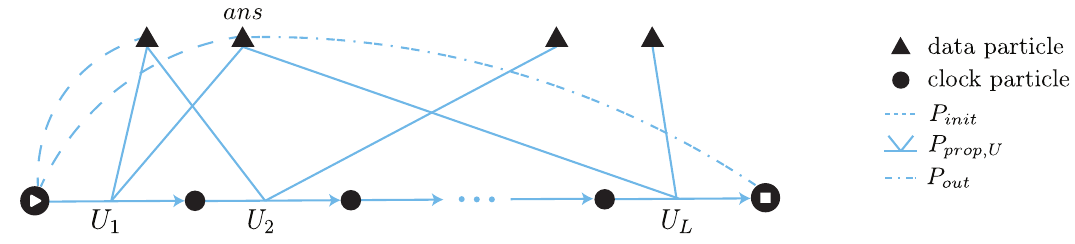}
    \caption{Representation of a $k$-QSAT instance which encodes a $\QMA_1$ verification circuit $U = U_L \ldots U_1$. For simplicity, we let $U$ act on four data qubits: two ancilla qubits and two qubits for the witness state. The ancillas are those present in $P_{init}$ clauses, and the witness state qubits those that are un-initialized. The ancilla measured at the end of the computation is labeled $ans$. The leftmost clock particle has a ``start'' icon to represent the action of the $P_{start}$ clause, and the rightmost a ``stop'' icon for the $P_{stop}$ clause. The $P_{clock}$ clauses are represented by the arrows on top of the $P_{prop,U}$ lines that connect clock particles together. It is represented by an arrow since this clause establishes the clock progression.}
    \label{fig:reductioninstance}
\end{figure}
\indent

\section{A \texorpdfstring{BQP\textsubscript{1}}{BQP\_1}-complete Problem Leveraging Monogamy} \label{section:monogamy}

In this section, we prove \cref{thm:bqplarge}, which states that the quantum satisfiability problem \textsc{Linear-Clock-Ternary-QSAT} is $\BQP_1^{\mathcal{G}_8}$-complete. In particular, in \cref{subsection:overview}, we provide an overview of how we use the circuit-to-Hamiltonian transformation introduced in \cref{subsubsection:qmahard} to engineer a satisfiability problem that is in $\BQP_1^{\mathcal{G}_8}$ and $\BQP_1^{\mathcal{G}_8}$-hard. In \cref{subsection:definition}, we present the detailed definition of the problem and formally begin the proof of its containment. In particular, we describe the structure of input instances and introduce terminology useful for later subsections. \cref{subsection:monosat} continues with an in-depth analysis on the satisfiability of instances, categorizing each one as trivially unsatisfiable, trivially satisfiable, or one requiring the assistance of a quantum algorithm. Then, based on the previous subsection, \cref{subsection:algorithmmono} presents the hybrid quantum-classical algorithm that decides the satisfiability of all instances along with an analysis of its correctness. This concludes that the problem is contained in $\BQP_1^{\mathcal{G}_8}$. Lastly, \cref{subsection:hardnessmono} shows that the problem is $\BQP_1^{\mathcal{G}_8}$-hard. From now on, we will omit the superscript and simply write $\BQP_1$ and $\QMA_1$.

\subsection{Construction overview} \label{subsection:overview}

\begin{figure}[t]
    \centering
    \begin{subfigure}{.5\textwidth}
        \centering
        \includegraphics[width=1\linewidth]{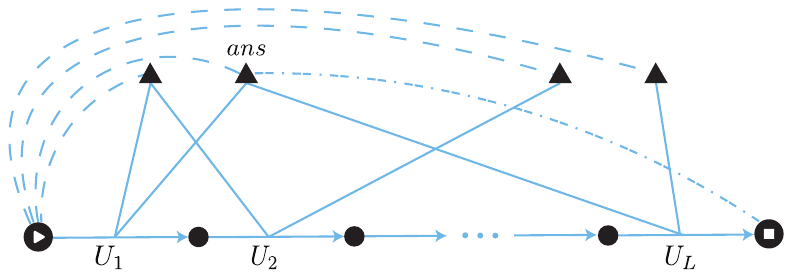}
        \caption{}
        \label{fig:qinstance}
    \end{subfigure}
    \vspace{\baselineskip}
    \begin{subfigure}{1\textwidth}
        \centering
        \includegraphics[width=1\linewidth]{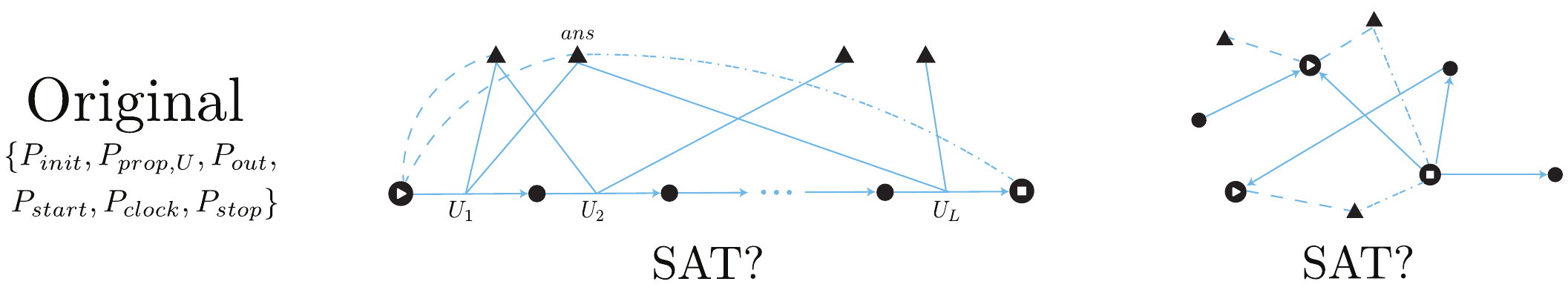}
        \caption{}
        \label{fig:troublesome}
    \end{subfigure}
    \vspace{\baselineskip}
    \begin{subfigure}{1\textwidth}
        \centering
        \includegraphics[width=1\linewidth]{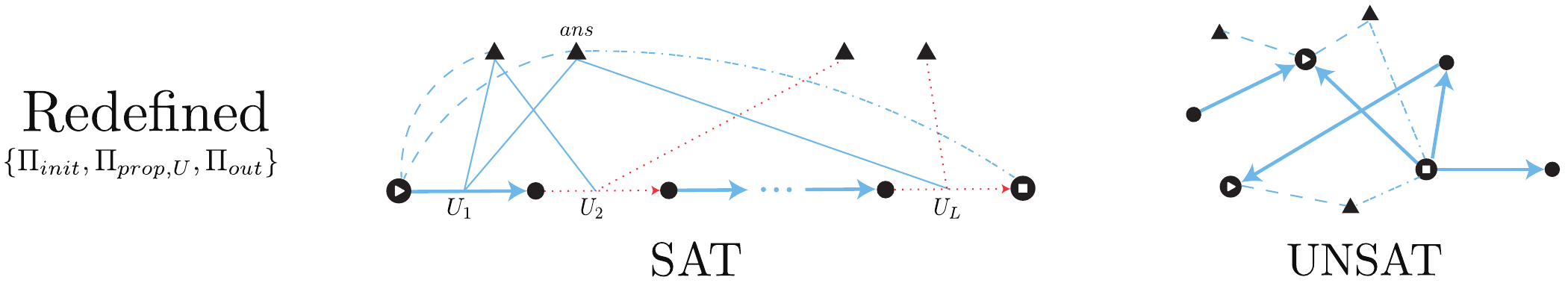}
        \caption{}
        \label{fig:troublesome2}
    \end{subfigure}
    \caption{(a) Example of a typical instance that encodes the computation of a $\BQP_1$ circuit $U_L \ldots U_1$. The satisfiability of the instance can also be decided in $\BQP_1$. (b) Examples of the troublesome instances whose satisfiability is not known to be decidable with a $\BQP_1$ algorithm. (c) The above instances recast with the new set of projectors of \cref{eqn:Oinit,eqn:Oprop,eqn:Oout}. The bold blue arrows represent the $\Pi_{prop,U}$ clauses which now also indicate the particles should be maximally entangled, and the dotted red arrows those that are connected to undefined logical qudits. With these projectors, their satisfiability can be more easily decided. The instance on the left is satisfiable due to the undefined clauses, while the one on the right is unsatisfiable as any potential satisfying state would violate monogamy of entanglement. The instance in (a) has the same meaning/satisfiability with either set of projectors.}
    \label{}
\end{figure}

The goal of the construction is to design a QSAT problem that can encode the computation of any quantum circuit in $\BQP_1$, while also being able to solve any of its instances in quantum polynomial-time with perfect completeness and bounded soundness. We define the problem using projectors $\Pi_{init}$, $\Pi_{prop,U}$, and $\Pi_{out}$ similar to $P_{init}$, $P_{prop,U}$, and $P_{out}$ defined in \cref{eqn:projs}.\footnote{The projectors $P_{start}$, $P_{clock}$, and $P_{stop}$ associated with the clock encoding remain unchanged and are integrated into the definitions of $\Pi_{init}$, $\Pi_{prop}$, and $\Pi_{out}$.} To see why our projectors must differ from the original ones, consider the QSAT problem built with $\{P_{init}, P_{prop,U}, P_{out}, P_{start}, P_{clock}, P_{end}\}$. Showing that the problem is $\BQP_1$-hard is straightforward, as we can encode the circuit that computes the answer to a $\BQP_1$ problem in a similar way as that shown in \cref{subsection:circtoh}. This time however, all data particles in the instance should be initialized, preventing having free particles that can accommodate a witness state (see \cref{fig:qinstance}). The difficulty lies in demonstrating that every instance generated with a polynomial number of these projectors can also be decided in $\BQP_1$. There is a fundamental and a practical limitation for this:

\begin{itemize}
    \item[--] Instances which encode the computation of a $\QMA_1$ problem, e.g.\ the instance in \cref{fig:reductioninstance} and the left instance in \cref{fig:troublesome}, are valid inputs. This is problematic since it is unknown how to decide these instances in $\BQP_1$ (and doing so would show that $\BQP_1=\QMA_1$).
    \item[--] Input instances may form intricate structures complicating the task of deciding if a satisfying state exists, e.g.\ the right instance in \cref{fig:troublesome}.
\end{itemize}
\noindent
We define the projectors $\Pi_{init}$, $\Pi_{prop}$, and $\Pi_{out}$ to address these two difficulties (see \cref{fig:troublesome2}). Importantly, these projectors do not significantly alter the proof that the problem is $\BQP_1$-hard and can proceed as mentioned. Now, let us briefly discuss how we overcome both difficulties. \\
\indent
Instances like those in \cref{fig:reductioninstance}, which have a proper structure and uninitialized data particles, are prototypical examples of $\QMA$ instances. These ``free'' particles give one the freedom to guess if there exists a state they can be in such that the instance can be satisfied (or equivalently be provided with such a state which we verify). To address this issue, we remove the need to guess a satisfying state by introducing a new \textit{undefined} basis state $\ket{?}$ (making the data particles 3-dimensional), such that setting the free data particles to this state always results in a satisfiable instance. More specifically, we achieve this by defining $\Pi_{prop,U}$ so that if any data particle in the clause is in state $\ket{?}$, the clause is satisfied without needing to apply the associated unitary.\footnote{This shows that although the data particles are $3$-dimensional and the unitaries are gates from a set designed to act on qubits, there is no issue because the gates will never act on undefined data particles.} Then, for these instances, the satisfying state is given by a truncated version of the history state (without a witness) since the computation is no longer required to elapse past the first $\Pi_{prop,U}$ clause acting on an undefined state. We say the instance is now ``trivially satisfiable'' as its structure alone suffices to determine its satisfiability. The formal statement of the claims presented here is given in \cref{lemma:truncated}. \\
\indent
To determine the satisfiability of intricate instances, the projectors are now also defined to leverage the principle of monogamy of entanglement. Each clock particle is equipped with two $2$-dimensional auxiliary subspaces $C\!A$ and $C\!B$ (making them $12$-dimensional) and the $\Pi_{prop,U}$ clauses are then defined to require that the $C\!B$ subspace of the predecessor clock particle forms a $\ket{\Phi^+}$ Bell pair with the $C\!A$ subspace of its successor. Then, if a $C\!A$ or $C\!B$ subspace is required to form more than one Bell pair, the principle of monogamy of entanglement states that only one of these clauses can be satisfied, and so the instance is unsatisfiable. Therefore, instances that are not deemed unsatisfiable because of this reason must form one-dimensional chains with a unique ``time'' direction. Finally, to guarantee that $\Pi_{init}$ and $\Pi_{out}$ only act on the ends of the chain, these make use of a new \textit{endpoint} particle consisting of a single two-dimensional space $EC$ and require that it also forms a Bell pair with either the $C\!A$ (for $\Pi_{init}$) or $C\!B$ (for $\Pi_{out}$) subspace of a clock particle. \cref{fig:monogamydemo} provides a visual summary on how these projectors use monogamy to form one-dimensional instances. \\
\indent
Although these modifications do not get rid off all difficulties, they are enough to determine the satisfiability of all input instances via a hybrid algorithm. Briefly, the classical part of the algorithm evaluates the structure of the clauses in the instance and concludes whether it is trivially unsatisfiable, trivially satisfiable, or is one requiring the assistance of a quantum subroutine. Trivially unsatisfiable instances are those whose clause arrangement imply one or several clauses cannot be simultaneously satisfied, like those that violate monogamy of entanglement. On the other hand, trivially satisfiable instances are those whose clauses do not create any conflicts but whose structure is simple enough that the satisfying state can be inferred, like those with proper structure and uninitialized data particles. We show that the only type of instances that are not in either one of these cases, are those like \cref{fig:qinstance} which express the computation of a quantum circuit on initialized ancilla qubits. For these instances, the classical algorithm makes use of a quantum subroutine that executes the quantum circuit expressed by the instance, while simultaneously measuring the eigenvalues of relevant projectors. The measurement outcomes indicate whether the instance should be accepted or rejected.

\begin{figure}[t]
    \centering
    \includegraphics[width=0.9\textwidth]{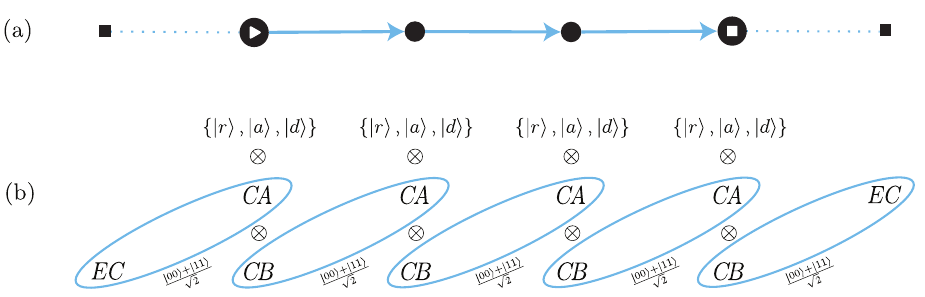}
    \caption{(a) A simplified graphical representation of a one-dimensional instance built with projectors $\Pi_{init}$, $\Pi_{prop,U}$, and $\Pi_{out}$ analogous to previous figures. The data particles are not shown here. The endpoint particles are represented by the black squares, and its connections to clock particles via $\Pi_{init}$ and $\Pi_{out}$ clauses are given by the dotted blue lines. (b) A more descriptive and accurate representation of the figure above where the Hilbert space of both clock and endpoint particles is explicit. The blue ellipses signify that the encircled subspaces are required to be maximally entangled. Each clock particle has a free subspace spanned by $\ket{r}$, $\ket{a}$, and $\ket{d}$.}
    \label{fig:monogamydemo}
\end{figure}

\subsection{Problem definition} \label{subsection:definition}
The problem we claim is $\BQP_1$-complete is the following:

\begin{defn}[\textsc{Linear-Clock-Ternary-QSAT}] \label{defn:lctqsat}
    The problem {\rm \textsc{Linear-Clock-Ternary-QSAT}} is a quantum constraint satisfaction problem defined on the $17$-dimensional Hilbert space

    \begin{equation*}
        \begin{aligned}
            \mathcal{H} = \; \; & \Span\{\ket{0_L}, \ket{1_L}, \ket{?_L}\} \oplus \Span\{\ket{0_{EC}}, \ket{1_{EC}}\}                                                               \\
                                & \oplus (\Span\{\ket{r_C},\ket{a_C}, \ket{d_C}\} \otimes \Span\{\ket{0_{C\!A}}, \ket{1_{C\!A}}\} \otimes \Span\{\ket{0_{C\!B}}, \ket{1_{C\!B}}\}),
        \end{aligned}
    \end{equation*}
    \noindent
    consisting of a logical, endpoint, and clock subspaces. The problem consists of $5$ types of projectors: $\Pi_{init}$, $\Pi_{prop,U}$ (one for every $U \in \{H,HT, (H \otimes H) \textnormal{CNOT}\}$), and $\Pi_{out}$, each acting on at most four $17$-dimensional qudits. Since the definitions of these projectors are too long to fit here, we define them implicitly as the projectors with the same null space as their corresponding positive semi-definite operator:

    \begin{equation} \label{eqn:Oinit}
        \begin{aligned}
            O_{init} :=    \; \; & \Pi_{L,1} + \Pi_{C,2} + \Pi_{E,3} +  \Pi_{start, 2} + (I - \ket{0} \! \bra{0})_1 \otimes \ket{a} \! \bra{a}_2 \\ &+ (I - \ket{\Phi^+} \! \bra{\Phi^+}_{C\!A,EC})_{2,3},
        \end{aligned}
    \end{equation}

    \begin{equation} \label{eqn:Oprop}
        \begin{aligned}
            O_{prop,U} := \; \; & \Pi_{L,1} + \Pi_{L,2} + \Pi_{C,3} + \Pi_{C,4}                                                      \\
                                & +  (\Pi_{work,U} + I^{\otimes 2} \otimes \Pi_{clock,D})(\Pi_D \otimes \Pi_D \otimes I^{\otimes 2}) \\ &+ (I^{\otimes 2} \otimes \Pi_{clock,?})(I - \Pi_D \otimes \Pi_D \otimes I^{\otimes 2}) \\
                                & + (I - \ket{\Phi^+} \! \bra{\Phi^+}_{C\!A,C\!B})_{3,4},
        \end{aligned}
    \end{equation}
    \noindent
    and

    \begin{equation} \label{eqn:Oout}
        \begin{aligned}
            O_{out} :=    \; \; & \Pi_{L,1} + \Pi_{C,2} + \Pi_{E,3} + \Pi_{stop,2} + \ket{0} \! \bra{0}_1 \otimes \ket{a} \! \bra{a}_2 \\ &+ (I - \ket{\Phi^+} \! \bra{\Phi^+}_{C\!B,EC})_{2,3}.
        \end{aligned}
    \end{equation}
    \indent
    Here, we also use a slight abuse of notation to denote operators like $0^5 \oplus (\ket{a} \! \bra{a} \otimes I^4)$ simply as $\ket{a} \! \bra{a}$. Additionally, $O_{init}$, $O_{prop,U}$, and $O_{out}$ are composed of the $4$-local projector

    \begin{equation*}
        \Pi_{work,U} := \frac{1}{2} \left[I^{\otimes2} \otimes \ket{ar} \! \bra{ar} + I^{\otimes 2} \otimes \ket{da} \! \bra{da} - U \otimes \ket{da} \! \bra{ar} - U^\dagger \otimes \ket{ar} \! \bra{da}\right],
    \end{equation*}
    \noindent
    and other $1$- and $2$-local projectors: the role-assigning projectors
    \begin{equation*}
        \begin{aligned}
            \Pi_C & := I - (\ket{r}\!\bra{r}_C + \ket{a}\!\bra{a}_C + \ket{d}\!\bra{d}_C),            \\
            \Pi_E & := I - (\ket{0} \! \bra{0}_{EC} + \ket{1} \! \bra{1}_{EC}),                       \\
            \Pi_L & := I - (\ket{0}\! \bra{0}_{L} + \ket{1} \! \bra{1}_{L} + \ket{?} \! \bra{?}_{L}),
        \end{aligned}
    \end{equation*}
    \noindent
    the clock start/stop projectors

    \begin{equation*}
        \begin{aligned}
            \Pi_{start} & := \ket{r} \! \bra{r}, \\
            \Pi_{stop}  & := \ket{d} \! \bra{d},
        \end{aligned}
    \end{equation*}
    \noindent
    the projector onto the data subspace
    \begin{equation*}
        \Pi_D := \ket{0} \! \bra{0}_{L} + \ket{1} \! \bra{1}_{L},
    \end{equation*}
    \noindent
    and the clock-keeping projectors

    \begin{equation*}
        \begin{aligned}
            \Pi_{clock,D} & := \ket{r} \! \bra{r} \otimes (I - \ket{r} \! \bra{r}) + \ket{a} \! \bra{a} \otimes (I - \ket{r} \! \bra{r}) + \ket{d} \! \bra{d} \otimes \ket{r} \! \bra{r}, \\
            \Pi_{clock,?} & := \ket{r} \! \bra{r} \otimes (I - \ket{r} \! \bra{r}) + \ket{a} \! \bra{a} \otimes (I - \ket{r} \! \bra{r}) + \ket{d} \! \bra{d} \otimes I.
        \end{aligned}
    \end{equation*}
    \noindent
    Evidently, $\Pi_{init}$ and $\Pi_{out}$ are $3$-local, while $\Pi_{prop,U}$ is $4$-local (on high-dimensional qudits). \\
    \indent
    There is also a promise. We are assured that for every instance considered either (1) there exists a state $\ket{\psi_{sat}}$ on $n$ $17$-dimensional qudits such that $\Pi_i \! \ket{\psi_{sat}} = 0$ for all $i$, or (2) $\Sigma_i \bra{\psi} \! \Pi_i \! \ket{\psi} \geq 1/poly(n)$ for all $\ket{\psi}$. \\
    \indent
    The goal is to output ``YES'' if (1) is true, or output ``NO'' otherwise.
\end{defn}
\noindent
This definition can be summarized as follows:

\begin{center}
    \fbox{
        \begin{minipage}{5.5 in}
            \textbf{Linear-Clock-Ternary-QSAT}\\
            \noindent
            \begin{tabular}{l p{4.7 in}}
                \textit{Input:}   & An integer $n$, and a set of projectors $\{\Pi_i\} \subseteq \{\Pi_{\textnormal{init}}, \Pi_{\textnormal{out}}, \Pi_{\textnormal{prop},\mathcal{G}}\}^{poly(n)}$, where $\mathcal{G} \in \{ H, HT, (H \otimes H) \textnormal{CNOT}\}$ and each projector acts nontrivially on at most four $17$-dimensional qudits. \\
                \textit{Promise:} & Either (1) there exists an $n$-qudit state $\ket{\psi_{sat}}$ such that $\Pi_i \! \ket{\psi_{sat}} = 0$ for all $i$, or (2) $\Sigma_i \bra{\psi} \! \Pi_i \! \ket{\psi} \geq 1/poly(n)$ for all $\ket{\psi}$.                                                                                                       \\
                \textit{Goal:}    & Output ``YES'' if (1) is true, or output ``NO'' otherwise.
            \end{tabular}
        \end{minipage}
    }
\end{center}

Let us note a few things about \cref{defn:bqp1}. First, observe that each of the individual projectors that compose $O_{init}$, $O_{prop,U}$, and $O_{out}$ belongs to the set $\mathcal{P}$ of \cref{defn:P}. As we will discuss in \cref{subsection:algorithmmono}, this suffices for the quantum algorithm and it is not necessary to demonstrate that $\Pi_{init}$, $\Pi_{prop,U}$, and $\Pi_{out}$ are also elements from this set. Second, the projectors $\Pi_{start}$, $\Pi_{clock}$, and $\Pi_{stop}$ that in \cref{subsubsection:qmahard} were independent, have now been incorporated into the definitions of the projectors $\Pi_{init}$, $\Pi_{prop}$ and $\Pi_{out}$ ($\Pi_{clock}$ has been renamed to $\Pi_{clock,D}$). Additionally, we include a new term $\Pi_{clock,?}$ to restrict the allowed states of the clock register whenever a logical qudit is undefined. Lastly, observe that the unitaries $U$ within the $\Pi_{prop,U}$ clauses are not exactly from the Clifford+T gate set like in \cref{subsubsection:qmahard}, but from a slight variation where $T$ and CNOT are also accompanied by a Hadamard gate. We use this set for its property that each unitary necessarily changes computational basis states. We remark that while the potentially satisfiable instances encode circuits using this set, the algorithm that decides these instances by executing the circuits can be actually thought of as using the Clifford+T gate set. Thus, our claim that the problem is in $\BQP_1$ with the Clifford+T gate set is accurate.\\
\indent
Now, let us analyze these projectors in greater detail and determine which states lie in their null spaces. \\

\subsubsection{Initialization and termination}

The $\Pi_{init}$ term of \cref{eqn:Oinit} is a projector acting on three qudits. In the first line, the $\Pi_L, \Pi_C$ and $\Pi_E$ projectors demand that these qudits serve the roles of logical, clock, and endpoint qudits respectively. Then, the term $\Pi_{start}$ demands that the clock qudit cannot be $\ket{r}$. Finally, the last term of the line, $(I - \ket{0} \! \bra{0})_1 \otimes \ket{a}\!\bra{a}_2$, corresponds to the initialization of the logical qudit (similar to $P_{init}$ in \cref{eqn:projs}), requiring that when the clock qudit is $\ket{a}$, the logical qudit is $\ket{0}$. The projector in the second line demands that the $EC$ subspace of the endpoint qudit forms a $\ket{\Phi^+}$ Bell pair with the $C\!A$ endpoint of the clock qudit. Considering the demands of all these projectors, one can show that the states

\begin{equation} \label{eqn:initsat}
    \begin{tikzpicture}[baseline=(current  bounding  box.center)]
        % First part of the equation
        \node (A) at (0,-0.8) {$\ket{\textnormal{x}}$};

        \node (otimes0) at (0.5,-0.8) {$\otimes$};

        % Stacked part of the equation
        \node (ket_r) at (1.3,0) {$\ket{d}$};
        \node (otimes1) at (1.3,-0.4) {$\otimes$};
        \node (CA) at (1.3,-0.8) {$C\!A$};
        \node (otimes2) at (1.3,-1.2) {$\otimes$};
        \node (CB) at (1.3,-1.6) {$C\!B$};

        % Second part of the equation
        \node (otimes4) at (2.1,-0.8) {$\otimes$};

        \node (B) at (2.6,-0.8) {$EC$};

        \node (and) at (3.8,-0.8) {and};

        \node (A) at (5,-0.8) {$\ket{0}$};

        \node (otimes0) at (5.5,-0.8) {$\otimes$};

        % Stacked part of the equation
        \node (ket_r) at (6.3,0) {$\ket{a}$};
        \node (otimes1) at (6.3,-0.4) {$\otimes$};
        \node (CA) at (6.3,-0.8) {$C\!A$};
        \node (otimes2) at (6.3,-1.2) {$\otimes$};
        \node (CB) at (6.3,-1.6) {$C\!B$};

        % Second part of the equation
        \node (otimes4) at (7.1, -0.8) {$\otimes$};

        \node (B) at (7.6,-0.8) {$EC$};

        \draw[decorate,decoration={brace,amplitude=5pt,mirror,raise=5pt},thick] (1.2,0.2) -- (1.2,-1.8) node[midway,xshift=-10pt] {};
        \draw[decorate,decoration={brace,amplitude=5pt,raise=5pt},thick] (1.4,0.2) -- (1.4,-1.8) node[midway,xshift=10pt] {};

        \draw[decorate,decoration={brace,amplitude=5pt,mirror,raise=5pt},thick] (6.2,0.2) -- (6.2,-1.8) node[midway,xshift=-10pt] {};
        \draw[decorate,decoration={brace,amplitude=5pt,raise=5pt},thick] (6.4,0.2) -- (6.4,-1.8) node[midway,xshift=10pt] {};

        \draw[myblue, thick, rotate around={0:(1.95,-0.8)}] (1.95,-0.8) ellipse (1.2cm and 0.35cm);
        \draw[myblue, thick, rotate around={0:(6.95,-0.8)}] (6.95,-0.8) ellipse (1.2cm and 0.35cm);
    \end{tikzpicture}
\end{equation}

\noindent
are the only states that satisfy all clauses of $\Pi_{init}$. Here, $\textnormal{x} \in \{0,1,?\}$ and the blue ellipse represents that the subspaces are maximally entangled. Observe that in the first state, the state of the clock qudit alone suffices to satisfy the clause and the constraint on the logical qudit is not enforced. \\
\indent
$\Pi_{out}$ is defined similarly as $\Pi_{init}$, except that it incorporates a projector similar to $P_{out}$ instead of $P_{init}$ and the $C\!A$ and $C\!B$ subspaces swap roles. One can show that the satisfying states are

\begin{equation}\label{eqn:outsat}
    \begin{tikzpicture}[baseline=(current  bounding  box.center)]
        % First part of the equation
        \node (A) at (0,-0.8) {$\ket{\textnormal{x}}$ };

        \node(otimes0) at (0.5, -0.8) {$\otimes$};

        % Stacked part of the equation
        \node (ket_r) at (1.3,0) {$\ket{r}$};
        \node (otimes1) at (1.3,-0.4) {$\otimes$};
        \node (CA) at (1.3,-0.8) {$C\!A$};
        \node (otimes2) at (1.3,-1.2) {$\otimes$};
        \node (CB) at (1.3,-1.6) {$C\!B$};

        % Second part of the equation
        \node (otimes4) at (2.1,-0.8) {$\otimes$};

        \node (B) at (2.6,-0.8) {$EC$};

        \node (and) at (3.8,-0.8) {and};

        \node (A) at (5,-0.8) {$\ket{\textnormal{y}}$};

        \node(otimes0) at (5.5, -0.8) {$\otimes$};

        % Stacked part of the equation
        \node (ket_r) at (6.3,0) {$\ket{a}$};
        \node (otimes1) at (6.3,-0.4) {$\otimes$};
        \node (CA) at (6.3,-0.8) {$C\!A$};
        \node (otimes2) at (6.3,-1.2) {$\otimes$};
        \node (CB) at (6.3,-1.6) {$C\!B$};

        \node (otimes4) at (7.1,-0.8) {$\otimes$};

        % Second part of the equation
        \node (B) at (7.6,-0.8) {$EC$,};

        \draw[decorate,decoration={brace,amplitude=5pt,mirror,raise=5pt},thick] (1.2,0.2) -- (1.2,-1.8) node[midway,xshift=-10pt] {};
        \draw[decorate,decoration={brace,amplitude=5pt,raise=5pt},thick] (1.4,0.2) -- (1.4,-1.8) node[midway,xshift=10pt] {};

        \draw[decorate,decoration={brace,amplitude=5pt,mirror,raise=5pt},thick] (6.2,0.2) -- (6.2,-1.8) node[midway,xshift=-10pt] {};
        \draw[decorate,decoration={brace,amplitude=5pt,raise=5pt},thick] (6.4,0.2) -- (6.4,-1.8) node[midway,xshift=10pt] {};

        \draw[myblue, thick, rotate around={29:(1.5,-1.4)}] (1.95,-1.4) ellipse (1.2cm and 0.4cm);
        \draw[myblue, thick, rotate around={29:(6.5,-1.4)}] (6.95,-1.4) ellipse (1.2cm and 0.4cm);

    \end{tikzpicture}
\end{equation}

\noindent
for $\textnormal{x} \in \{0,1,?\}$ and $\textnormal{y} \in \{1,?\}$. Note that in this construction, we allow the $\Pi_{out}$ clause to be satisfied when the logical qudit is in the state $\ket{1}$ or $\ket{?}$. \\
\indent
As we will see shortly, $\Pi_{init}$ and $\Pi_{out}$ will still serve the same primary role as in the proof that $k$-QSAT is hard for $\QMA_1$: to initialize data qudits at the start of the computation, and to verify the state of data qudits at its conclusion. In addition, the $\Pi_{start}$ and $\Pi_{stop}$ projectors give $\Pi_{init}$ and $\Pi_{out}$ another purpose: to obtain a running clock. We will observe that when both $\Pi_{start}$ and $\Pi_{stop}$ are present in an instance, the state where all clock qudits are either $\ket{r}$ or $\ket{d}$ cannot be a satisfying state. \\
\indent
Writing the satisfying states of clauses as in \cref{eqn:initsat} and \cref{eqn:outsat} will become cumbersome when considering instances with multiple clauses and satisfying states that are superpositions of qudits. For this reason, from now on, we only write the state of the logical qudit and the primary subspace of the clock qudit---the one that actually represents a state of the clock---and forget about the auxiliary subspaces. We can do so because satisfying clauses does not require mixing the auxiliary subspaces with the actual clock and logical subspaces, and any state whose auxiliary subspaces are not of the form shown in \cref{fig:monogamydemo} must violate one of the $\Pi_{init}$, $\Pi_{prop}$ or $\Pi_{out}$ clauses.

\subsubsection{Propagation and clock} \label{subsubsection:prop}

The first line of $\Pi_{prop,U}$ demands that two of the four multi-purpose qudits serve as logical qudits and the other two as clock qudits. The fourth line demands that both clock qudits form a Bell pair in the $C\!B$ subspace of the predecessor and the $C\!A$ subspace of the successor. The second and third line express the conditions for propagation, where each line specifies different requirements depending on the state of the logical qudits.\\
\indent
First, suppose that none of logical qudits are undefined, i.e.\ they are in the joint state $\ket{\phi}$ such that $(\bra{?} \otimes I)\ket{\phi} = (I \otimes \bra{?}) \ket{\phi} = 0$. Under these conditions, the term $(\Pi_D \otimes \Pi_D \otimes I^{\otimes 2})$ in the second line is not satisfied, while $(I - \Pi_D \otimes \Pi_D \otimes I^{\otimes 2})$ in the third line is. Then, to satisfy all four lines of the $\Pi_{prop,U}$ clause, the qudits must be in a state that lies in the null space of the term $(\Pi_{work,U} + I^{\otimes 2} \otimes \Pi_{clock,D})$. Since $\Pi_{work,U}$ is identical to $P_{prop,U}$ of \cref{eqn:projs}, this term requires that there is a usual propagation of the computation, while maintaining a correct form of the clock qudits. The only states that satisfy all terms are

\begin{equation} \label{eqn:propsat}
    \ket{\phi} \otimes \ket{r r} \textnormal{, \; \; \;} \ket{\phi} \otimes \ket{d d}  \textnormal{\; \; \; and \; \; \;} \frac{\ket{\phi} \otimes \ket{a r} + U \! \ket{\phi} \otimes \ket{d a}}{\sqrt{2}}.
\end{equation}
\indent
On the other hand, if one of the logical qudits in the clause is undefined, we obtain the opposite behavior: $(\Pi_D \otimes \Pi_D \otimes I^{\otimes 2})$ is satisfied, while $(I - \Pi_D \otimes \Pi_D \otimes I^{\otimes 2})$ is not. Then, the satisfying state of the clause is that which lies in the null space of $(I^{\otimes 2} \otimes \Pi_{clock,?})$, where $\Pi_{clock,?}$ is similar to $\Pi_{clock,D}$ except that the successor clock qudit must be $\ket{r}$. Then, the only states that satisfy all terms are

\begin{equation}\label{eqn:nullspacehprop}
    \{\ket{? 0},\ket{? 1},\ket{0 ?},\ket{1 ?},\ket{? ?}\} \otimes \{\ket{rr},\ket{ar}\}.
\end{equation}
\noindent
In contrast to the states that satisfy the clauses with well-defined logical qudits, the satisfying state of clauses with undefined logical qudits consists only of the term previous to the application of the unitary. For ease, we will refer to the clauses with these two different types of logical qudit values as \textit{well-defined} and \textit{undefined} clauses. See also \cref{fig:troublesome}, where we draw undefined $\Pi_{prop,U}$ clauses as red dotted arrows. \\
\indent
Finally, observe that the state where the clock qudits are set to $\ket{rr}$ is enough to satisfy the $\Pi_{prop,U}$ clause independently from the state of the logical qudits.

\subsubsection{Instances}

As will be discussed in more detail in the following subsections, showing that LCT-QSAT can be solved in $\BQP_1$ requires us to determine the satisfiability of any possible instance created using a polynomial amount of $\Pi_{init}, \Pi_{prop,U},$ and $\Pi_{out}$ clauses. Let us briefly present the types of instances that a collection of such clauses may form, and establish several definitions and notation useful for the rest of the text. The analysis of the satisfiability of these instances will, for the most part, be postponed until the next subsection. We note that from now on, we will refer to $\Pi_{prop,U}$ simply as $\Pi_{prop}$, as we will generally not be concerned with the associated unitary.\\
\indent
The most general instance we can receive as input is one with no structure whatsoever: one where the qudits of the instance are required to serve multiple roles. In other words, the qudits of the instance may be acted on by multiple $\Pi_C$, $\Pi_E$, and $\Pi_L$ clauses. While proved formally in \cref{lemma:singletypequditmono}, these clauses cannot be satisfied simultaneously as the subspaces are orthogonal to each other and so the instance should be rejected. For instances where all qudits serve a single role, the clock qudits and the $\Pi_{prop}$ clauses connecting them make an important criterion for their satisfiability. Focusing on these two elements only, an instance can then be thought of as a collection of disjoint directed sub-graphs. We refer to each one of these sub-graphs as a \textit{clock component}, and refer to both the clock component and its logical qudits as a \textit{sub-instance}. We now list some particular arrangements of the clauses within a clock component that are of interest and which affect the satisfiability of the instance. These are also illustrated in \cref{fig:instances}.

\begin{figure}[t]
    \centering
    \includegraphics[width=\textwidth]{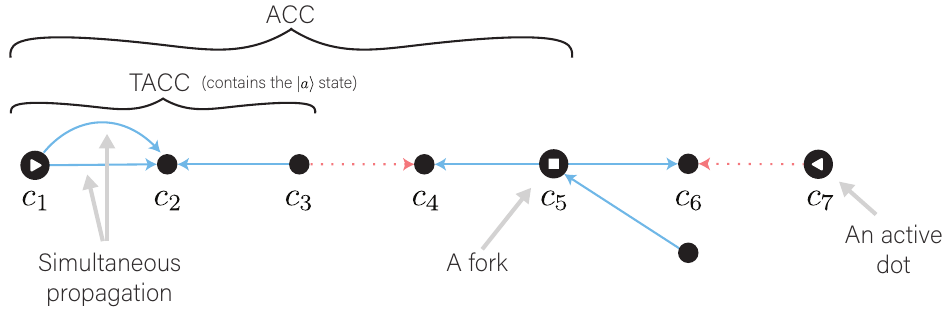}
    \caption{A clock component containing examples of relevant events that may occur in the clock component. The clock component has two active clock chains: one given by $c_1$ to $c_5$, and the other by $c_7,c_6,c_5$. The first chain is truncated due to an undefined clause, and hence $c_1$ to $c_3$ forms a TACC. Similarly, the second ACC is also truncated, making $c_7$ an active dot. For clarity, the endpoint qudits are not drawn.} \label{fig:instances}
\end{figure}

\begin{defn}[Forks]
    A clock component is said to have a \textup{fork} if there exists a clock qudit that is connected via $\Pi_{prop}$ clauses to more than two distinct clock qudits.
\end{defn}

\begin{defn}[Simultaneous propagation]
    A clock component is said to have \textup{simultaneous propagation} if there are two or more $\Pi_{prop}$ clauses acting on the same pair of clock qudits.
\end{defn}
\indent
In a clock component, pairs of $\Pi_{init}$ and $\Pi_{out}$ clauses and the clauses connecting them define a crucial part of the instance called the \textit{active clock chain}, abbreviated as ACC.

\begin{defn}[Active clock chain] \label{defn:acc}
    In a clock component, an \textup{active clock chain} is a set of clock qudits and $\Pi_{prop}$ clauses in the one-dimensional path connecting a clock qudit present in a $\Pi_{init}$ clause to a clock qudit present in a $\Pi_{out}$ clause, such that no other $\Pi_{init}$ or $\Pi_{out}$ clause acts on the clock qudits of this chain.
\end{defn}
\noindent
An active clock chain receives this name because setting the clock qudits in the chain to either one of the inactive states violates the $\Pi_{start}$ and $\Pi_{stop}$ projectors within the $\Pi_{init}$ and $\Pi_{out}$ clauses. Therefore, to satisfy the clauses of the chain, at least one of the clock qudits must be in an active state. However, the possible location of the active state in the satisfying state also depends on the existence of undefined clauses within the ACC. This leads us to the definition of \textit{truly active clock chains} (TACCs)---the part of the chain that necessarily contains an active state.

\begin{defn}[Truly active clock chain] \label{defn:tacc}
    A \textup{truly active clock chain} is a subset of clock qudits and $\Pi_{prop}$ clauses within an active clock chain consisting of the clock qudit present in the $\Pi_{init}$ clause up to (and including) the first clock qudit present in an undefined $\Pi_{prop}$ clause. If there is no such clause, the truly active clock chain is identical to the active clock chain.
\end{defn}
\noindent
We define the length $L$ of an ACC and TACC as the number of $\Pi_{prop}$ clauses within the chain. A chain of length $L$ then involves $L+1$ clock qudits. It is worth noting that some chains may also have $L = 0$ which may arise in an instance where both a $\Pi_{init}$ and $\Pi_{out}$ clause act on the same clock qudit, or when an undefined clause acts on the clock qudit with the starting $\Pi_{init}$ clause. We will call these chains \textit{active dots}. In this section, active dots and chains of non-zero length behave quite similarly, however, the distinction between the two becomes more relevant in the construction of \cref{section:our}. \\
\indent
For the final pieces of notation, we use $T$ instead of $L$ to the refer to the length of a TACC when it is a strict subset of an ACC. Finally, we let $q$ denote the number of logical qudits involved in clauses within a TACC. \\
\indent
Now, let us discuss the satisfiability of the possible input instances.

\subsection{Instance satisfiability} \label{subsection:monosat}

In this subsection we formally begin the proof that \textsc{Linear-Clock-Ternary-QSAT} is contained in $\BQP_1$. Here, we focus on determining the satisfiability conditions of all input instances, which will serve as the back-bone of the algorithm presented in the following subsection. We present our findings as a series of lemmas, and remark that we work under the assumption that the instance considered in a lemma has not been decided by any of the previous lemmas. For clarity and continuity, we omit the proofs of the lemmas here, and instead collect them in \cref{appendix:monogamy}. \\
\indent
In summary, we decide the satisfiability of an instance by evaluating all of its sub-instances and the logical qudits that they may share. An instance is unsatisfiable if at least one of its sub-instances is also unsatisfiable, and conversely, the instance is satisfiable iff all sub-instances are satisfiable. We will show that most sub-instances can be decided based on the arrangement of their clauses. However, there are three types of sub-instances which we do not know how to decide classically and hence make use of a quantum algorithm. The first type of sub-instance is that with a single one-dimensional TACC spanning the whole ACC, properly initialized logical qudits, and a coherent flow of time like that shown in \cref{fig:qinstance}. These are the sub-instances that can be identified with a quantum circuit acting on a ``data'' register and a ``clock'' register finalized by measurements of some of the data qubits, just like in previous uses of the circuit-to-Hamiltonian construction.\footnote{As we will also show at a later time, these instances can encode the computation of a $\BQP_1$ problem. Deciding them classically (either deterministically or probabilistically) would show that this model of classical computation is equivalent to quantum computation with perfect completeness.} The other types of instances are those with a single TACC and a coherent flow of time but contain simultaneous propagation clauses. There are two variations as the TACC may either span the whole ACC, or may be truncated due to undefined $\Pi_{prop}$ clauses (the undefined clauses may only occur after the first simultaneous propagation clause, as the sub-instance would be trivial otherwise). These three sub-instances are shown in \cref{fig:qinputs}.\\
\indent
To begin, when presented with a general instance, the first criterion regarding its satisfiability is that each qudit must serve a single role, e.g.\ data, clock, or endpoint:

\begin{restatable}[Single-type qudits]{lemma}{singletypequditmono} \label{lemma:singletypequditmono}
    If there is a $17$-dimensional qudit in the instance acted on by at least two projectors $\Pi_\alpha$ and $\Pi_{\alpha'}$, where $\alpha,\alpha' \in \{C,E,L\}$ and $\alpha \neq \alpha'$, the instance is unsatisfiable.
\end{restatable}
\noindent
Therefore, instances that may have a satisfying state must have qudits that serve a single role at all times. By construction, the next criterion for determining the satisfiability of instances comes from the monogamy of entanglement conditions between clock and endpoint qudits. To properly analyze this structure, we separate the instance into smaller sub-instances by considering the sets of connected clock qudits (the clock components mentioned previously).\footnote{Although endpoint qudits are also able to connect clock qudits together via $\Pi_{init}$ and $\Pi_{out}$ clauses, we do not consider such clock components at the moment to avoid possible confusion. Besides, as will be shown shortly in \cref{lemma:uniqueclockmono}, satisfying these clauses inevitably leads to a violation of monogamy of entanglement.} Clearly, partitioning an instance by considering only its clock qudits may not generate completely disjoint sub-instances. These could still share logical qudits among them. We proceed by first evaluating the satisfiability of individual sub-instances and only then evaluate how these relate to each other. This is reasonable because the satisfying state of a sub-instance determines which of its logical qudits are actually constrained and should therefore be considered ``shared''.

\begin{figure}[t]
    \centering
    \includegraphics[width=0.9\textwidth]{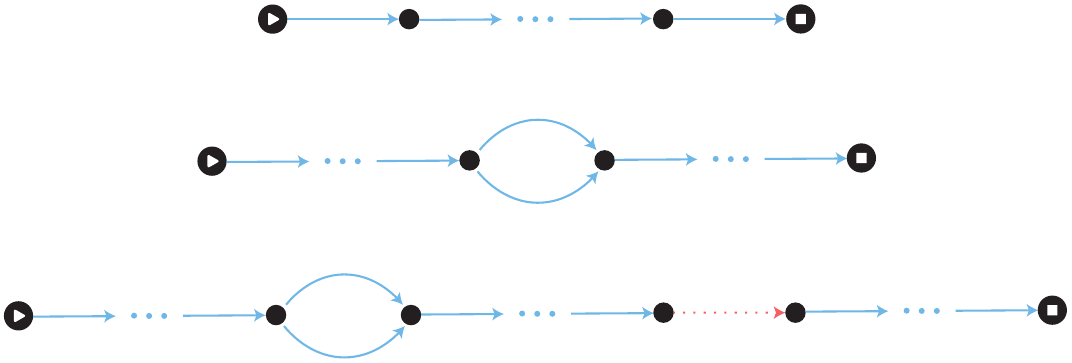}
    \caption{Examples of sub-instances whose satisfiability is determined with a quantum algorithm. Top: Instance with an active clock chain containing no undefined or simultaneous propagation clauses. The length of the chain may be $0 \leq L < \mathcal{O}(n)$. Middle: Instance with simultaneous propagation clauses followed by well-defined clauses. Bottom: Instance with simultaneous propagation clauses followed by both well-defined and undefined $\Pi_{prop}$ clauses.}
    \label{fig:qinputs}
\end{figure}

\subsubsection{A clock component and its logical qudits} \label{subsubsection:clock}

As mentioned previously and as illustrated by the following four lemmas, the principle of monogamy of entanglement allows us to reject a large number of clock components. Assuming that the clock components have at least one $\Pi_{prop}$ clause, we can show that the only kind of sub-instances not deemed unsatisfiable by the monogamy condition are those that form a single one-dimensional chain of clock qudits with a unique direction. Besides, if they contain $\Pi_{init}$ and $\Pi_{out}$ clauses, these must be positioned at the ends of the chain, with all arrows pointing away from $\Pi_{init}$ and into $\Pi_{out}$. We note that these instances may contain simultaneous propagation clauses, as long as they point in the correct direction. Let us state the lemmas corresponding to these claims, starting with the lemmas regarding the connections of clock qudits for clock components with multiple clock qudits. Afterwards, we briefly discuss the clock components formed by a single clock qudit.

\begin{restatable}[Clock qudit two neighbor maximum]{lemma}{linearchainmono} \label{lemma:linearchainmono}
    If a clock qudit in a clock component is connected to more than two other clock or endpoint qudits, the instance is unsatisfiable.
\end{restatable}

\begin{restatable}[Unique direction of a clock chain]{lemma}{directionmono} \label{lemma:directionmono}
    If a clock qudit in the chain has two successors or two predecessors, the instance is unsatisfiable.
\end{restatable}
\noindent
As promised, these two lemmas reject all clock components that do not form a one-dimensional chain of clock qudits where all $\Pi_{prop}$ clauses point in the same direction. These clock components may be either in the form of a cycle or a one-dimensional line. However, neither lemma can reject the case where a clock qudit in the chain is connected to its successor and/or predecessor via multiple $\Pi_{prop}$ clauses (which may have different unitaries and act on different logical qudits). The principle of monogamy of entanglement is incapable of ruling out these instances since each additional $\Pi_{prop}$ (with the correct direction) does not require entangling new subspaces of the clock qudits. As a result, the same already-entangled subspaces can be reused. Although this behavior is a required feature of the construction for $\Pi_{init}$ clauses, we cannot prevent the same from happening for $\Pi_{prop}$ clauses.\footnote{To prove $\BQP_1$-hardness it is necessary to stack multiple $\Pi_{init}$ clauses on the same clock qudit in order to initialize all of the required logical qudits; see \cref{fig:qinstance}.} The stacking of these clauses presents a great downside and an extra difficulty in deciding the satisfiability of an instance. In Ref.\ \cite{meiburg2021quantum}, Meiburg crucially missed this important event, which as we will see shortly, requires a more comprehensive quantum algorithm than the one presented there. We discuss these instances in more detail in \cref{subsubsection:simultaneous}, for now, let us present two lemmas regarding the $\Pi_{init}$ and $\Pi_{out}$ clauses.

\begin{restatable}[Unique endpoint qudit for $\Pi_{init}/\Pi_{out}$ clauses]{lemma}{uniqueclockmono} \label{lemma:uniqueclockmono}
    If an endpoint qudit is connected to more than one clock qudit, or connected to a single clock qudit via a $\Pi_{init}$ and $\Pi_{out}$ clause, the instance is unsatisfiable.
\end{restatable}

\begin{restatable}[Unique clock qudit for $\Pi_{init}/\Pi_{out}$ clauses]{lemma}{uniqueendpointmono} \label{lemma:uniqueendpointmono}
    Let $c_0$ refer to the clock qudit that all $\Pi_{prop}$ clauses point away from, and $c_L$ the one at the other end of the chain. If a $\Pi_{init}$ clause acts on a clock qudit that is not $c_0$, the instance is unsatisfiable. Similarly, if a $\Pi_{out}$ clause acts on a clock qudit that is not $c_L$, the instance is unsatisfiable.
\end{restatable}
\noindent
The first lemma shows that satisfiable instances cannot have clock components joined together by endpoint qudits, which is why we ignored these cases in the first place. Furthermore, this lemma also rejects instances where an endpoint qudit is present in both $\Pi_{init}$ and $\Pi_{out}$ clauses even if they connect to the same clock qudit. The second lemma rejects instances where $\Pi_{init}$ or $\Pi_{out}$ clauses act on the wrong clock qudits. Finally, observe that these lemmas also do not disallow stacking $\Pi_{init}$ (or $\Pi_{out}$) clauses as long as they act on the same clock and endpoint qudit. \\
\indent
To summarize, the sub-instances not rejected by any of the four lemmas above must be those with one-dimensional clock chains where all $\Pi_{prop}$ clauses point in the same direction. Moreover, if the sub-instance contains $\Pi_{init}$ or $\Pi_{out}$ clauses, all $\Pi_{init}$ clauses must act on $c_0$ and a single endpoint qudit, and all $\Pi_{out}$ on $c_L$ and a single endpoint qudit (the endpoint qudits must also be different). In other words, $\Pi_{init}$ and $\Pi_{out}$ mark the endpoints of the chain and no other $\Pi_{prop}$ clause can extend past them. Also, these show that cyclic instances that may be satisfiable cannot have any $\Pi_{init}$ or $\Pi_{out}$ clauses. \\
\indent
Now, let us show that clock components without both $\Pi_{init}$ and $\Pi_{out}$ clauses, regardless of whether they have undefined or simultaneous propagation clauses, are trivially satisfiable.

\begin{restatable}[Lack of an ACC]{lemma}{noendpointsmono} \label{lemma:noendpointsmono}
    If a clock component does not have at least one $\Pi_{init}$ and one $\Pi_{out}$ clause, the sub-instance is satisfiable.
\end{restatable}
\indent
Now it should be apparent why the order in which we evaluate the clock component matters. For example, we cannot conclude that a clock component with only a $\Pi_{init}$ clause has a satisfying state without first checking that the instance is linear, the direction of its $\Pi_{prop}$ clauses, and the arrangement of the endpoint qudits.\\
\indent
The sub-instances not decided by any of the lemmas above must be those consisting entirely of an ACC which also displays a clear computational flow of time, i.e.\ they have a start, an end, and all $\Pi_{prop}$ clauses in between point away from the starting qudit. The ACC however may contain simultaneous propagation clauses. At the moment, we can gather that if these sub-instances are satisfiable, at least one of the clock qudits of the chain must be in the state $\ket{a}$. This is because setting all clock qudits of the chain to $\ket{r}$ violates the $\Pi_{init}$ clause, setting them to $\ket{d}$ violates the $\Pi_{out}$ clause, and the $\Pi_{clock}$ (both $\Pi_{clock,D}$ and $\Pi_{clock,?}$) clauses penalize any direct clock transitions from $\ket{d}$ to $\ket{r}$ and vice versa. \\
\indent
For a moment, consider one of these sub-instances without simultaneous propagation clauses and imagine we had defined the behavior of $\Pi_{prop}$ clauses in a single way (as in \cref{eqn:projs}) instead of conditioning on the state of the logical qudit. The presence of the $\ket{a}$ state would trigger the propagation of computation (in order to satisfy the $\Pi_{prop}$ clauses) through the whole chain, and the only satisfying state would be a history state acting on both initialized and free logical qudits. As mentioned previously, the free logical qudits give one the opportunity to guess a solution and the sub-instance therefore results in a difficulty likely greater than $\BQP_1$. Conditioning on the state of the logical qudits thus offers us other possible satisfying assignments. At this moment however, there is ambiguity on which logical qudits must be initialized or may remain free. The following proposition offers a way forward:\footnote{We state this as a proposition rather than a lemma as it concerns the form of a potential satisfying state, and not a condition about the satisfiability of the instance.}

\begin{restatable}[Initialization of logical qudits]{prop}{activetaccmono} \label{prop:activetaccmono}
    Let $c_0, \ldots, c_L$ with $L \geq 0$ be the clock qudits of an ACC. Let $c_0$ be the qudit present in the $\Pi_{init}$ clauses and $c_L$ the one in the $\Pi_{out}$ clauses. A state of the qudits of the sub-instance where $c_0$ is not $\ket{a}$ in any basis state of the superposition is not a satisfying state.
\end{restatable}
\noindent
This proposition implies that if a satisfying state of the sub-instance exists, $c_0$ must be $\ket{a}$ at some point in ``time'' (must be $\ket{a}$ in at least one of the basis states of the satisfying superposition). Consequently, at this ``time'', all logical qudits present in the $\Pi_{init}$ clauses attached to $c_0$ are required to be in the state $\ket{0}$. The other logical qudits remain free. This then allows us to tell apart the $\Pi_{prop}$ clauses that must be well-defined from those that may not. \\
\indent
\cref{prop:activetaccmono} also requires us to make a decision. Recall that since sub-instances may share logical qudits, it is possible that a sub-instance constrains a logical qudit to be $\ket{0}$, but in a different sub-instance this logical qudit is not initialized. Should we think of the logical qudit as initialized to $\ket{0}$ in both sub-instances? Or should the constraint be considered locally and let the qudit be $\ket{0}$ in the first sub-instance and free in the second? As is discussed in more detail in \cref{appendix:choice}, while both ways yield the same result, we choose to perform this operation globally as it simplifies the remaining parts of the analysis. This way, when we encounter an undefined clause while analyzing a clock component, it ensures us that it is not initialized by any clock component. As a result, the TACCs are fixed which simplifies the analysis of instances. This will become even clearer when considering the classical algorithm of \cref{subsection:algorithmmono}, in particular steps (3) and (4).\\
\indent
Importantly, this proposition also shows that it is not possible to abuse the undefined states and claim that all instances are satisfied by setting all logical qudits to $\ket{?}$. Simply, this assignment would not be consistent with the state of the clock qudits.\\
\indent
Finally, while \cref{prop:activetaccmono} demonstrates that TACCs must have at least one active state, we now show that they must have exactly one such state.

\begin{restatable}[Unique active state in a TACC]{prop}{uniqueactivetaccmono} \label{prop:uniqueactivetaccmono}
    Let $c_0, \ldots, c_T$ with $0 < T \leq L$ be the clock qudits of a TACC with non-zero length. A state in which at any given time there are two clock qudits $c_a$ and $c_b$ with $0 \leq a < b \leq T$ in the state $\ket{a}$ is not a satisfying state.
\end{restatable}
\noindent
These two propositions show that satisfying the clauses within TACCs require that there is exactly one active clock qudit, which is in agreement with the desired encoding of the clock in \cref{eqn:clock}. As a consequence of this statement, satisfying the $\Pi_{prop}$ clauses within the chain (which are all well-defined) implies that the satisfying state must demonstrate the evolution of a state under the circuit $U=U_T \ldots U_1$ (see \cref{eqn:propsat}). Furthermore, if the TACC terminates with a $\Pi_{out}$ clause, then these are the same conditions that apply on the clock chain of instances like \cref{fig:reductioninstance}, and we can conclude that, if satisfiable, it must be uniquely satisfied by a history state. However, it is not entirely clear if the same holds for TACCs that terminate because of an undefined clause. We now show that these chains are always satisfied by a history state that is truncated at the time the computation reaches the undefined clause.

\begin{restatable}[Satisfying state of a truncated ACC]{lemma}{truncated} \label{lemma:truncated}
    Let $c_0, \ldots, c_L$ be the clock qudits of an ACC, and suppose $\Pi_{prop}^{(c_T,c_{T+1})}$ with $0 \leq T < L$ is the first propagation clause that acts on a free logical qudit $l_u$. The clock qudits $c_0, \ldots, c_T$ then form a TACC. Furthermore, assume the TACC has no simultaneous propagation clauses, and let $U_1, \ldots, U_T$ be the unitaries associated with the $\Pi_{prop}$ clauses of this chain. Finally, let $S$ denote the set of all logical qudits of the sub-instance, and $D$ with $\norm{D} = q$ the subset of these qudits that are initialized by the $\Pi_{init}$ clause on $c_0$. Under these conditions, the sub-instance is trivially satisfiable and the satisfying state is given by

    \begin{equation} \label{eqn:truncatedhist}
        \ket{\psi_{thist}} := \left[ \frac{1}{\sqrt{T+1}}\sum_{t = 0}^{T} U_{t} \ldots U_0 \ket{0}^{\otimes q} \otimes \ket{?}_{l_u} \otimes \ket{ \underbrace{d \ldots d}_\text{t} a_t  \underbrace{r \ldots r}_\text{T-t}} \right] \otimes \ket{\phi}_{S\setminus D} \otimes \ket{\underbrace{r \ldots r}_\text{L-T}},
    \end{equation}
    \noindent
    where $\ket{\phi}$ is an arbitrary state of the logical qudits of the sub-instance not present in clauses within the TACC.
\end{restatable}
\noindent
The satisfying state shows that the clauses within the TACC do enforce constraints on their corresponding logical qudits by first requiring that they are initialized to $\ket{0}$ and then acted on by unitary operations. Conversely, the $\Pi_{prop}$ and $\Pi_{out}$ clauses outside the truly active clock chain do not constrain their logical qudits. Hence, when considering logical qudits present across multiple clock components, the ones that should actually be considered as ``shared'' are those belonging to clauses stemming from truly active clock chains. This also suggests that the classical algorithm presented in the following section should first evaluate the satisfiability of the clauses that share logical qudits before deciding that the sub-instances with truncated ACCs (and no simultaneous propagation clauses within the TACC) are trivially satisfiable. \\
\indent
The previous lemma is the last statement we can make about deciding the satisfiability of a sub-instance based on the structure of its clock component. This is because if the sub-instance has not been decided by the lemmas above and its logical qudits do not conflict with other sub-instances, they must be one of the three types of ``quantum'' sub-instances shown in \cref{fig:qinputs}. Indeed, if not rejected by the lemmas above, the clock component of the sub-instance must form a one-dimensional chain with the following properties:

\begin{itemize}
    \item[--] There is at least one $\Pi_{init}$ clause, which must act on a clock qudit at one of the ends of the chain. If there are multiple such clauses, they must all act on this same qudit, as well as use the same endpoint qudit.
    \item[--] There is at least one $\Pi_{out}$ clause, which must act on the qudit at the other end of the chain. If there are multiple such clauses, they must all act on this same qudit, as well as use the same endpoint qudit. The endpoint qudit must be different than the one used by the $\Pi_{init}$ clauses.
    \item[--] There is at least one $\Pi_{prop}$ clause, which must point away from the clock qudit with the $\Pi_{init}$ clauses and towards the qudit with the $\Pi_{out}$ clauses. If there are multiple $\pi_{prop}$ clauses, they must all have this same direction.
\end{itemize}
\noindent
Additionally, it is possible for these sub-instances to have multiple $\Pi_{prop}$ clauses involving the same pair of clock qudits, and if this is true, they might also have undefined clauses after the simultaneous propagation. These are the three types of ``quantum'' sub-instances just mentioned.

\begin{remark} \label{remark:trivial}
    Among these ``quantum'' sub-instances, there are two that are actually trivially decidable. These are the sub-instances with no undefined or simultaneous propagation clauses, where the logical qudits present in $\Pi_{out}$ clauses are (1) initialized but not acted by any $\Pi_{prop}$ clause, or (2) not initialized.\footnote{By construction, the second type of sub-instance cannot be acted by a $\Pi_{prop}$ clause as otherwise it would create an undefined clause.} It is not difficult to see that the first type is trivially unsatisfiable, while the second type is trivially satisfiable. Despite their triviality, they will be decided by the quantum algorithm. This aims to simplify the classical algorithm and minimize the number of cases it needs to handle.
\end{remark}

\medbreak
\noindent
\textbf{Single clock qudit components}\\
\noindent
Let us now provide some assurance that clock components consisting of a single clock qudit can also be decided using the lemmas above or must have a structure resembling that of longer chains that are decided by the quantum subroutine.\\
\indent
By definition, this clock component can only be acted on by either $\Pi_{init}$ and/or $\Pi_{out}$ clauses. \cref{lemma:linearchainmono} states that the clock qudit must be connected to at most two endpoint qudits. Then, due to the lack of $\Pi_{prop}$ clauses, \cref{lemma:directionmono,lemma:uniqueendpointmono} do not apply. If the sub-instance is not decided by the other lemmas above, then it must be acted on by both $\Pi_{init}$ and $\Pi_{out}$ clauses, where there is a unique endpoint qudit for the $\Pi_{init}$ clauses and a separate endpoint qudit for $\Pi_{out}$ clauses. For such a sub-instance, the only valid state of the clock qudit must be $\ket{a}$, implying that the logical qudits present in the $\Pi_{init}$ clauses must be initialized to $\ket{0}$ and the logical qudits present in the $\Pi_{out}$ clauses to $\ket{1}$. The instance then falls under two categories: either there is a logical qudit present in both $\Pi_{init}$ and $\Pi_{out}$ clauses, or there is none. These are the two sub-instances mentioned in \cref{remark:trivial}. As such, after verifying that their shared logical qudits (if any) do not create any conflicts, they will be decided by the quantum algorithm.

\subsubsection{Shared logical qudits} \label{subsubsection:shared}

As mentioned previously, partitioning the instance into sub-instances based on connected clock qudits may not result in completely disjoint sub-instances as some of the logical qudits may overlap. We now determine the joint satisfiability of these sub-instances. \\
\indent
To this end, we consider one of the shared logical qudits and determine if there exists a state of the sub-instances such that all clauses in which it appears can be satisfied simultaneously. The observation at the end of the previous subsection allows us to make a simplification. There, we noted that the logical qudits of a sub-instance that are actually constrained are those present in the clauses of the TACC. Thus, the relevant shared logical qudits are exclusively those that are present in clauses from multiple TACCs. \\
\indent
The satisfiability of the joint sub-instances depends on the clauses applied to the logical qudit. Let us now present the lemmas about the satisfiability of these cases, starting with the lemma where the shared qudit is acted by a unitary gate from a clock component.

\begin{restatable}[No $\Pi_{prop}$ clause on shared qudit]{lemma}{sharedhprop} \label{lemma:sharedhprop}
    If a logical qudit is present in a $\Pi_{prop}$ clause of a TACC and this qudit is also acted on by a clause from a different TACC, the instance is unsatisfiable.
\end{restatable}
\noindent
This lemma shows the unsatisfiability of shared logical qudits that are acted on by a single or multiple unitaries. The remaining cases are those were the logical qudit is acted on only by $\Pi_{init}$ clauses, $\Pi_{out}$ clauses, or a mix of both types of clauses.\footnote{Observe that the case where the logical qudit is only acted by $\Pi_{out}$ clauses can occur if the ACCs of the clock components do not have undefined clauses.} The satisfiability of these cases is given by the following lemmas.

\begin{restatable}[No $\Pi_{init}$ and $\Pi_{out}$ clauses on shared qudit]{lemma}{mutualboth} \label{lemma:mutualboth}
    If a logical qudit is acted on by a $\Pi_{init}$ clause from a clock component and a $\Pi_{out}$ clause from another clock component, the instance is unsatisfiable.
\end{restatable}

\begin{restatable}[Either all $\Pi_{init}$ or all $\Pi_{out}$ clauses]{lemma}{mutualsingle} \label{lemma:mutualsingle}
    If a logical qudit is present in clauses from multiple clock components but the clauses are either all $\Pi_{init}$ or all $\Pi_{out}$ clauses, there exists a state of the logical qudit that simultaneously satisfies these clauses.
\end{restatable}
\noindent
These lemmas conclude the analysis of the shared logical qudits. Now, let us finally discuss the unresolved issue of simultaneous propagation clauses.

\subsubsection{Simultaneous propagation clauses} \label{subsubsection:simultaneous}

None of the lemmas in \cref{subsubsection:clock} which analyze the clock component in the sub-instance are able to decide if simultaneous $\Pi_{prop}$ clauses are trivially satisfiable or unsatisfiable. At best, \cref{lemma:directionmono} establishes that these create unsatisfiable instances if the direction of one of the $\Pi_{prop}$ clauses is incorrect. The same can be said about the lemmas in \cref{subsubsection:shared}. Even in the case where the unitaries associated with simultaneous $\Pi_{prop}$ clauses act on the same logical qudits, \cref{lemma:sharedhprop} is unable to determine the unsatisfiability of the instance since the clauses belong to the same clock component. We now show the necessary and sufficient conditions for which these instances may admit a satisfying state.

\begin{restatable}{prop}{simultaneousprop} \label{prop:simultaneousprop}
    Let $c_0, \ldots, c_T$ with $0 < T \leq L$ be the clock qudits of a TACC with non-zero length. Suppose that, for the first time in the chain, $k$ distinct $\Pi_{prop}$ clauses act on $c_t$ and $c_{t+1}$ with $0 \leq t < T$ and an arbitrary pair of logical qudits. The logical qudits may or may not be repeated. Let $U_{t+1,0}, \ldots, U_{t+1,k-1}$ denote the unitaries associated with these $k$ clauses, and $\ket{\phi_t}$ be the state of the logical register at time $t$, i.e.\ $\ket{\phi_t} = U_{t,0} \ldots U_{0,0} \ket{0}^{\otimes q}$ with $U_{0,0} = I$. All $k$ $\Pi_{prop}$ clauses can be satisfied simultaneously if and only if $U_{t+1,i} \ket{\phi_t} = U_{t+1,j} \ket{\phi_t}$ for all $i,j \in [k]$.
\end{restatable}
\indent
This proposition shows that to determine the satisfiability of simultaneous propagation clauses, we can compare the action of each unitary on the state $\ket{\phi_t}$. However, since $t \in \mathcal{O}(n)$ and the set of gates used is universal for quantum computation, it is unlikely that we can evaluate this classically. \\
\indent
Another point to consider is whether there are actually sub-instances where two or more simultaneous propagation clauses can be satisfied given that these can only be of three types: $H$, $HT$, or $(H \otimes H)$CNOT. If the gates act on different qudits, then this is not too difficult to imagine. In the case where the unitaries act on at least one common qudit, observe that $H_0 \otimes I_1 \ket{\phi_t} = (H_0 \otimes H_1) \textnormal{CNOT}_{0,1} \ket{\phi_t}$ if $\ket{\phi_t} = \ket{0} \otimes \ket{H_+}$ and $\ket{H_+}$ is the $+1$-eigenstate of the Hadamard gate.

\subsection{Algorithm and correctness} \label{subsection:algorithmmono}

In this subsection we present the hybrid algorithm showing that LCT-QSAT can be decided in polynomial time using both a classical and quantum computer while adhering to the soundness and completeness conditions mentioned in \cref{defn:bqp1}. In other words, we show that LCT-QSAT is contained within $\BQP_1$. \\
\indent
In summary, the first part of the algorithm determines the satisfiability of instances by following the lemmas of \cref{subsection:monosat}. As mentioned there, the satisfiability of an instance is first (and most of the time) dictated by the arrangement of its clauses. Since evaluating the clauses of an instance consist of simple lookups in a graph of polynomial size, these steps can be efficiently performed with a classical computer. If the algorithm does not reject (meaning it does not encounter any unsatisfiable clauses), then we know for certain that the sub-instances within have proper structure and the shared logical qudits create no conflicts. These sub-instances then fall under two categories: either they are trivially satisfiable (those that lack an ACC or have a truncated ACC) or those we have identified that use a quantum algorithm. The algorithm analyzes each sub-instance, and makes use of a quantum subroutine if it happens to be of the second kind. Lastly, if the classical and quantum algorithm do not reject, the algorithm accepts the instance. \\
\indent
Before describing the algorithm in detail, let us make some observations. To meet the completeness and soundness conditions of $\BQP_1$, we have to make sure that the algorithm does not reject satisfiable instances and does not accept unsatisfiable instances with high probability. We will prove this in more detail later, but this will partly follow from the correctness of the lemmas and propositions discussed in \cref{subsection:monosat}. Furthermore, perfect completeness requires that the algorithm consistently performs operations with perfect accuracy. This is of most concern in the quantum algorithm, which we keep in mind in its design. Finally, we note that to demonstrate that the problem is in $\BQP_1$, the algorithm only needs to decide if the instance is satisfiable---it is not required to produce the satisfying state.

\subsubsection{Classical algorithm} \label{sec:classical_algorithm}
Our algorithm is the following:

\begin{enumerate}[label=(\arabic*)]
    \item \textit{Verify that all qudits in the instance serve a unique role}. For every qudit in the instance, check whether it is acted on by $\Pi_C$, $\Pi_E$, or $\Pi_L$ clauses. If a qudit is acted on by more than one type of clause, reject. Otherwise, label the qudit as either $C$, $E$, or $L$ depending on the type of clause acting on it. Correctness is given by \cref{lemma:singletypequditmono}.
    \item \textit{Consider the sub-instances that make up the instance. Each sub-instance is composed by a clock component (but now also considering all connections made by endpoint qudits) and its associated logical qudits.} For every sub-instance, perform the following checks:
          \begin{enumerate}[label=(2.\arabic*)]
              \item If a clock qudit is connected to more than two other distinct clock or endpoint qudits, reject. Correctness is given by \cref{lemma:linearchainmono}.
              \item If a clock qudit in the chain has two successors or two predecessors, reject. Correctness is given by \cref{lemma:directionmono}.
              \item If an endpoint qudit is connected to more than one clock qudit, reject. Similarly, if an endpoint qudit is connected to the same qudit via both a $\Pi_{init}$ and $\Pi_{out}$ clause, reject. Correctness is given by \cref{lemma:uniqueclockmono}.
              \item If a $\Pi_{init}$ clause acts on a qudit other than the one at the beginning of the chain, reject. Similarly, if a $\Pi_{out}$ clause acts on a qudit other than the one at the end of the chain, reject. Correctness is given by \cref{lemma:uniqueendpointmono}.
              \item If there are no $\Pi_{init}$ and $\Pi_{out}$ clauses, ignore all clauses of the sub-instance for the rest of the algorithm. Correctness is given by \cref{lemma:noendpointsmono}.
          \end{enumerate}
    \item \textit{Identify the undefined $\Pi_{prop}$ clauses.} For every logical qudit, check whether it is acted on by at least one $\Pi_{init}$ clause. If it is not, mark all $\Pi_{prop}$ clauses this qudit is part of as undefined.
    \item \textit{Now, the satisfiability of the instance depends on the shared logical qudits of clock components. Importantly, only those within TACCs are considered.} Identify the logical qudits involved in truly active clock chains of multiple clock components. For every one of these qudits perform the following checks:\footnote{A step to include \cref{lemma:mutualsingle} is not necessary as it requires no action.}
          \begin{enumerate}[label=(4.\arabic*)]
              \item If the qudit is acted on by a $\Pi_{prop}$ clause from a clock component, and is also present in another clause stemming from a different clock component, reject. Correctness is given by \cref{lemma:sharedhprop}.
              \item If the qudit is acted on by both $\Pi_{init}$ and $\Pi_{out}$ clauses, reject. Correctness is given by \cref{lemma:mutualboth}.
          \end{enumerate}
    \item \textit{Now that we have verified the logical qudits of truncated active clock chains, we can safely ignore them.} For every sub-instance, check whether it contains an undefined propagation clause. If it does and there are no simultaneous propagation clauses before the undefined clause, ignore all clauses of the sub-instance. Correctness is given by \cref{lemma:truncated}.
    \item \textit{By this point in the algorithm, all trivial sub-instances have been decided, and the logical qudits of those that remain do not create conflicts.} For every remaining sub-instance, run the quantum algorithm to determine its satisfiability.
    \item Accept the instance. \textit{If the algorithm reaches this step, then none of the sub-instances were rejected in the previous steps so we accept the instance.}
\end{enumerate}

\subsubsection{Quantum algorithm}

As mentioned previously, the classical algorithm tasks the quantum subroutine to decide three types of sub-instances: those with unique active clock chain that have (1) no simultaneous propagation clauses and well-defined $\Pi_{prop}$ clauses, (2) simultaneous propagation clauses followed by well-defined $\Pi_{prop}$ clauses, or (3) simultaneous propagation clauses followed by one or multiple undefined $\Pi_{prop}$ clauses, shown in \cref{fig:qinputs}. As the clauses outside the TACC can be satisfied by choosing the state of the clock qudits appropriately, the goal of the quantum algorithm is to decide if there exists a state that satisfies all constraints of the TACC. For the remainder of this discussion, we let the TACC be of length $T$, act on $q$ logical qudits, and let $k_{t,t+1}$ denote the number of simultaneous propagation clauses involving clock qudits $c_t$ and $c_{t+1}$ with $t \in [T]$. We note that since the logical qudits of the TACC cannot be $\ket{?}$, in the actual quantum circuits they are simply referred to as logical \textit{qubits}.\\
\indent
Let us consider the first type of sub-instance which does not contain simultaneous propagation clauses, i.e.\ $k_{t,t+1} = 1$ for all $t \in [T]$. For these sub-instances, we know that the $\Pi_{init}$ and $\Pi_{prop}$ clauses of the chain are uniquely satisfied by the history state corresponding to the circuit expressed by the TACC (mentioned below \cref{eqn:hamils}). Therefore, the algorithm's only task is to determine if this state also satisfies the $\Pi_{out}$ clauses.\\
\indent
One verification method, analogous to Bravyi's algorithm in \cref{subsubsection:inqma}, involves constructing the history state and measuring the eigenvalue of the $\Pi_{out}$ clauses to verify if they can be satisfied. A quick inspection shows that this method is wasteful as the $\Pi_{out}$ clauses are only concerned with the state of the qudits at the end of the computation, i.e.\ only one out of the $T+1$ basis states in the history state's superposition is relevant. A simpler approach is to create the quantum circuit $U = U_T \ldots U_1$ expressed by the TACC, execute it on the initial state $\ket{0}^{\otimes q}$, measure the state of the logical qubits present in the $\Pi_{out}$ clauses at the end of the computation, and flip the measurement outcomes. If any of these result in ``1'', we reject. To see why these methods are equivalent, consider a $\Pi_{out}$ clause that acts on logical qudit $l_{out}$. The probability of measuring eigenvalue $1$ with Bravyi's method (see \cref{eqn:probmeasure}) is $\Pr(\textnormal{outcome } 1) = \bra{\psi_{hist}} \Pi_{out} \ket{\psi_{hist}} = (T+1)^{-1} \bra{0}^{\otimes q}U^\dagger_{0} \ldots U^\dagger_{T} \Pi^{(0)} U_{T} \ldots U_{0} \ket{0}^{\otimes q}$, where we observed that all other terms of $\Pi_{out}$ have no contribution, and defined $\Pi^{(0)} := \ket{0}\!\bra{0}_{l_{out}} \otimes I_{rest}$. It is then apparent that this quantity is proportional to the probability that $l_{out}$ is measured at the end of the circuit and yields outcome ``0''. We flip the measurement outcomes in the alternative method so in both cases, observing ``1'' signifies a clause violation. \\
\indent
Unfortunately, for the second and third type of sub-instance, these procedures cannot apply directly since there is no longer a unique history state and the simultaneous $\Pi_{prop}$ clauses should also be verified. To resolve this, let us first develop a method similar to Bravyi's that works for these cases, and then determine a more efficient way of implementing it that avoids creating the full history state. As we will soon see, establishing the equivalence between the two methods is quite useful, as proving the correctness of the algorithm (completeness and soundness) is easier in Bravyi's approach. \\

\medbreak
\noindent
\textbf{Bravyi's Approach}\\
\noindent
Although sub-instances may now contain multiple history states, \cref{prop:simultaneousprop} shows that if the sub-instance is satisfiable, there is a single history state independent of the propagation ``path'' chosen. We can thus define a ``privileged'' history state

\begin{equation} \label{eqn:phist}
    \ket{\psi_{phist}} := \frac{1}{\sqrt{T+1}} \sum_{t=0}^{T} U_{t,0} \ldots U_{0,0} \ket{0}^{\otimes q} \otimes \ket{ \underbrace{d\ldots d}_{t} a_t \underbrace{r \ldots r}_{T-t}},
\end{equation}
\noindent
where the computation is performed with unitaries $U_{t,0}$ for all $1 \leq t \leq T$ and use this state to test whether the other simultaneous $\Pi_{prop}$ clauses are also satisfied by this state. We can then design the quantum algorithm to measure the eigenvalue of these clauses on $\ket{\psi_{phist}}$, and at the end, measure the eigenvalues of the $\Pi_{out}$ clauses that form part of the TACC (if any). \\

\medbreak
\noindent
\textbf{Our Approach}\\
\noindent
With the intention of building a simpler approach that does not require building a history state, consider the following claim.

\begin{restatable}{claim}{measuringprop} \label{claim:measuringprop}
    Let $\Pi_{prop}^{(c_t,c_{t+1})}$ be a clause acting on clock qudits $c_t$ and  $c_{t+1}$ with associated unitary $U_{t+1,j}$, where $j \in [k_{t,t+1}]$. Measuring the eigenvalue of this clause on the state $\ket{\psi_{phist}}$, is equivalent to applying circuit $\mathcal{C}$ of \cref{fig:measurecircuit} to the state $\ket{\phi_t}$ and an ancilla qubit in the state $\ket{0}$. Here, $\ket{\phi_t}$ is the state of the data register at time $t$, i.e.\ $\ket{\phi_t} = U_{t,0} \ldots U_{0,0} \ket{0}^{\otimes q}$. Furthermore, this circuit can be simulated perfectly by a circuit $\mathcal{C}_{\mathcal{G}_8}$ that uses a constant number of gates from the set $\mathcal{G}_8 = \{H,T,\textnormal{CNOT}\}$ and an additional ancilla qubit.
\end{restatable}
\noindent
Thus, we can build a quantum circuit that consists of evolving the state according to the circuit expressed by the TACC, but if it encounters simultaneous propagation clauses ahead, it applies circuit $\mathcal{C}_{\mathcal{G}_8}$ instead. If the measurement of the ancilla yields outcome ``1'', we reject the instance. We would like the algorithm to continue repeating this pattern of updating the current state and performing checks when necessary, however, there is one nuisance that we must address. If the instance is unsatisfiable and one of the checks mistakenly passes (as could be the case if the simultaneous unitaries produce similar states), the post-measurement state becomes $(I-\Pi_{prop}) \ket{\psi_{phist}}$ (in Bravyi's description). If the algorithm is allowed to continue, the following checks would not use the privileged history state, but rather this perturbed state. Although this might still be permissible by the soundness condition, for ease of analysis, we design the algorithm to avoid it altogether. Therefore, after applying a check and measuring the ancilla qubit, we simply restart the algorithm, with the condition that in this new iteration, the algorithm skips all previously performed checks. \\

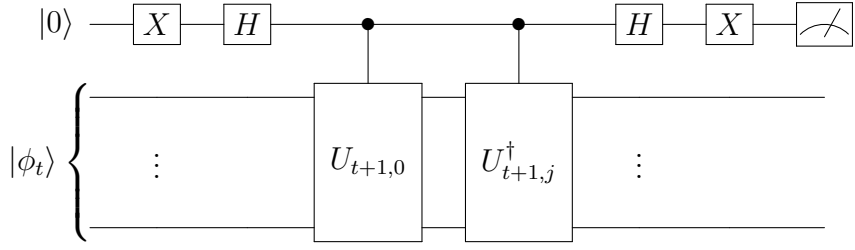
\begin{figure}[t]
    \centering
    \large
    \[ \Qcircuit @C=1.4em @R=1.2em {
        \lstick{\ket{0}} & \gate{X} & \gate{H} & \ctrl{1} & \ctrl{1} & \gate{H} & \gate{X} & \meter \\
        & \qw & \qw & \multigate{2}{U_{t+1,0}} & \multigate{2}{U_{t+1,j}^\dagger} & \qw & \qw & \qw \\
        &  \qvdots & \nghost{} & \nghost{U_{t+1,0}} & \nghost{U_{t+1,j}^\dagger}  & \qvdots &   \\
        & \qw & \qw & \ghost{U_{t+1,0}}  & \ghost{U_{t+1,j}^\dagger} & \qw & \qw  & \qw
        \inputgroupv{2}{4}{0.8em}{2.1em}{\ket{\phi_t}\;} \\
        }
    \]
    \caption{Circuit $\mathcal{C}$ that evaluates whether propagating $\ket{\phi_t}$ with $U_{t+1,0}$ is the same as propagating with $U_{t+1,j}$. Here, $U_{t+1,0}$ and $U_{t+1,j}$ act nontrivially on at most two qubits of the data register. This circuit is equivalent to measuring the eigenvalue of projector $\Pi_{prop}$ (with unitary $U_{t+1,j}$) on $\ket{\psi_{phist}}$. If $U_{t+1,0}$ and $U_{t+1,j}^\dagger$ are unitaries from the set $\mathcal{G}_8$, the circuit can be perfectly simulated with this set.}
    \label{fig:measurecircuit}
\end{figure}

\medbreak
\noindent
\textbf{Algorithm}\\
\noindent
Explicitly, our quantum algorithm is the following:

\begin{enumerate}[label=(6.\arabic*)]
    \item Initialize a register containing $q+2$ qubits, each initialized to $\ket{0}$. \textit{The first $q$ qubits correspond to the logical qudits of the TACC, while the other two are the ancillas used in $\mathcal{C}_{\mathcal{G}_8}$.}
    \item For every $t \in [T]$, do the following:
          \begin{enumerate}[label=(6.2.\arabic*)]
              \item If $k_{t,t+1} > 1$:
                    \begin{itemize}
                        \item[--] Choose an index $j \in [k_{t,t+1}]$. Evaluate circuit $\mathcal{C}_{\mathcal{G}_8}$ using $U_{t+1,j}$.
                        \item[--] Remove the $\Pi_{prop}$ clause with associated unitary $U_{t+1,j}$ from the sub-instance.
                        \item[--] If the measurement outcome is ``1'', reject. Otherwise, go back to step (6.1).
                    \end{itemize}
              \item Otherwise, apply $U_{t+1,0}$ to the data register to obtain $\ket{\phi_{t+1}}$.
          \end{enumerate}
    \item \textit{At this step, the algorithm has produced the state $\ket{\phi_T}$.} If there's at least one $\Pi_{out}$ clause involving $c_T$, do the following:
          \begin{enumerate}[label=(6.3.\arabic*)]
              \item Choose one of the $\Pi_{out}$ clauses and measure the corresponding data qubit. Flip the resulting outcome.
              \item Remove this $\Pi_{out}$ projector from the sub-instance.
              \item If the outcome is ``1'', reject. Otherwise, go back to step (6.1).
          \end{enumerate}
    \item Return to the classical algorithm.
\end{enumerate}

\subsubsection{Runtime}
For an instance with $n$ qudits and $\mathcal{O}(poly(n))$ clauses, the algorithm reaches a decision in $\mathcal{O}(poly(n))$ time. \\
\indent
In detail, steps (1), (3), and (4) require iterating over the qudits of the instance (either all of them or some particular subset). For each qudit, the algorithm performs checks based on the clauses acting on this qudit. Given that there are $n$ qudits and $\mathcal{O}(poly(n))$ clauses acting on each qudit, each step takes $n \cdot \mathcal{O}(poly(n)) = \mathcal{O}(poly(n))$ time. Step (2) iterates over all sub-instances. Observe that because each sub-instance must contain at least one qudit, the total number of sub-instances is $\mathcal{O}(n)$, each composed of $\mathcal{O}(n)$ qudits. For each sub-instance, in steps (2.1)--(2.3), the algorithm iterates over all clock and endpoint qudits and searches over the clauses acting on these qudits. These steps also take $\mathcal{O}(n) \cdot \mathcal{O}(poly(n)) = \mathcal{O}(poly(n))$ time. Steps (2.4) and (2.5) simply iterate over the clauses of a sub-instance, thus requiring $\mathcal{O}(poly(n))$ time. Then, it is apparent that step (2) takes $\mathcal{O}(poly(n))$ time. In step (5), the algorithm once again iterates over all sub-instances and searches over the clauses present in the sub-instance. If we suppose that removing a clause takes $\mathcal{O}(1)$ time, step (5) can be completed in $\mathcal{O}(poly(n))$ time. \\
\indent
Now, let us consider the runtime of the quantum subroutine used in step (6). The input to this quantum subroutine is a TACC of length $T \in \mathcal{O}(n)$ and possibly $\mathcal{O}(poly(n))$ simultaneous propagation clauses acting on qudits $c_t$ and $c_{t+1}$ for any $t \in [T]$. Step (6.1) takes $\mathcal{O}(n)$ time, since $d \in \mathcal{O}(n)$. In (6.2), to progress from time $t$ to $t+1$, the algorithm applies circuit $\mathcal{C}_{\mathcal{G}_8}$ to the state $\ket{0} \otimes \ket{\phi_t}$ $\mathcal{O}(poly(n))$ times. Moreover, after each application, the algorithm restarts (with one less clause) and must again produce $\ket{\phi_t}$ by applying $\mathcal{O}(n)$ unitaries to the initial state $\ket{0}^{\otimes q}$. Thus, the step $t \rightarrow t+1$ takes $\mathcal{O}(poly(n)) \cdot \mathcal{O}(n) = \mathcal{O}(poly(n))$ time. Then, since $T \in \mathcal{O}(n)$, step (6.2) takes $\mathcal{O}(poly(n))$ time. Finally, note that in step (6.3), the algorithm must perform $\mathcal{O}(poly(n))$ measurements. Similar to the previous step, after each measurement, the algorithm restarts (with one less $\Pi_{out}$ clause) and must again create $\ket{\phi_T}$ by applying $T \in \mathcal{O}(n)$ unitaries. This step then takes $\mathcal{O}(poly(n))$ time. This shows that the quantum algorithm then takes $\mathcal{O}(poly(n))$ time, and considering that it may be used $\mathcal{O}(poly(n))$ times, all of step (6) takes $\mathcal{O}(poly(n)) \cdot \mathcal{O}(poly(n))$ time. This concludes that the algorithm is efficient.

\subsubsection{Completeness and soundness} \label{subsubsection:completenesssoundnessmono}

To prove the correctness of the algorithm above, we have to show that for any input instance $x$ of LCT-QSAT, the algorithm presented above can distinguish between the cases where $x$ is a yes-instance or a no-instance. To reiterate, $x$ is a yes-instance if there exists a state $\ket{\psi_{sat}}$ such that for all projectors $\{\Pi_i\}$ that compose the instance, $\Pi_i \ket{\psi_{sat}} = 0$ for all $i$. On the other hand, $x$ is a no-instance if for all possible states, the sum of the constraint violations is not too small, i.e.\ $\sum_{i} \bra{\psi} \Pi_i \ket{\psi} \geq 1/poly(n)$ for any state $\ket{\psi}$. Furthermore, to meet the completeness and soundness conditions associated with $\BQP_1$, we have to show that all yes-instances are always labeled as ``satisfiable'' by the hybrid algorithm, and no-instances are labeled as ``unsatisfiable'' most of the time. \\
\indent
In \cref{subsection:monosat}, we provided an exhaustive list of instances that can decided based on their structure. As the classical algorithm is based on this list, the algorithm never rejects instances with a valid structure and always rejects instances whose structure implies two or more clauses cannot be satisfied simultaneously by any state. If the classical algorithm is able to determine the instance, then both the completeness and soundness conditions are clearly satisfied. However, if the instance is not decided by the classical algorithm, there must be at least one sub-instance that requires the quantum subroutine. \\

\medbreak
\noindent
\textbf{Quantum sub-instances}\\
\noindent
Importantly, the classical algorithm ensures that none of the logical qudits involved in the TACCs of these sub-instances interfere with other sub-instances, and so we can consider them independently. Let $\{S^{(1,2)}_j\}$ denote the set of sub-instances of the first and second type, and $\{S^{(3)}_j\}$ the set of sub-instances of the third type. Here, each $S_j$ is itself a set of clauses that define the sub-instance. \\
\indent
To determine whether the instance is a yes-instance or no-instance, we present the following claim:

\begin{claim}
    If the quantum sub-instances $\{S_j\}$ are satisfiable, the \textnormal{candidate state}
    \begin{equation} \label{eqn:cand}
        \ket{\psi_{cand}} = \bigotimes_{j} \ket{\psi_{phist}}_{S^{(1,2)}_j} \bigotimes_{k} \left( \ket{\psi_{phist}} \otimes \ket{?}_{l_u} \otimes \underbrace{\ket{r \ldots r}}_{L-T} \right)_{S^{(3)}_k},
    \end{equation}
    is always a satisfying state. Here, $\ket{\psi_{phist}}_{S_j}$ is the privileged history state defined in \cref{eqn:phist} where the circuit corresponds to the TACC associated with sub-instance $S_j$ and $l_u$ is the first undefined logical qudit of the TACC.
\end{claim}

\begin{proof}
    Assume that the quantum sub-instances $\{S_j\}$ are satisfiable. Consequently, by \cref{prop:simultaneousprop}, the sub-instances with simultaneous propagation clauses have a single computational history. Then, because the TACC encompasses the whole ACC in the first and second type of sub-instances, they must be uniquely satisfied by a history state. These are the $\ket{\psi_{phist}}_{S^{(1,2)}_j}$ terms. This is not the case for the third type of sub-instance as it may be satisfied either by a truncated history state that leverages the undefined $\Pi_{prop}$ clause or by a history state that encompasses the whole ACC. However, it is not difficult to see that if the latter history state is a satisfying state, then so must the truncated state.\footnote{The reverse implication is not necessarily true because even though the truncated ACC may be satisfiable, considering a longer chain with more constraints can only make the sub-instance harder and possibly unsatisfiable.} These are the $(\ket{\psi_{phist}} \otimes \ket{?}_{l_u} \otimes \ket{r \ldots r})_{S^{(3)}_k}$ terms above.
\end{proof}
\noindent
The claim allows us to show that if the quantum algorithm can determine or verify that this candidate state is not a satisfying state, we can conclude that the instance must be unsatisfiable as no other state can be a satisfying state. \\
\indent
The quantum algorithm presented above makes use of this claim and determines the satisfiability of the instance by indirectly measuring the violation of constraints incurred by the state $\ket{\psi_{cand}}$. In practice however, the algorithm only measures the constraints related to the TACC of the sub-instances, and even then it does not measure the $\Pi_{init}$ or lone $\Pi_{prop}$ projectors. This is because, first, the satisfying state of the clauses beyond the TACC is trivial, and second, because the privileged history states clearly satisfy the $\Pi_{init}$ and lone $\Pi_{prop}$ projectors. This is why upon receiving the TACC of a sub-instance $S_j$, the quantum algorithm only measures the eigenvalue of the simultaneous $\Pi_{prop}$ projectors (if any) and the $\Pi_{out}$ projectors within (if any) on the state $\ket{\psi_{phist}}_{S_j}$. Let $K_j \subset S_j$ and $M_j \subset S_j$ denote these two sets of projectors, respectively. Then, for a projector $\Pi_i \in K_j \cup M_j$, the probability that the algorithm determines the constraint is not satisfied (observes outcome ``1''; equivalent to measuring eigenvalue 1) and rejects the instance is given by

\begin{equation*}
    p_{i,j} = \bra{\psi_{phist}} \Pi_i \ket{\psi_{phist}}_{S_j}.
\end{equation*}
\noindent
Then, the algorithm's acceptance probability $AP$, the probability that none of the measured ancilla qubits
yield outcome ``1'', can be bounded as

\begin{equation}\label{eqn:boundap}
    1 - \sum_{i,j} p_{i,j} \leq AP \leq 1 - \max_{i,j} p_{i,j}.
\end{equation}

\medbreak
\noindent
\textbf{Completeness}\\
\noindent
If the input instance $x$ is a yes-instance, then all of the quantum sub-instances $\{S_j\}$ can be satisfied. For each $S_j$, the quantum algorithm measures the eigenvalues of projectors $K_j$ and $M_j$ on the state $\ket{\psi_{phist}}_{S_j}$. Since these privileged history states satisfy the sub-instances, it follows that $p_{i,j} = 0$ for all $i$ and $j$. Then, from \cref{eqn:boundap}, we can conclude that $AP = 1$. \\
\indent
The state of all qubits that satisfies all clauses of the instance can be written as

\begin{equation*}
    \ket{\psi_{*}} = \ket{\psi_{else}} \otimes \ket{\psi_{cand}},
\end{equation*}
where $\ket{\psi_{else}}$ is the joint satisfying state of the sub-instances that are trivially satisfiable and $\ket{\psi_{cand}}$ is the state described in \cref{eqn:cand}. The algorithm does not explicitly compute $\ket{\psi_{else}}$, but determines that such a state exists, which is sufficient. In other words, $\ket{\psi_{*}} = \ket{\psi_{sat}}$. \\

\medbreak
\noindent
\textbf{Soundness}\\
\noindent
As mentioned above, the quantum algorithm can conclude that the instance is unsatisfiable if the state $\ket{\psi_{cand}}$ receives a penalty greater than $1/poly(n)$. To prove the soundness condition, it is necessary that the algorithm detects this case with probability $\geq 2/3$ (the probability that the circuit accepts the instance is $\leq 1/3$). We now show that this is the case. \\
\indent
The key idea is to leverage the problem's promise. Recall that for no-instances, we are promised that for any $n$-qudit state $\ket{\psi}$, the equation $\sum_{\Pi_i \in x} \bra{\psi} \Pi_i \ket{\psi} \geq 1/poly(n)$ must hold. It then follows that for the state $\ket{\psi_*}$ which satisfies the instance in the yes-case, we now have

\begin{equation} \label{eqn:totpenalty}
    \begin{aligned}
        \sum_{\Pi_i \in x} \bra{\psi_{*}} \Pi_i \ket{\psi_{*}} & = \sum_{\Pi_i \in \{S_j\}} \bra{\psi_{cand}}\Pi_i \ket{\psi_{cand}}                        \\
                                                               & = \sum_{j} \sum_{\Pi_i \in S_j} \bra{\psi_{phist}}\Pi_i \ket{\psi_{phist}}_{S_j}           \\
                                                               & = \sum_{j} \sum_{\Pi_i \in K_j \cup M_j} \bra{\psi_{phist}} \Pi_i \ket{\psi_{phist}}_{S_j} \\
                                                               & \geq 1/poly(n),
    \end{aligned}
\end{equation}
\noindent
where in the first line we made use of the fact that the projectors of the trivially decidable instances were all determined to be satisfiable by the classical algorithm. In the second line, we expanded the expression and grouped some terms together, while in the third line we noted that $\ket{\psi_{phist}}_{S_j}$ always satisfies the $\Pi_{init}$ and the lone $\Pi_{prop}$ projectors within. Then, the maximum violation of a projector in the instance can be bounded as

\begin{equation*}
    \max_{i,j} p_{i,j} \geq \frac{\sum_{j} \sum_{\Pi_i \in K_j \cup M_j} \bra{\psi_{phist}} \Pi_i \ket{\psi_{phist}}_{S_j}}{\textnormal{total \# of projectors}} \geq \frac{1/poly(n)}{\textnormal{total \# of projectors}} \geq 1/poly(n),
\end{equation*}
\noindent
where the second inequality follows from \cref{eqn:totpenalty}, and the last inequality follows from the fact that there are at most a polynomial number of projectors in an instance. Combining this with \cref{eqn:boundap}, we get that $AP \leq 1 - 1/poly(n)$. Finally, as explained below \cref{defn:onesidederr}, since the gap between the acceptance probabilities of yes- and no-instances is given by an inverse polynomial, it is possible to reduce this acceptance probability below $1/3$ and satisfy the soundness condition of $\BQP_1$.

\subsection{Hardness} \label{subsection:hardnessmono}

Here, we show that the problem \textsc{Linear-Clock-Ternary-QSAT} is $\BQP_1$-hard, thus completing the proof of \cref{thm:bqplarge}. In other words, we show that we can encode any instance $x$ of a $\BQP_1$ promise problem $A = (A_{yes}, A_{no})$, into an instance $x'$ of LCT-QSAT, such that any yes-instance (no-instance) of $A$ also becomes a yes-instance (no-instance) of our problem. \\
\indent
Let $U_x = U_L \ldots U_1$ with $U_i \in \mathcal{G}_8$ and $L = poly(n)$ be the $\BQP_1$ circuit where given an instance $x$ of problem $A$, decides whether $x \in A_{yes}$ or $x \in A_{no}$. The input to the circuit is the $q$-qubit ancilla register $\ket{0}^{\otimes q}$, where $q = poly(n)$. Let $ans$ be the qubit that when measured provides this decision. The reduction $x \mapsto x'$ requires $q + (L+1) + 2$ qudits, each of dimension $17$, and is the following:

\begin{enumerate}
    \item Choose $q$ qudits to serve as the logical qudits of the computation, $L+1$ qudits as the clock qudits, and $2$ qudits as the endpoint qudits. Define an ordering for each type of qudit.
    \item For all $i \in [q]$, create a $\Pi_{init}$ clause acting on $c_0$, $e_0$ and $l_i$.
    \item For all $j \in [L]$, create a $\Pi_{prop}$ clause with associated unitary $U_{j+1}$, which acts on clock qudits $c_j$ and $c_{j+1}$ (with $c_j$ being the predecessor of $c_{j+1}$) and the two logical qudits that the unitary originally acts on.
    \item Create one $\Pi_{out}$ clause acting on $c_L$, $e_1$ and $l_{ans}$.
\end{enumerate}
\indent
The resulting instance $x'$ consists of a one-dimensional clock chain of length $L$ ($L+1$ clock qudits) with a unique direction, no simultaneous propagation clauses, and fully initialized logical qudits (see \cref{fig:qinstance}). \\
\indent
To prove the correctness of the reduction, let us first consider the case where $x \in A_{yes}$. Here, $x'$ is satisfied by the history state

\begin{equation} \label{eqn:histhardness}
    \ket{\psi_{hist}} = \frac{1}{\sqrt{L+1}} \sum_{t=0}^L U_t \ldots U_0 \ket{0}^{\otimes q} \otimes \ket{\underbrace{d \ldots d}_t a_t \underbrace{r \ldots r}_{L-t}},
\end{equation}
\noindent
where $U_0 = I$.\footnote{More explicitly, the satisfying state is given by $\ket{\psi_{hist}} \otimes \ket{\psi_{aux}}$ where the second state accounts for the Bell pairs formed in the auxiliary subspaces of qudits.} As mentioned previously, this history state immediately satisfies the $\Pi_{init}$ and $\Pi_{prop}$ clauses of the instance, while the satisfiability of the $\Pi_{out}$ clause depends on whether the evaluation of $U$ on the initial state results in qudit $l_{ans}$ being in the state $\ket{1}$. This is true since, by assumption, $U$ accepts $x$ with certainty. In other words, in the original instance, $ans$ always produces outcome ``1'', implying that in our instance $l_{ans}$ always results in the state $\ket{1}$. \\
\indent
When $x \in A_{no}$, we have to demonstrate that any state of the qudits in the instance incurs a penalty greater than an inverse polynomial. Unfortunately, this is more complex that in the previous case, and so it requires a more in-depth analysis.

\subsubsection{No small eigenvalues for no-instances}

In this section it will be more convenient to think of the $O_{init}$, $O_{prop}$, $O_{out}$ operators as positive semi-definite operators (so as defined in \cref{eqn:Oinit,eqn:Oprop,eqn:Oout}) instead of projectors. This does not change in any way the task we have to prove. In fact, this makes the problem ``harder'' as the eigenvalues of excited states may now be closer to $0$ instead of all of them being $+1$-eigenvalues. \\
\indent
As in \cref{subsubsection:qmahard}, it is useful to group the similar types of operators in the LCT-QSAT instance in terms of Hamiltonians. The instance can then be expressed as

\begin{equation} \label{eqn:hamilsoundness}
    H = H_{roles} + H_{aux} + H_{clock} + H_{prop} + H_{init} + H_{out},
\end{equation}
\noindent
where $H_{roles}$ collects all of the $\Pi_E, \Pi_C$ and $\Pi_L$ terms, $H_{prop}$ the $\Pi_{work,U}(\Pi_D \otimes \Pi_D \otimes I^{\otimes 2})_{j,j+1}$ terms ($L$ of them), $H_{init}$ the $(I - \ket{0} \! \bra{0})_{l_i} \otimes \ket{a} \! \bra{a}_{c_0}$ terms ($q$ of them), and $H_{out}$ the single $\ket{0} \! \bra{0}_{l_{ans}} \otimes \ket{a} \! \bra{a}_{c_L}$ term. Additionally,

\begin{equation*}
    H_{aux} := (I - \ket{\Phi^+} \! \bra{\Phi^+}_{EC,C\!A})_{e_0,c_0} + (I - \ket{\Phi^+} \! \bra{\Phi^+}_{C\!B,EC})_{c_L,e_1} + \sum_{j=0}^L (I - \ket{\Phi^+} \! \bra{\Phi^+}_{C\!A, C\!B})_{c_j,c_{j+1}}
\end{equation*}
\noindent
collects the terms that require the auxiliary subspaces to form $\ket{\Phi^+}$ Bell pairs, and

\begin{equation} \label{eqn:hclock}
    \begin{aligned}
        H_{clock} & := \sum_{j = 0}^L \left[ \Pi_D \otimes \Pi_D \otimes \Pi_{clock,D} + (I^{\otimes 2} \otimes \Pi_{clock,?}) (I - \Pi_D \otimes \Pi_D \otimes I^{\otimes 2})\right]_{i,i+1} \\
                  & + (\underbrace{\Pi_{start} + \ldots + \Pi_{start}}_{n})_{0} + \: (\Pi_{stop})_{L},
    \end{aligned}
\end{equation}
\noindent
collects the $L$ terms that constrain the states of the clock based on the state of the logical qudits, the $\Pi_{start}$ projectors, and the single $\Pi_{stop}$ projector. Since $H_{clock}$ only acts on the clock qudits of the instance, we omit the $c$ in the indices to avoid clutter. The Hamiltonian of \cref{eqn:hamilsoundness} acts on the Hilbert space $\mathcal{H} = (\mathbb{C}^{17})^{\otimes q + L + 3}$. \\
\indent
To show that the Hamiltonian of \cref{eqn:hamilsoundness} has no small eigenvalues, we will partition the Hilbert space into ``good'' and ``bad'' subspaces and demonstrate that in either subspace, the smallest non-zero eigenvalue is always greater than an inverse polynomial. We will begin by considering $\mathcal{S}_{roles}$ and $\mathcal{S}_{roles}^\perp$ which represent the null space of $H_{roles}$ and its orthogonal subspace, respectively. Clearly, $\mathcal{S}_{roles} \subset \mathcal{H}$, $\mathcal{S}_{roles}^\perp \subset \mathcal{H}$, and $\mathcal{S}_{roles} \oplus \mathcal{S}_{roles}^\perp = \mathcal{H}$. We will show that states in the orthogonal subspace (where qudits are not labeled properly) receive a constant ``energy'' penalty, whereas for the actual null space, we will apply the same logic and partition it into two subspaces, $S_{aux}$ and $S_{aux}^\perp$, which we analyze in a similar way. We will repeat this process a third time (for $H_{clock}$), after which we will observe that when restricted to the ``good'' subspaces, the Hamiltonian of \cref{eqn:hamilsoundness} is effectively identical to the one of Ref.\ \cite{kitaev2002classical} and simply rely on their result.\footnote{The clock encoding is different, but it does not affect the analysis whatsoever.} Observe that the order in which we consider these subspaces is relevant since, for example, we cannot conclude anything about the clock register before first ensuring that this takes place within the subspace where all arrows between clock qudits point in a unique direction. \cref{fig:projection} summarizes these steps.

\begin{figure}[t]
    \centering
    \includegraphics[width=\textwidth]{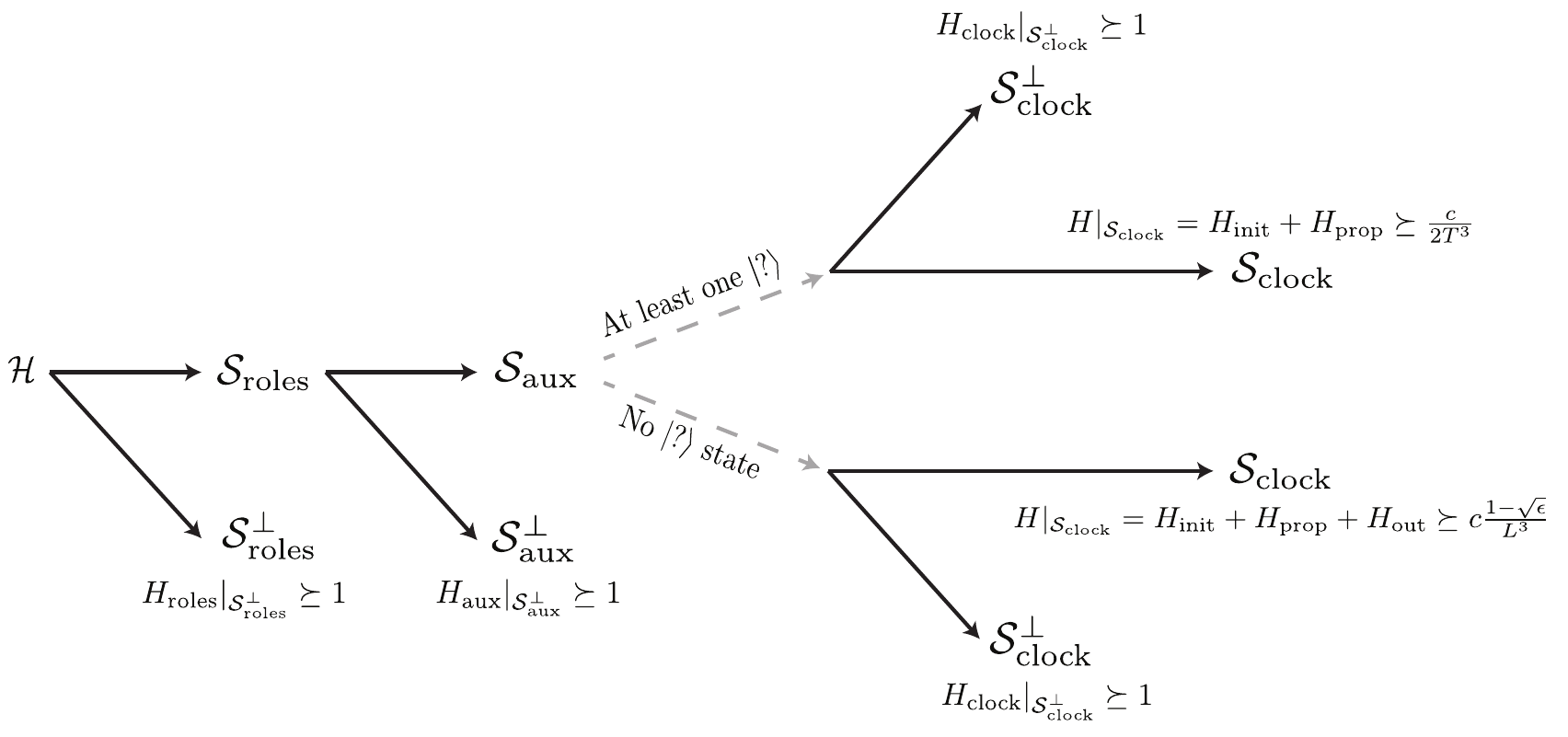}
    \caption{Diagram summarizing the steps of the proof showing that the Hamiltonian of \cref{eqn:hamilsoundness} has no small eigenvalues. Starting by considering the full Hilbert space $\mathcal{H}$, we partition the space into two (represented by the forking arrows): the null space corresponding to a term in the Hamiltonian, and the orthogonal subspace. Underneath each subspace, we write the smallest eigenvalue of the Hamiltonian term restricted to that subspace. The diagram shows that in all subspaces, there can be no state that receives a small energy penalty. Another important property is that all subspaces are invariant under the action of the Hamiltonian.}
    \label{fig:projection}
\end{figure}
\indent
To begin, observe that $\mathcal{S}_{roles}$ is the product of subspaces where each subspace corresponds to the correct assignment of the particle. For the instance above, we expect

\begin{equation*}
    \mathcal{S}_{roles} = (\mathbb{C}^3)^{\otimes q} \otimes (\mathbb{C}^{12})^{\otimes{L+1}} \otimes (\mathbb{C}^2)^{\otimes 2}.
\end{equation*}
\noindent
The Hamiltonian of \cref{eqn:hamilsoundness} (and the projectors within) does not convert particles of one type to a different type, so both $\mathcal{S}_{roles}$ and $S_{roles}^\perp$ are invariant under its action. Any state in $S_{roles}^\perp$ has at least one particle that does not respect its assigned role and hence violates one of the $\Pi_E$, $\Pi_L$, or $\Pi_C$ clauses within $H_{roles}$ and receives a penalty $\geq 1$.\\
\indent
Next, we consider $H_{aux}$ restricted to the null space $\mathcal{S}_{roles}$, which we denote by $H_{aux}|_{\mathcal{S}_{roles}}$. The null space of this Hamiltonian is

\begin{equation*}
    \mathcal{S}_{aux} = (\mathbb{C}^3)^{\otimes q} \otimes (\mathbb{C}^3)^{\otimes L+1} \otimes \ket{\psi_{aux}},
\end{equation*}
\noindent
where $\ket{\psi_{aux}}$ is the state giving the linear structure to the clock qudits. To reiterate, here, the $C\!B$ subspace of each clock qudit $c_1$ to $c_{L-1}$ forms a $\ket{\Phi^+}$ Bell pair with the $C\!A$ subspace of its successor. The clock qudits at the endpoints behave differently. The $C\!A$ subspace of $c_0$ forms a Bell pair with the $EC$ subspace of the first endpoint particle, while the $C\!B$ subspace of $c_L$ forms a Bell pair with the $EC$ subspace of the second endpoint particle. This subspace and the orthogonal subspace are invariant under the action of $H$. Then, observe that since the four Bell states form a basis of the $2$-qubit space, the orthogonal subspace is spanned by the tensor product of any combination of the other three states. It is then clear that any state in $\mathcal{S}_{aux}^\perp$ violates at least one of the projectors within $H_{aux}$, and thus receives a penalty $ \geq 1$. \\
\indent
We can now consider $H_{clock}|_{\mathcal{S}_{aux}}$. In its current form, this Hamiltonian is difficult to analyze since the clock qudits are intertwined with the logical qudits. Indeed, depending on the state of the logical qudits, the terms within might result in either $\Pi_{clock,D}$ or $\Pi_{clock,?}$, which enforce slightly different constraints. To simplify the analysis, we partition the space into two subspaces: one where there are no logical qudits in the undefined state, and the other subspace where at least one of them is in this state. Importantly, these subspaces are invariant under the action of $H$ since there is no transformation that takes a logical qudit in the state $\ket{?}$ to the state $\ket{0}$ or $\ket{1}$, and vice versa. Consequently, any logical qudit in the state $\ket{?}$, remains in this state. \\

\medbreak
\noindent
\textbf{At least one undefined logical qudit} \\
\noindent
Let $O_{prop}^{(c_T,c_{T+1})}$ with $0 \leq T < L$ be the first $O_{prop}$ clause of the instance that acts on an undefined logical qudit. Let us refer to this qudit as $l_u$ with $0 \leq u < q$, and for simplicity, let us assume that this is the only undefined qudit of the instance as any additional such qudits can only increase the energy further. In this case, the clock Hamiltonian of \cref{eqn:hclock} becomes

\begin{equation*}
    \begin{aligned}
        H_{clock}(T,u) & = \underbrace{\Pi_{start} + \ldots + \Pi_{start}}_{q} + \: \Pi_{stop} + \sum_{i=0}^{T-1} (\Pi_{clock,D})_{i,i+1} + \: (\Pi_{clock,?})_{T,T+1}                               \\
                       & + \sum_{j = T+1}^L \left[ \Pi_D \otimes \Pi_D \otimes \Pi_{clock,D} + (I^{\otimes 2} \otimes \Pi_{clock,?}) (I - \Pi_D \otimes \Pi_D \otimes I^{\otimes 2})\right]_{j,j+1},
    \end{aligned}
\end{equation*}
\noindent
where we have not assumed a specific form of the clock projectors that act on clock qudits $T+1$ to $L$. In fact, the exact form of these clauses is not of much relevance. Indeed, observe that as shown in \cref{eqn:nullspacehprop}, the clock qudit $c_{T+1}$ can only be in the state $\ket{r}$ in order to satisfy $\Pi_{clock,?}$. Then, regardless of whether the next clause is $\Pi_{clock,D}$ or $\Pi_{clock,?}$, both clauses require that $c_{T+2}$ must also be fixed to $\ket{r}$ in order not to incur a penalty. The same is true for all remaining clauses. \\
\indent
The null space and orthogonal subspace of this Hamiltonian are

\begin{equation*}
    \begin{aligned}
        \mathcal{S}_{clock}(T, u)       & = \mathcal{D}_u \otimes \Span\{ \ket{C_0}, \ldots, \ket{C_T}\},                         \\
        \mathcal{S}_{clock}^\perp(T, u) & =  \mathcal{D}_u \otimes \Span(\mathcal{I} \cup \{ \ket{C_{T+1}}, \ldots, \ket{C_L}\}),
    \end{aligned}
\end{equation*}
\noindent
where $\mathcal{I}$ represents the set of illegal clock states and $\mathcal{D}_u := (\mathbb{C}^2)^{\otimes u} \otimes \ket{?}_{u} \otimes (\mathbb{C}^2)^{\otimes q - u - 1}$ is the space of logical qudits considering that the undefined logical qudit is at position $u$. Both null spaces are invariant under the action of $H$ as there are no terms that transform legal states to illegal ones, or terms that force the clock to progress past $\ket{C_T}$ (as shown below in \cref{eqn:undefinedhprop}, there is no $(T+1)$-th $\Pi_{work}$ clause within $H_{prop}$). We can now conclude that any state in $\mathcal{S}_{clock}^\perp$ violates at least one of the $\Pi_{clock,D}$ or $\Pi_{clock,?}$ terms and thus receives a penalty $\geq 1$. \\
\indent
Now, we can consider $H$ restricted to the null space of $S_{clock}(T,u)$. Within this subspace, $H$ simplifies to

\begin{equation} \label{eqn:undefinedhprop}
    H|_{\mathcal{S}_{clock}(T,u)} = H_{init} + H_{prop}(T,u); \hspace{2.5em} H_{prop}(T,u) := \sum_{i=0}^{T-1} (\Pi_{work})_{i,i+1},
\end{equation}
\noindent
where we again made use of the fact that $H_{prop}$ also depends on the undefined logical qudit, and the fact that $H_{out}$ is projected out regardless of the position of the first undefined logical qudit. Here, $H_{prop}(T,u)$ is the Hamiltonian often seen in literature \cite{kitaev2002classical,bravyi2006efficient,stoquasticLH}, which demands the application of the unitary circuit $U_T \ldots U_1$ on a data register. \\
\indent
We can now bound the smallest eigenvalue of $H_{init} + H_{prop}(T,u)$ by using Kitaev's Geometric Lemma stated in \cref{lemma:geometric}. As done in Ref.\ \cite{kitaev2002classical}, the smallest non-zero eigenvalues of these operators are found by performing a change of basis to the ``rotating'' frame. Under the transformation, $H_{init}$ remains unchanged, while $H_{prop}$ takes the shape of a $(T+1) \times (T+1)$ tridiagonal matrix. The smallest non-zero eigenvalues are $\gamma(H_{init}) = 1$ and $\gamma(H_{prop}) = c/T^2$, with

\begin{equation*}
    \mathcal{S}_{init} = \mathcal{D}_u \otimes \Span\{\ket{C_1}, \ldots, \ket{C_T}\} \hspace{2em} \textnormal{and} \hspace{2em} \mathcal{S}_{prop} = \mathcal{D}_u \otimes \frac{1}{\sqrt{T+1}} \sum_{t=0}^T \ket{C_t}
\end{equation*}
\noindent
as their respective null spaces.\footnote{Note that the term $\ket{0}^{\otimes q} \otimes \ket{C_0} \! \bra{C_0}$, which is normally in the null space of $H_{init}$, is no longer in this space since one of the logical qudits is assumed to be in the state $\ket{?}$.} Given the simple form of these null spaces, it is straightforward to show that in the Geometric Lemma of \cref{eqn:geometric},

\begin{equation*}
    \alpha = \max_{\ket{\eta} \in \mathcal{S}_{prop}} \bra{\eta} \Pi_{\mathcal{S}_{init}} \ket{\eta} = \frac{T}{T+1} = 1 - \frac{1}{T+1},
\end{equation*}
\noindent
and so

\begin{equation*}
    H|_{\mathcal{S}_{clock}(T,u)} \succeq \min\{ \gamma(H_{init}),\gamma(H_{prop}) \} \frac{(1-\alpha)}{2} \succeq \frac{c}{2T^3}.
\end{equation*}
\noindent
This concludes that the smallest eigenvalue of $H$ scales as $\Omega (1/T^3)$. \\

\medbreak
\noindent
\textbf{No undefined logical qudit} \\
\noindent
In this case, the clock Hamiltonian becomes

\begin{equation*}
    H_{clock} = \underbrace{\Pi_{start} + \ldots + \Pi_{start}}_{q} + \: \Pi_{stop} + \sum_{i=0}^L (\Pi_{clock,D})_{i,i+1},
\end{equation*}
\noindent
with null space and orthogonal subspace given by
\begin{equation*}
    \mathcal{S}_{clock} = (\mathbb{C}^2)^{\otimes q} \otimes \Span\{ \ket{C_0}, \ldots, \ket{C_L} \} \hspace{2em} \textnormal{and} \hspace{2em}  \mathcal{S}_{clock}^\perp = (\mathbb{C}^2)^{\otimes q} \otimes \Span(\mathcal{I}).
\end{equation*}
\noindent
As mentioned previously, these subspaces are invariant under $H$ as there is no term that transforms a legal clock state into an illegal state, or vice versa. Any state in $\mathcal{S}_{clock}^\perp$ violates at least one of the clauses within $H_{clock}$ and thus receives a penalty $\geq 1$. \\
\indent
By restricting the Hamiltonian of \cref{eqn:hamilsoundness} to $\mathcal{S}_{clock}$, we obtain that $H|_{\mathcal{S}_{clock}} = H_{init} + H_{prop} + H_{out}$ where $H_{prop} = \sum_{i=0}^L (\Pi_{work})_{i,i+1}$. Letting $H_1 := H_{init} + H_{prop}$ and $H_2 := H_{prop}$ and again moving to the rotating frame, we can write the null spaces of these terms as

\begin{equation*}
    \mathcal{S}_1 = \Big( \ket{0}^{\otimes q} \otimes \ket{C_0} \Big) \oplus
    \Big( (\mathbb{C}^2)^{\otimes q} \otimes \Span\{\ket{C_1}, \ldots \ket{C_{L-1}}\} \Big) \oplus
    \Big( U^\dagger(\ket{1}_{ans} \otimes (\mathbb{C}^2)^{\otimes q-1}) \otimes \ket{C_L} \Big)
\end{equation*}
\noindent
and

\begin{equation*}
    \mathcal{S}_2 = (\mathbb{C}^2)^{\otimes q} \otimes \frac{1}{\sqrt{L+1}} \sum_{t = 0}^L \ket{C_t}.
\end{equation*}
\noindent
Except for the different clock encoding and the absence of the subspace for witness qudits, these Hamiltonians and their corresponding null spaces are similar to those of Ref.\ \cite{kitaev2002classical}. By following the same steps as in that proof (or alternatively pages $236$-$239$ in Ref.\ \cite{gharibian2015quantum}), we can obtain that the lowest energy penalty a state may receive is $\geq c(1 - \sqrt{\epsilon}) L^{-3}$ where $c$ is a positive constant and $\epsilon \leq 1/3$ is the probability that the verifier of problem $A$ wrongfully accepts. \\
\indent
The soundness analysis concludes that in the case where there is an undefined logical qudit, $H \succeq c(2T^3)^{-1}$, and $H \succeq c(1 - \sqrt{\epsilon}) L^{-3}$ when there are none. Since $T \leq L = poly(n)$, we have that $\sum_i \bra{\psi} \Pi_i \ket{\psi} \geq 1/poly(n)$ for any state $\ket{\psi}$ on $q + L + 3$ 17-dimensional qudits, and so $x'$ is a no-instance of LCT-QSAT.

\section{Reducing the Qudit Dimensionality} \label{section:our}

In this section, we remove the use of monogamy of entanglement and show that the resulting problem, now referred to as \textsc{Semilinear-Clock-Ternary-QSAT}, is $\BQP_1^{\mathcal{G}_8}$-complete and uses lower-dimensional qudits.

We demonstrate that even without the monogamy condition, we can obtain similar results about the satisfiability of instances. We show that the rules of propagation, the choice of clock encoding, and the requirement to keep a consistent state of the clock register at all times suffice to demonstrate that TACCs with forks or $\Pi_{prop}$ clauses pointing in different directions are unsatisfiable. Additionally, we also show that $\Pi_{init}$ and $\Pi_{out}$ still serve as endpoints, albeit as a weaker version, and thus there is no need for endpoint qudits. Together, these results show that while the use of monogamy in the construction does facilitate some proofs, it is not crucial for the construction. Removing this condition and the use of endpoint qudits allows us to reduce the local dimension from $17$ down to $6$. The main difference in this construction is that sub-instances with multiple ACCs may now have a satisfying state. These sub-instances take the form of one-dimensional TACCs joined together by complex webs of connections between clock qudits, thus explaining the term \textit{semilinear} in the name of the problem (see \cref{fig:semilinearsat}). We show that if these sub-instances are potentially satisfiable, they can be further separated into smaller instances composed only of the TACCs, and the clauses connecting them can be satisfied trivially.

The problem is defined similarly to \cref{defn:lctqsat}, except that the terms $\Pi_E$ and $I - \ket{\Phi^+} \! \bra{\Phi^+}$ are omitted in the definitions of $\Pi_{init}$, $\Pi_{out}$, and $\Pi_{prop,U}$.

\begin{defn}[Semilinear-Clock-Ternary-QSAT] \label{defn:slctqsat}
  The problem {\rm \textsc{Semilinear-Clock-Ternary-QSAT}} is a quantum constraint satisfaction problem defined on the $6$-dimensional Hilbert space

  \begin{equation*}
    \mathcal{H} = \Span\{\ket{0}, \ket{1}, \ket{?}\} \oplus \Span\{\ket{r}, \ket{a}, \ket{d}\},
  \end{equation*}
  \noindent
  consisting of a logical and a clock subspace. The problem consists of $5$ types of clauses: the initialization clause $\Pi_{init}$, the propagation clauses $\Pi_{prop,U}$ (one for every $U \in \mathcal{G} = \{H,HT, (H \otimes H) {\rm CNOT}\}$), and a termination clause $\Pi_{out}$, each acting on at most four $6$-dimensional qudits. These are defined as

  \begin{equation} \label{eqn:Oinitfree}
    O_{init} := \; \; \Pi_{L,1} + \Pi_{C,2} + \Pi_{start, 2} + (I - \ket{0} \! \bra{0})_1 \otimes \ket{a} \! \bra{a}_2,
  \end{equation}

  \begin{equation} \label{eqn:Opropfree}
    \begin{aligned}
      O_{prop,U} := \; \; & \Pi_{L_1} + \Pi_{L,2} + \Pi_{C,3} + \Pi_{C,4}                                                    \\
                          & +  (\Pi_{work} + I^{\otimes 2} \otimes \Pi_{clock,D})(\Pi_D \otimes \Pi_D \otimes I^{\otimes 2}) \\ &+ (I^{\otimes 2} \otimes \Pi_{clock,?})(I - \Pi_D \otimes \Pi_D \otimes I^{\otimes 2}),
    \end{aligned}
  \end{equation}
  \noindent
  and
  \begin{equation} \label{eqn:Ooutfree}
    O_{out} :=    \; \; \Pi_{L,1} + \Pi_{C,2} + \Pi_{stop,2} + \ket{0} \! \bra{0}_1 \otimes \ket{a} \! \bra{a}_2,
  \end{equation}
  \noindent
  where again $\Pi_{init}, \Pi_{prop,U}$, and $\Pi_{out}$ are the projectors with the same nullspace as the one defined by the positive semi-definite operators $O_{init}$, $O_{prop,U}$, and $O_{out}$. The projectors found in these clauses and the rest of the definition is the same as shown in \cref{defn:lctqsat}. Evidently, $\Pi_{init}$, $\Pi_{out}$ and $\Pi_{prop,U}$ are still at most $4$-local clauses.
\end{defn}
\indent
The null spaces of the operators $O_{init}$, $O_{prop,U}$, and $O_{out}$ in the definition above are similar to those shown in \cref{eqn:initsat,eqn:outsat,eqn:propsat}, with the exception that there is no longer a $\ket{\psi_{aux}}$ to represent the Bell pairs formed in the auxiliary subspaces.

\begin{figure}[t]
  \centering
  \includegraphics[width=0.9\textwidth]{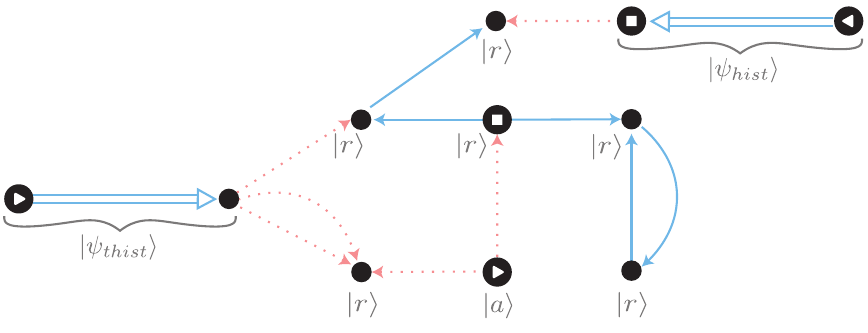}
  \caption{A clock component of a SLCT-QSAT sub-instance that may be satisfiable. The long double-lined arrows represent non-zero length TACCs where all $\Pi_{prop}$ clauses within are well-defined and have the indicated direction. They may still contain simultaneous propagation clauses. The clock component hence has two TACCs plus an additional active dot at the bottom of the image (the qudit with the ``start'' icon present in two undefined $\Pi_{prop}$ clauses). The TACC on the left of the image is truncated and so the clauses within this chain are trivially satisfiable by a truncated history state $\ket{\psi_{thist}}$, while the TACC on the right requires a quantum computer to determine whether a history state $\ket{\psi_{hist}}$ satisfies the $\Pi_{out}$ clause. The clauses involving the active dot can be satisfied by fixing this clock qudit to $\ket{a}$, the adjacent clock qudits to $\ket{r}$, and the logical qudits present in the $\Pi_{init}$ clauses to $\ket{0}$. The remaining clauses can be satisfied by setting the remaining clock qudits to $\ket{r}$, leaving their logical qudits unconstrained. To conclude that the instance is satisfiable, it is also necessary to verify that the logical qudits involved in the three TACCs do not conflict with each other.}
  \label{fig:semilinearsat}
\end{figure}

\subsection{Instance satisfiability} \label{subsection:satisfiability}

Since the definition of the clauses is similar to those of \cref{defn:lctqsat}, and the locality is also identical, there is no new phenomena or new structures of instances to consider. Therefore, let us begin with the analysis of instances. The proofs of the lemmas and propositions shown here can be found in \cref{appendix:nomonogamy}.

The first lemma dictating the satisfiability of the instance is the same as before, except now there are no clock endpoints.

\begin{restatable}[Single-type qudits]{lemma}{singletypequdit} \label{lemma:singletypequdit}
  If a $6$-dimensional qudit in the instance is acted on by both a $\Pi_C$ and $\Pi_L$ clause, the instance is unsatisfiable.
\end{restatable}
\noindent
As before, we now consider the clock components of the sub-instance. This time however, the clock components are simply the connected components defined solely by clock qudits, and not both clock qudits and endpoint qudits.

\subsubsection{A clock component and its logical qudits} \label{subsubsection:clockfree}

In \cref{subsection:monosat}, the need to adhere to the principle of monogamy of entanglement provided the first criteria for determining the satisfiability of instances and established the unsatisfiability of a large number of instances. Indeed, \cref{lemma:linearchainmono,lemma:directionmono,lemma:uniqueclockmono,lemma:uniqueendpointmono} determined that clock components with forks, badly oriented $\Pi_{prop}$ clauses, or misplaced $\Pi_{init}$ and $\Pi_{prop}$ were unsatisfiable. In this construction, we can no longer make such statements. However, we find that these components can still be decided, and are in fact satisfiable provided the instance does not have an active clock chain.

\begin{restatable}[Lack of an ACC]{lemma}{noendpoints} \label{lemma:noendpoints}
  If a clock component does not have at least one $\Pi_{init}$ and one $\Pi_{out}$ clause, the sub-instance is satisfiable.
\end{restatable}
\noindent
Essentially, because there are no constraints demanding that monogamy of entanglement is preserved, we have skipped \cref{lemma:linearchainmono,lemma:directionmono,lemma:uniqueclockmono,lemma:uniqueendpointmono} and jumped directly to the analog of \cref{lemma:noendpointsmono}. We remark that although the statement here is the same as in \cref{lemma:noendpointsmono}, this lemma decides a much larger class of sub-instances.

Sub-instances that are not decided by the previous lemma must have at least one $\Pi_{init}$ and one $\Pi_{out}$ clause, forming an ACC of some length $L$ with $0 \leq L < \mathcal{O}(n)$. In the previous construction, the monogamy of entanglement condition and endpoint qudits straightforwardly dictated that potentially satisfiable sub-instances must have exactly one ACC. Here, we do not have such a luxury. Instead, we have to evaluate the satisfiability of more complex clock components with potentially multiple ACCs. Ultimately, we will see that the satisfiability of the instance reduces to determining the satisfiability of TACCs.

At the moment, all we can gather is that a satisfying state of the clauses of the chain requires one of the clock qudits to be in the state $\ket{a}$. However, it is unknown whether this is enough to trigger the propagation of computation as it is also unknown which clauses can be well-defined or undefined. To determine this, we show a more powerful generalization of \cref{prop:activetaccmono}.

\begin{restatable}[Initialization of logical qudits]{prop}{activetacc} \label{prop:activetacc}
  Let $c_0, \ldots c_L$ with $L \geq 0$ be a chain of clock qudits where $c_0$ is a qudit present in a $\Pi_{init}$ clause, and $c_L$ is a qudit present in a $\Pi_{out}$ clause. Suppose the $\Pi_{prop}$ clauses of the chain point in arbitrary directions and suppose that there might be other $\Pi_{init}$ clauses acting on other clock qudits within the chain (i.e.\ $c_0, \ldots, c_L$ need not be an ACC). A state where $c_0$ is not $\ket{a}$ in any basis state of the superposition is not a satisfying state.
\end{restatable}

\noindent
This proposition allows us to show that in any possible satisfying state, all of the clock qudits of the component acted on by $\Pi_{init}$ must be allowed to be $\ket{a}$ at some point in time. This is because for any such clock qudit ($c_0$), there exists a path that connects it to another clock qudit with a $\Pi_{out}$ clause. And while this path may contain other clock qudits present in $\Pi_{init}$ clauses, the proposition assures us that $c_0$ must be active at some point in time. Consequently, since this proposition applies to all clock components of the instance, we can identify for the whole instance which $\Pi_{prop}$ clauses may be well-defined or undefined. In the upcoming algorithm, this reflected in step (1.2).

Then, we can consider the following two lemmas:

\begin{restatable}[$\Pi_{prop}$ clauses acting on $c_0$ point outward]{lemma}{hinitdirections} \label{lemma:hinitdirections}
  Consider a clock qudit present in a $\Pi_{init}$ clause of a clock component. If this qudit is also present in a $\Pi_{prop}$ clause (either well-defined or undefined) that points towards it, the instance is unsatisfiable.
\end{restatable}

\begin{restatable}[One well-defined $\Pi_{prop}$ clause maximum]{lemma}{singlewelldefined} \label{lemma:singlewelldefined}
  Consider a clock qudit present in a $\Pi_{init}$ clause and a well-defined $\Pi_{prop}$ clause. If this qudit is also present in another $\Pi_{prop}$ clause (either well-defined or undefined) that connects to a new clock qudit, the instance is unsatisfiable.
\end{restatable}
\noindent
These lemmas show that if the sub-instance has a potential satisfying state, the $\Pi_{prop}$ clauses that act on this clock qudit must fall into one of the following two scenarios: either all of the $\Pi_{prop}$ clauses are undefined and point away from this qudit, or there is a single well-defined $\Pi_{prop}$ clause that points away from it. These two cases are illustrated in \cref{fig:allowedinits}. Now observe that since the clock component is fully connected, at least one of these paths leads to a clock qudit with a $\Pi_{out}$ clause. Therefore, in the first scenario, the clock qudit forms an active dot. In the second scenario, the qudit forms part of a single chain of well-defined clauses that either terminates at some time $T$ due to an undefined $\Pi_{prop}$ clause, or reaches the clock qudit with the $\Pi_{out}$ clause, in which case $T = L$. Technically, this is not a TACC as some clock qudits within the chain might contain other $\Pi_{init}$ clauses; however, as will be discussed in \cref{lemma:outsidetacc}, this event makes the instance unsatisfiable. For now, let us assume that the chain is an actual TACC and come back to this issue at a later time.

The case where the clock qudit forms an active dot can be satisfied easily. We simply set its state to $\ket{a}$ and initialize all logical qudits present in its $\Pi_{prop}$ clause to $\ket{0}$. Arguing about the satisfiability of the (more general) TACCs requires more work.

\begin{figure}[t]
  \centering
  \begin{subfigure}[b]{0.2\textwidth}
    \centering
    \includegraphics{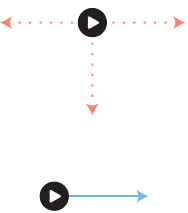}
    \caption{}
    \label{fig:allowedinits}
  \end{subfigure} \hspace{0.05\textwidth}
  \begin{subfigure}[b]{0.2\textwidth}
    \centering
    \includegraphics{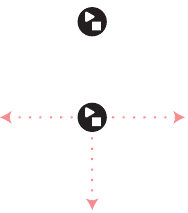}
    \caption{}
    \label{fig:alloweddots}
  \end{subfigure} \hspace{0.05\textwidth}
  \begin{subfigure}[b]{0.42\textwidth}
    \centering
    \includegraphics{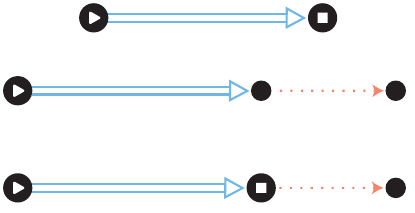}
    \caption{}
    \label{fig:allowedtaccs}
  \end{subfigure}
  \caption{(a) Examples of the only types of $\Pi_{prop}$ clauses acting on a clock qudit with a $\Pi_{init}$ clause that may be satisfiable. We assume that the clock qudit is not present in a $\Pi_{out}$ clause. Top: one or many undefined $\Pi_{prop}$ clauses that point away from the qudit. Bottom: at most one well-defined $\Pi_{prop}$ clause that points away from the qudit. (b) Example of the only type of active dots that may be satisfiable. Top: an active dot with no $\Pi_{prop}$ clauses. Bottom: an active dot with one or many undefined $\Pi_{prop}$ clauses that point away from the qudit. (c) Examples of the only types of non-zero length TACCs that may be satisfiable. Top: a TACC that encompasses the whole ACC with no clauses past the $\Pi_{out}$ clause. Middle: a truncated TACC. Bottom: a TACC that encompasses the whole ACC but has clauses past the $\Pi_{out}$ clause. For the last two cases, the qudit outside the TACC is allowed to be present in $\Pi_{out}$ clauses.}
  \label{fig:sattacc}
\end{figure}

\medbreak
\noindent
\textbf{Satisfiability of a TACC and its neighboring clauses}\\
\noindent
\cref{prop:activetacc} shows that for active clock chains (and their generalization that includes $\Pi_{init}$ clauses in between the chain) $c_0$ must be $\ket{a}$ at some point in time. Since this qudit is always part of a TACC, we have that the TACC  must contain at least one active spot. We now show that these chains must contain exactly one such state at any given time, and use this to show that potentially satisfiable TACCs must be one-dimensional chains with a unique direction.

\begin{restatable}[Unique active state in a TACC]{prop}{uniqueactivetacc} \label{prop:uniqueactivetacc}
  Let $c_0, \ldots, c_T$ with $0 < T \leq L$ be the clock qudits of a TACC with non-zero length. A state in which at any given time there are two clock qudits $c_a$ and $c_b$ with $0 \leq a < b \leq T$ in the state $\ket{a}$ is not a satisfying state.
\end{restatable}
\indent
We now demonstrate that potentially satisfiable TACCs must be one-dimensional chains where all $\Pi_{prop}$ clauses of the chain point away from the clock qudit with the $\Pi_{init}$ clause.

\begin{restatable}[No forks in a TACC]{lemma}{linearchain} \label{lemma:linearchain}
  Let $c_0, \ldots, c_T$ with $1 < T \leq L$ be the clock qudits of a TACC with length at least two, and let $c_j$ be a clock qudit that is not at the ends of the chain, i.e.\ $0 < j < T$. If $c_j$ is connected to more than two other clock qudits, the instance is unsatisfiable.
\end{restatable}

\begin{restatable}[Unique direction of a TACC]{lemma}{direction} \label{lemma:direction}
  Let $c_0, \ldots, c_T$ with $0 < T \leq L$ be the clock qudits of a TACC with non-zero length. If there exists a $\Pi_{prop}$ clause within the TACC that points towards $c_0$, the instance is unsatisfiable.
\end{restatable}

\noindent
Note that unlike in the previous construction where these two lemmas (\cref{lemma:linearchainmono,lemma:directionmono}) applied to the whole active clock chain, here they only apply within the TACC.

As \cref{lemma:linearchain} only referred to the clock qudits in the interior of the chain, it remains to answer whether the clock qudits at the ends of the TACC may take part in other $\Pi_{prop}$ clauses. In the previous construction, this question was trivial since the auxiliary subspaces and the principle of monogamy of entanglement ensured that (for potentially satisfiable instances) the clock qudit with the $\Pi_{init}$ clauses and the one with the $\Pi_{out}$ clauses marked the beginning and end of the chain, respectively. The only case where there could be clauses beyond the TACC, was when the it was a strict subset of the ACC, i.e.\ when an undefined clause truncated the propagation within an ACC. However, even in this case, the clock qudit with the $\Pi_{out}$ clauses marked the end of the component. We now show that we can achieve a similar behavior in this construction.

\begin{restatable}[Neighboring clauses of a non-zero length TACC]{lemma}{outsidetacc} \label{lemma:outsidetacc}
  Let $c_0, \ldots, c_T$ with $0 < T \leq L$ be the clock qudits of a TACC with non-zero length. $\Pi_{prop}$ clauses that lie outside the TACC with the following properties result in an unsatisfiable instance:
  \begin{enumerate}
    \item If they are well-defined or undefined and include $c_0$.
    \item If they are well-defined and include $c_T$.
    \item If they are undefined, include $c_T$, and point towards it.
  \end{enumerate}
\end{restatable}
\noindent
In this lemma, it is important that $T$ should be able to equal $L$, since without endpoint qudits to guarantee that the chain ends at the clock qudit with the $\Pi_{out}$ clause, it is possible that the chain extends beyond this point. This lemma does not reject the only other type of $\Pi_{prop}$ clause available: that where the clause is undefined, includes $c_T$, and points away from it. This is because these sub-instances may be satisfiable. Even in the previous construction this was a possibility if the TACC was a strict subset of the ACC ($T < L$), as seen in \cref{lemma:truncated}. The difference now is that the clock chain may continue past the clock qudit with the $\Pi_{out}$ clause. Moreover, since the clauses beyond this TACC may not be part of other TACCs, lemmas \cref{lemma:direction,lemma:linearchain} may not apply and the clauses could form forks or cycles, with clauses pointing in any direction. This behavior can be seen in \cref{fig:semilinearsat}.

Now, let us briefly discuss the behavior of active dots. Here, they may occur when the clock qudit is present in both $\Pi_{init}$ and $\Pi_{out}$ clauses, or the clock qudit is present in at least one $\Pi_{init}$ clause and the ACC is immediately truncated by an undefined clause. Since these TACCs do not contain $\Pi_{prop}$ clauses, all there is to do is to analyze the $\Pi_{prop}$ clauses external to the active dot. The next lemma shows that the clock qudit of the TACC behaves like $c_T$ in \cref{lemma:outsidetacc}.

\begin{restatable}[Neighboring clauses of an active dot]{lemma}{outsidedot} \label{lemma:outsidedot}
  Let $c_d$ be an active dot. $\Pi_{prop}$ clauses that act on $c_d$ with the following properties result in an unsatisfiable instance:
  \begin{enumerate}
    \item If they are well-defined or undefined, and point towards $c_d$.
    \item If they are well-defined and point away from $c_d$.
  \end{enumerate}
\end{restatable}
\noindent
Thus, the only type of $\Pi_{prop}$ clauses that can act on this qudit are those that are undefined and point away from the qudit. This is shown in \cref{fig:alloweddots}.

This lemma and \cref{lemma:outsidetacc} show that a TACC (regardless of whether it has zero length, is truncated, or encompasses the whole ACC) is allowed to continue only if the next $\Pi_{prop}$ clause is undefined and points away from the last clock qudit of the TACC.

The observations that the clock component is able to continue past the clock qudit with the $\Pi_{out}$ clause and the clauses outside TACCs may fork or have any direction make the analysis of the satisfiability of these clock components seem more daunting; however, as we will now show, the lemmas above suffice to determine the satisfiability of all instances. \\

\medbreak
\noindent
\textbf{Putting the pieces together}\\
\noindent
Now we can put the previous lemmas into the bigger picture and show how they can allow us to decide the satisfiability of an instance or reach the point where the quantum algorithm can take over.

To begin, \cref{lemma:singletypequdit} ensures that if the instance is potentially satisfiable, all qudits in the instance must serve a single role. Afterwards, we focus on the sub-instances, and use \cref{lemma:noendpoints} to ignore those that do not have both $\Pi_{init}$ and $\Pi_{out}$ clauses. For each of the remaining sub-instances, \cref{lemma:hinitdirections} states that clock qudits present in $\Pi_{init}$ clauses must have one of two possible structures: either the clock qudit is present only in undefined clauses that point away from it,  or is part of well-defined $\Pi_{prop}$ clauses that involve at most one other clock qudit (simultaneous propagation clauses are allowed). If the clock qudit has the first of these structures then this clock qudit forms an active dot. On the other hand, if the clock qudit has the second structure, then it must be part of a non-zero active clock chain. \cref{lemma:direction,lemma:linearchain} then dictate that this non-zero length TACC must be a linear chain and all $\Pi_{prop}$ clauses in the chain must point away from the clock qudit with the $\Pi_{init}$ clause. The TACC may either finish because of an undefined clause or end when it reaches the clock qudit with the $\Pi_{out}$ clause. In either case, \cref{lemma:outsidetacc} shows that any adjacent clauses to the end of the TACC must be undefined and point away from the clock qudit that culminates the TACC. Unfortunately, we are unable to make any statement about the shape, direction, or definition of any $\Pi_{prop}$ clauses beyond this one.

In summary, these lemmas show that potentially satisfiable sub-instances must have a clock component that is a collection of TACCs and active dots joined by a complex web of clock qudits consisting only of $\Pi_{prop}$ and $\Pi_{out}$ clauses (see \cref{fig:semilinearsat}). Crucially, as there are no $\Pi_{init}$ clauses in this complex web of clauses to demand that there is an active clock qudit, the clauses within the web can be satisfied trivially by setting their clock qudits to $\ket{r}$. Moreover, since the clauses adjacent to the TACCs point away from them, this assignment does not conflict with the possible satisfying history states or truncated history states of the TACCs. In other words, the complex web of $\Pi_{prop}$ and $\Pi_{out}$ that link the TACCs can be thought of as being the joint successor of all TACCs.

As mentioned, since the clauses connecting the TACCs can be satisfied without constraining any of their logical qudits, we can effectively think of the TACCs of the sub-instance as smaller disconnected sub-instances. For TACCs that are truncated (including active dots), we can treat them as separate sub-instances by placing an artificial $\Pi_{stop}$ clause on the last clock qudit of the TACC and ignore all clauses beyond. This does not alter the TACC since this is the exact same effect that an undefined clause pointing away from the qudit has on the clock component. This idea is illustrated in \cref{fig:separatetaccs}. We remark that the chains with artificial $\Pi_{stop}$ clause at one end should not be confused with chains that end with $\Pi_{out}$ clauses which constrain both clock and logical qudits. All of the sub-instances now have a simpler clock component where they begin with a clock qudit present in a $\Pi_{init}$ clause, a series of well-defined $\Pi_{prop}$ clauses forming a one-dimensional chain and pointing away from the clock qudit with the $\Pi_{init}$ clause, and end with a clock qudit present in a $\Pi_{stop}$ clause (some which may be artificially placed and others due to a $\Pi_{out}$ clause).

\begin{figure}[t]
  \centering
  \includegraphics[width=0.75\textwidth]{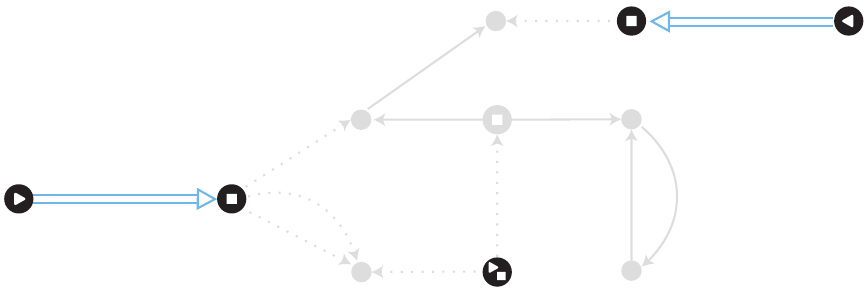}
  \caption{The clock component of \cref{fig:semilinearsat} where the trivially satisfiable clauses are now in gray and the TACCs are highlighted. The TACC on the left and the TACC on the bottom receive an artificial $\Pi_{stop}$ clause acting on their last clock qudit. These TACCs can now be treated as separate sub-instances.}
  \label{fig:separatetaccs}
\end{figure}

\subsubsection{Shared logical qudits and simultaneous propagation}

The discussion about shared logical qudits and simultaneous propagation clauses is the same as in the previous construction. There, the auxiliary subspaces helped determine that the active clock chains should be, as the name implies, one-dimensional chains of qudits where all $\Pi_{prop}$ clauses point in a unique direction. Previous to this, we argued that the relevant logical qudits should be those involved in clauses stemming from the TACC. We have obtained these same results in this construction above.

Similarly, the statements about simultaneous propagation clauses still hold, as they did not rely on auxiliary subspaces or endpoint qudits.

\subsection{Algorithm and correctness}

In this subsection we discuss the algorithm showing that SLCT-QSAT can be decided in polynomial-time with a quantum computer while demonstrating perfect completeness and bounded soundness, i.e.\ SLCT-QSAT is in \BQP$_1$. This is the first part of the proof of \cref{thm:bqpsmall}.

\subsubsection{Classical and quantum algorithm} \label{subsubsection:algorithm}

The first two steps in the algorithm change slightly compared to the previous one, while the third step remains unchanged. The first step no longer mentions $\Pi_E$, and in the second step, only (2.5) remains. To be concrete, these are:

\begin{enumerate}[label=(\arabic*)]
  \item \textit{Verify that all qudits in the instance serve a single role.} For every qudit in the instance, check whether it is acted on by $\Pi_{C}$ or $\Pi_{L}$ clauses. If a qudit is acted on by both types of clauses, reject. Otherwise, label the qudit as either $C$ or $L$ depending on the type of clause acting on it. Correctness is given by \cref{lemma:singletypequdit}.
  \item \textit{Ignore trivial sub-instances.} For every sub-instance, check whether it contains both $\Pi_{init}$ and $\Pi_{out}$ clauses. If it does not, ignore all clauses of the sub-instance for the rest of the algorithm. Correctness is given by \cref{lemma:noendpoints}.
  \item \textit{Identify the undefined $\Pi_{prop}$ clauses.} For every logical qudit, check whether it is acted on by at least one $\Pi_{init}$ clause. If it is not, mark all $\Pi_{prop}$ clauses this qudit is part of as undefined.
\end{enumerate}
\noindent
Next, we add a new step to account for the lemmas regarding the satisfiability of clock components and TACCs of this construction. To preserve the numbering of the algorithm, we label this step as (CC), referring to the steps related to the analysis of clock components. The step is the following:

\begin{enumerate}[label=(CC), leftmargin=1.2cm]
  \item For every sub-instance, do the following:
        \begin{enumerate}[label=(CC.\arabic*)]
          \item Identify all undefined $\Pi_{prop}$ clauses within the sub-instance and place an artificial $\Pi_{stop}$ clause on the predecessor qudit.
          \item \textit{Reject if the clock qudits with $\Pi_{init}$ clauses do not look like those of \cref{fig:allowedinits}.} For every clock qudit present in a $\Pi_{init}$ clause, but not acted on by any $\Pi_{stop}$ clause, perform the following checks:
                \begin{enumerate}[label=(CC.2.\arabic*)]
                  \item If the qudit participates in a $\Pi_{prop}$ clause that points towards it, reject. Correctness is given by \cref{lemma:hinitdirections}.
                  \item If the qudit participates in two or more $\Pi_{prop}$ clauses involving different clock qudits and one of these clauses is well-defined, reject. Correctness is given by \cref{lemma:singlewelldefined}
                \end{enumerate}
          \item \textit{Evaluate the segments of the chain that resemble the bottom drawing of \cref{fig:allowedinits} (these must be part of a non-zero length TACC). We reject the TACC if it does not look like those of \cref{fig:allowedtaccs}.} For each clock qudit present in a $\Pi_{init}$ clause and a well-defined $\Pi_{prop}$ clause, but not in any $\Pi_{stop}$ clause, do the following:
                \begin{enumerate}[label=(CC.3.\arabic*)]
                  \item Let $c_j$ with $j = 1$ be the successor clock qudit.
                  \item If $c_j$ is present in a $\Pi_{stop}$ clause, perform the following checks:
                        \begin{enumerate}
                          \item[--] If there is a $\Pi_{prop}$ clause (other than the one connecting it to $c_{j-1}$) that points towards it or is well-defined, reject. Correctness is given by \cref{lemma:outsidetacc}.
                          \item[--] \textit{Stop condition}. Stop the analysis of this sub-instance and continue with the next sub-instance.
                        \end{enumerate}
                  \item \textit{Reject if there's a fork.} If $c_j$ is connected to more than three distinct clock qudits, reject. Correctness is given by \cref{lemma:linearchain}.
                  \item If there's a $\Pi_{prop}$ clause between $c_j$ and $c_{j+1}$ that points towards $c_j$, reject. Correctness is given by \cref{lemma:direction}.
                  \item \textit{The previous two steps guarantee that $c_j$ has a single successor.} Replace $j \rightarrow j+1$ and go back to step (CC.3.2).
                \end{enumerate}
          \item \textit{Evaluate the active dots}. For every clock qudit present in both a $\Pi_{init}$ and $\Pi_{stop}$ clause, check whether it participates in a $\Pi_{prop}$ clause that points towards it, or participates in a well-defined clause. If any of these are true, reject. Correctness is given by \cref{lemma:outsidedot}.
          \item \textit{Separate the TACCs found in the previous steps}. Ignore the clauses that do not form part of TACCs.
        \end{enumerate}
\end{enumerate}

\noindent
After this step, the algorithm has separated the TACCs of the sub-instances into new smaller sub-instances, each with the same structure as in \cref{section:monogamy}. Consequently, the algorithm resumes with step (4) of this previous construction (see \cref{subsection:algorithmmono}). While steps (4), (6), and (7) proceed in the same way as before, step (5) needs a small modification:

\begin{enumerate}[label=(5)]
  \item \textit{Ignore truncated active clock chains.} For every sub-instance, check whether it ends with an artificial $\Pi_{stop}$ clause (\textit{not} an actual $\Pi_{out}$ clause). If it does, and there are is no simultaneous propagation before this clause, ignore the clauses of the sub-instance. Correctness is given by \cref{lemma:truncated}.
\end{enumerate}

\subsubsection{Runtime}

Compared to the previous algorithm, the only significant difference in this algorithm that may impact the runtime is the addition of step (CC). However, this step is also efficient.

In detail, step (CC.1) can be easily seen to take $\mathcal{O}(poly(n))$ time. Step (CC.2) iterates over the $\mathcal{O}(n)$ clock qudits of the sub-instance and verifies whether each qudit is present in a $\Pi_{init}$ clause and a $\Pi_{prop}$ clause with some property. This takes $\mathcal{O}(n) \cdot \mathcal{O}(poly(n)) = \mathcal{O}(poly(n))$ time. In step (CC.3), the algorithm traverses the TACC and for each clock qudit, it checks the clauses that act on it. Similarly to step (CC.2), this step also takes $\mathcal{O}(poly(n))$ time. Step (CC.4) takes $\mathcal{O}(poly(n))$ time as the algorithm simply looks over all of the clauses that may act on a qudit. Finally, for step (CC.5), note that there are at most $\mathcal{O}(poly(n))$ clauses that are not part of a TACC. Thus, removing them from the instance also takes $\mathcal{O}(poly(n))$ time.

\subsubsection{Completeness and soundness}

This section is no different than in the previous construction. \cref{subsection:satisfiability} covers an exhaustive list of instances that can be decided based on their structure. Since the algorithm above is based on this list, it consistently rejects instances whose structure implies two or more clauses cannot be satisfied simultaneously by any state. Conversely, it never rejects instances whose structure could be consistent with a satisfying state. If the classical algorithm alone determines the instance, then both the completeness and soundness conditions are clearly satisfied. However, if the instance is not decided by the classical algorithm, then there must be some sub-instances that require a quantum approach.

As mentioned, the quantum sub-instances in this construction are identical to those of the previous construction. The quantum algorithm that decides them is also the same here, and so all of the arguments presented in the previous construction also apply here.

\subsection{Hardness}

Here, we show that SLCT-QSAT is $\BQP_1$-hard, thus completing the proof of \cref{thm:bqpsmall}.

Let $U_x = U_L \ldots U_1$ with $U_i \in \mathcal{G}_8$ and $L = poly(n)$ be the $\BQP_1$ circuit where given an instance $x$ of an arbitrary problem $A$, decides whether $x \in A_{yes}$ or $x \in A_{no}$. The input to the circuit is the $q$-qubit ancilla register $\ket{0}^{\otimes q}$, where $q = poly(n)$. The reduction from $x$ to an instance $x'$ of SLCT-QSAT requires $q + (L+1)$ qudits ($2$ qudits less than in the previous construction), each of dimension $6$, and is the following:

\begin{enumerate}
  \item Choose $n$ qudits to serve as the logical qudits of the computation and $L+1$ qudits to serve as the clock qudits. Define an ordering for each type of qudit.
  \item For all $i \in [q]$, create a $\Pi_{init}$ clause acting on $c_0$ and $l_i$.
  \item For all $j \in [L]$, create a $\Pi_{prop}$ clause with associated unitary $U_{j+1}$, which acts on clock qudits $c_j$ and $c_{j+1}$ (with $c_j$ being the predecessor of $c_{j+1}$) and the two logical qudits that the unitary originally acts on.
  \item Create one $\Pi_{out}$ clause acting on $c_L$ and $l_{ans}$.
\end{enumerate}
\indent
The resulting instance $x'$ consists of a one-dimensional clock chain of length $L$ ($L+1$ clock qudits) with a unique direction, no simultaneous propagation clauses, and fully initialized logical qudits.

If $x$ is a yes-instance of $A$, $x'$ is satisfiable and is satisfied entirely by the history state of \cref{eqn:histhardness}.

On the other hand, when $x$ is a no-instance of $A$, we again show that $x'$ has no small eigenvalues by expressing the instance in terms of Hamiltonians as

\begin{equation} \label{eqn:hamilsoundnessfree}
  H = H_{roles} + H_{clock} + H_{prop} + H_{init} + H_{out}.
\end{equation}
\noindent
The definitions of these Hamiltonians remain mostly the same as in the previous construction, except that now $H_{roles}$ does not contain any $\Pi_E$ term. Moreover, observe that in contrast with \cref{eqn:hamilsoundness}, the equation above has no $H_{aux}$ term. We now show that despite these differences, the analysis within the null space of $H_{clock}$ from the previous section still applies, and we can again conclude that the smallest eigenvalue of $H$ scales with the inverse of a polynomial.

Recall that in the construction of \cref{section:monogamy}, the constraints on the auxiliary subspaces determined that the states within the null space of $H_{aux}$ were those with clock qudits with a unique successor and predecessor, and were all preceded by the clock qudit with the $\Pi_{init}$ clauses and succeeded by the one with the $\Pi_{out}$ clauses.\footnote{When demonstrating that the problem was in $\BQP_1$, the constraints also helped determine that satisfiable sub-instances could not have more clauses that extended past the ACC. This is not relevant when proving hardness as the instance we construct already satisfies this property.} As shown in \cref{lemma:direction,lemma:linearchain}, the satisfying states of SLCT-QSAT instances must also satisfy these properties. A quick inspection of the proofs of these statements shows that this is due to the clock constraints. Therefore, states within the null space of $H_{clock}$ must also satisfy these two properties, bypassing the need of an $H_{aux}$ term.

\section{\QCMA-complete Problem} \label{section:qcma}

\begin{figure}[t]
  \centering
  \begin{subfigure}[c]{0.48\textwidth}
    \centering
    \includegraphics[width=0.8\textwidth]{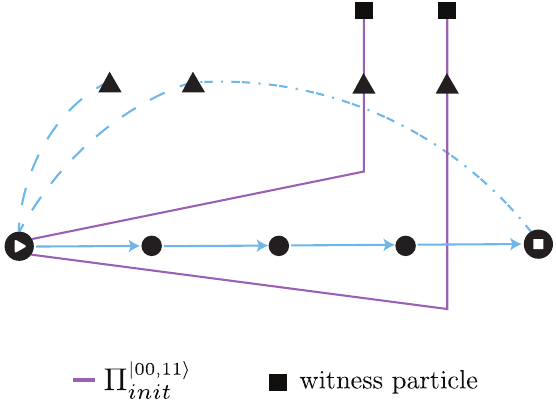}
    \caption{}
    \label{fig:qcmainstance}
  \end{subfigure}
  \hfill
  \begin{subfigure}[c]{0.48\textwidth}
    \centering
    \[ \Qcircuit @C=1.4em @R=1.2em {
      \lstick{\ket{0}} &  \qw  & \multigate{3}{\hspace{1.5em} U \hspace{1.5em}} & \qw & \qw\\
      \lstick{\ket{0}} &  \qw  & \ghost{\hspace{1.5em} U \hspace{1.5em}} & \qw & \meter \\
      & \qw & \ghost{\hspace{1.5em} U \hspace{1.5em}} & \qw & \qw \\
      & \qw & \ghost{\hspace{1.5em} U \hspace{1.5em}} & \qw & \qw \\
      & \qw & \qw & \qw & \qw \\
      & \qw & \qw & \qw & \qw
      \inputgroupv{3}{6}{1.2em}{2.43em}{\ket{\psi_{\rm wit}}\;\;\;\;\;} \\
      }
    \]
    \caption{}
    \label{fig:qcmacirc}
  \end{subfigure}
  \caption{(a) Toy example of an input ``quantum'' instance with a TACC of length $L=4$, acting on four logical qudits and two witness qudits. Although not illustrated, the $\Pi_{prop}$ clauses are assumed to have unitaries $U_1, \ldots, U_4$ which act only on the logical qudits of the instance. These unitaries define a circuit $U = U_4 U_3 U_2 U_1$. (b) Quantum circuit representing the instance on the left.}
  \label{fig:qcmadrawings}
\end{figure}

The constructions from \cref{section:monogamy,section:our} can be modified to generate $\QCMA_1^{\mathcal{G}_8}$-complete problems. However, since the second problem we presented is an improvement over the first, we only elaborate on the latter. Moreover, since $\QCMA_1^{\mathcal{G}_8} = \QCMA$ \cite{jordan2011qcma}, this results in a $\QCMA$-complete problem. Although there are already many problems known to be complete for this class \cite{gharibian2015ground, wocjan2003qcma, gosset2017qcma, chenqcma}, none of them are strong QCSPs. Again, for the rest of the section, we omit the superscript $\mathcal{G}_8$ when talking about $\BQP_1$ and $\QMA_1$.

\subsection{Construction overview}

At the beginning of \cref{subsection:monosat}, we argued that the unconstrained or ``free'' logical qudits of an instance allowed one to guess what state of these qudits (the witness state) might satisfy the instance. This freedom made the problem more difficult and thus not likely contained in $\BQP_1$. In fact, one could show that this problem is $\QMA_1$-complete.\footnote{We do not prove this in this paper since the much simpler $k$-QSAT problem is already known to be $\QMA_1$-complete for $k \geq 3$ \cite{bravyi2006efficient,gosset2016quantum}.} For this reason, we introduced the undefined state $\ket{?}$, which simplified these instances and made them decidable in $\BQP_1$. In this construction, we seek to construct a problem that sits in between these two classes so it is $\QCMA$-complete. To accomplish this, we desire to have ``free'' logical qudits to accommodate a witness state that helps verify whether the instance is satisfiable, but have some sort of constraint to demand that the state is classical.\footnote{We continue using the undefined state for logical qudits whose initial state is not constrained so the difficulty of the problem does not become $\QMA$.}

In practice, creating these constraints is a difficult task because if there are two states that satisfy a clause, any superposition of the two will do so as well. Instead, we set the constraints such that if there exists a quantum witness state that is part of a satisfying state, there is also a classical witness state. Loosely, we accomplish this by defining new \textit{witness qudits} and create a new constraint $\piinitcopy$ that connects a witness qudit with a logical one, and require that they are both either $\ket{00}$, $\ket{11}$, or in a superposition of the two. In this way, the two qudits are partially ``free'' as there is some freedom to their state yet posses some structure. Importantly, we ensure that the witness qudits do not form part of the computation after this initial point.

To see why this leads to the desired effect, consider the toy instance of \cref{fig:qcmadrawings} and suppose there exists a state $\ket{\psi_{\rm wit}}$ of the four ``free'' qudits that leads to a satisfying state. Observe that to satisfy the $\piinitcopy$ clauses, this state must be of the form $\ket{\psi_{\rm wit}} = (\alpha_{00} \ket{0000} + \alpha_{01} \ket{0011} + \alpha_{10} \ket{1100} + \alpha_{11} \ket{1111})_{L_1, W_1, L_2, W_2}$ with $\sum_{b \in \{0,1\}^2} \norm{\alpha_b}^2 = 1$, which we can rewrite in a more convenient form as $\ket{\psi_{\rm wit}} = \sum_{b \in \{0,1\}^2} \alpha_b \ket{b}_L \otimes \ket{b}_W$. Then, the state that satisfies all clauses of the instance is the history state

\begin{equation*}
  \begin{aligned}
    \ket{\psi_{hist}} & = \frac{1}{\sqrt{5}} \sum_{t = 0}^4 \left[ U_t \ldots U_0 \ket{00} \otimes \ket{\psi_{\rm wit}} \right] \otimes \ket{\underbrace{d \ldots d}_t a_t \underbrace{r \ldots r}_{4-t}} \\
                      & = \frac{1}{\sqrt{5}}\sum_{t=0}^4 \sum_{b \in \{0,1\}^2} \alpha_b \ket{\xi_b^t} \otimes \ket{b}_W \otimes \ket{\underbrace{d \ldots d}_t a_t \underbrace{r \ldots r}_{4-t}},
  \end{aligned}
\end{equation*}
\noindent
where $\ket{\xi_b^t} := U_t \ldots U_0 \ket{00} \otimes \ket{b}_L$. Let us show that there is also a classical witness that leads to a satisfying history state. First, observe that any basis state $\ket{b}_L \otimes \ket{b}_W$ with $\norm{\alpha_b} > 0$ from the decomposition of the witness satisfies the $\piinitcopy$ clauses. Consequently, the history state above but with initial state $\ket{00} \otimes \ket{b}_L \otimes \ket{b}_W$ satisfies the $\piinitzero$, $\piinitcopy$, and $\Pi_{prop}$ clauses of the instance. To show that this state also satisfies the $\Pi_{out}$ clause, recall that the clause is satisfied if at time $t=4$, the probability that the second qubit yields outcome ``1'' when measured is 1. This probability can be written as

\begin{equation*}
  \begin{aligned}
    \Pr(\textnormal{outcome } 1) & = \sum_{b,b' \in \{0,1\}^2} \alpha_b \alpha_{b'}^* \bra{\xi_{b'}^4} \otimes \bra{b'} \Pi^{(1)} \ket{\xi_b^4} \otimes \ket{b} \\
                                 & = \sum_{b \in \{0,1\}^2} \norm{\alpha_b}^2 \bra{\xi_b^4} \Pi^{(1)} \ket{\xi_b^4}
  \end{aligned}
\end{equation*}
\noindent
where $\Pi^{(1)} := \ket{1} \! \bra{1}_2 \otimes I_{rest}$, and in the second line we observed that $\braket{b'|b} = \delta_{b,b'}$. Then, by the assumption that the instance is satisfiable, it must be that $\bra{\xi_b^4} \Pi^{(1)} \ket{\xi_b^4} = 1$ for all basis states $\ket{b}$ with $\norm{\alpha_b} > 0$. This can also be understood as the probability that at the end of the circuit the second qubit yields outcome ``1'' when the witness is the basis state $\ket{b} \otimes \ket{b}$. Therefore, the history state

\begin{equation*}
  \ket{\psi_{hist}} = \frac{1}{\sqrt{5}} \sum_{t=0}^4 \left[U_t \ldots U_0 \ket{00} \otimes \ket{b}_L \right] \otimes \ket{b}_W \otimes \ket{\underbrace{d \ldots d}_t a_t \underbrace{r \ldots r}_{4-t}}
\end{equation*}
\noindent
with classical witness $\ket{\psi_{\rm wit}} = \ket{b}_L \otimes \ket{b}_W$ also satisfies all clauses of the instance.

While this toy example accurately conveys the idea behind the construction, there are more details we need to discuss to prove \cref{thm:qcma}. For this, let us define the problem and the new clause.

\subsection{Problem definition}

The problem we claim is $\QCMA$-complete is is the following:

\begin{defn}[\textsc{Witnessed SLCT-QSAT}] \label{defn:wslctqsat}
  The problem {\rm \textsc{Witnessed SLCT-QSAT}} is a QCSP defined on the $8$-dimensional Hilbert space

  \begin{equation*}
    \mathcal{H} = \Span\{\ket{0_L},\ket{1_L},\ket{?_L}\} \oplus \Span\{\ket{0_W},\ket{1_W}\} \oplus \Span\{\ket{r_C},\ket{a_C},\ket{d_C}\},
  \end{equation*}
  \noindent
  which now consists of a logical, a witness, and a clock subspace. Besides the operators $O_{init}$, $O_{prop}$, and $O_{out}$ defined previously in \cref{eqn:Oinitfree,eqn:Opropfree,eqn:Ooutfree}, the problem contains a new $3$-local operator

  \begin{equation*}
    O_{init}^{\scaleto{\ket{00,11}}{7pt}} := \; \; \Pi_{L,1} + \Pi_{W,2} +\Pi_{C,3} + \Pi_{start,3} + \ket{?} \! \bra{?}_1 \otimes \ket{a} \! \bra{a}_3  + (I - \ket{00} \! \bra{00} - \ket{11} \! \bra{11})_{1,2} \otimes \ket{a} \! \bra{a}_3,
  \end{equation*}
  \noindent
  where $\Pi_{W} := I - (\ket{0} \! \bra{0}_{W} + \ket{1} \! \bra{1}_{W})$ is the projector to the witness subspace. As before, $\piinitcopy$ refers to the projector with the same null space as $O_{init}^{\scaleto{\ket{00,11}}{7pt}}$.

  The promise remains the same as before. We are promised that for every instance considered either (1) there exists a {\rm quantum} state $\ket{\psi_{sat}}$ on $n$ $8$-dimensional qudits such that $\Pi_i \! \ket{\psi_{sat}} = 0$ for all $i$, or (2) $\Sigma_i \bra{\psi} \! \Pi_i \! \ket{\psi} \geq 1/poly(n)$ for all {\rm quantum} states $\ket{\psi}$.

  The goal is to output ``YES'' if (1) is true, or output ``NO'' otherwise.
\end{defn}
\indent
Before showing that this problem is contained within $\QCMA$ and is also $\QCMA$-hard, let us make some observations. Given that there are now multiple initialization clauses, we will refer to the original $\Pi_{init}$ clause used in \cref{section:monogamy,section:our} as $\piinitzero$, and use $\piinitdot$ when we do not care whether the clause is $\piinitzero$ or $\piinitcopy$. Observe that $\ket{?} \! \bra{?} \otimes \ket{a} \! \bra{a}$ prevents the logical qudit from being $\ket{?}$ when the clock qudit is $\ket{a}$. Then, together with the term $(I - \ket{00} \! \bra{00} - \ket{11} \! \bra{11}) \otimes \ket{a} \! \bra{a}$, this implies that the logical and witness qudits must be $\ket{00}$, $\ket{11}$, or any superposition of the two. Clearly, $\piinitcopy$ sets a ``weaker'' constraint than $\piinitzero$ as the logical qudit is not constrained to a single state, but it is also a ``stronger'' constraint than having completely free witness and logical qudits.

Now, let us show that this problem is in $\QCMA$.

\subsection{In \QCMA}

Recall that to demonstrate that the problem is in $\QCMA$, we have to show that when $x$ is a yes-instance, there is a basis state $\ket{y_x}$ that convinces us with certainty that the instance is satisfiable. On the other hand, if $x$ is a no-instance, we have to show that no classical basis state is accepted with high probability.

We demonstrate this result by following the same proof structure as in the previous constructions, starting with the satisfiability of clock components and shared logical (and now witness) qudits. To begin, observe that the introduction of the $\piinitcopy$ clause does not create new conflicts in the clock components. However, for the shared logical and witness qudits, there are six new cases to consider: three where the clauses stem from the same clock component, and three where they stem from different clock components. Here, we only state and prove the first three. The proofs of the remaining three cases follow easily by first recalling \cref{lemma:sharedhprop}, which states that the logical qudits cannot be present in $\Pi_{prop}$ clauses and so their state must be fixed, and then the proof can be completed by slightly adapting the proofs shown here. We also refer the reader to \cref{subsubsection:incorp} of the following $\coRP$ construction where there is only one case and we prove both versions of the statement. The first three cases are the following.

\begin{lemma} \label{lemma:logicalmultiplewitness}
  Consider a clock component that has at least one $\piinitdot$ and one $\Pi_{out}$ clause. Suppose there are two $\piinitcopy$ clauses acting on the same clock and logical qudit, but distinct witness qudits. The only states of the logical and two witness qudits that satisfy the $\piinitcopy$ clauses are of the form $\alpha_0 \ket{000} + \alpha_1 \ket{111}$ for any $\alpha_0, \alpha_1 \in \mathbb{C}$ that satisfies $\norm{\alpha_0}^2 + \norm{\alpha_1}^2 = 1$.
\end{lemma}

\begin{proof}
  Recall from \cref{prop:activetacc} that if the instance is satisfiable, both clauses must enforce their constraints on the shared logical qudit. Each $\piinitcopy$ then requires its logical and witness qudit to be in the same basis state and also forbids the logical qudit from being $\ket{?}$. Hence, these three qudits must be either $\ket{000}$ or $\ket{111}$. Furthermore, as mentioned in the overview, any valid superposition of the two also satisfies these clauses, and so the satisfying state is in general $\alpha_0 \ket{000} + \alpha_1 \ket{111}$.
\end{proof}

\begin{lemma} \label{lemma:witnessmultiplelogical}
  Consider a clock component that has at least one $\piinitdot$ and one $\Pi_{out}$ clause. Suppose there are two $\piinitcopy$ clauses acting on the same clock and witness qudit, but distinct logical qudits. The only states of the logical and two witness qudits that satisfy the $\piinitcopy$ clauses are of the form $\alpha_0 \ket{000} + \alpha_1 \ket{111}$ for any $\alpha_0, \alpha_1 \in \mathbb{C}$ that satisfies $\norm{\alpha_0}^2 + \norm{\alpha_1}^2 = 1$.
\end{lemma}

\begin{proof}
  The setup is similar to that in \cref{lemma:logicalmultiplewitness}, except now there are two witness qudits and one logical qudit. Despite this difference, it is easy to see that the proof of that statement also applies here.
\end{proof}

\begin{lemma} \label{lemma:differentinits}
  Consider a clock component that has at least one $\piinitdot$ and one $\Pi_{out}$ clause. Suppose a $\piinitzero$ and a $\piinitcopy$ clause act on the same clock and logical qudits. If the witness qudit is in the state $\ket{1}$, the instance is unsatisfiable.
\end{lemma}

\begin{proof}
  The first $\piinitzero$ clause requires the logical qudit to be in the state $\ket{0}$. The only way to satisfy $\piinitcopy$ is if both the logical and witness qudits are in the same state. Hence, if the witness qudit is in the state $\ket{1}$, the clauses cannot be satisfied simultaneously.
\end{proof}
\noindent
In a sense, these lemmas show that it is possible for a witness qudit to determine the state of multiple logical qudits, or vice versa. For a general input instance where the witness and logical qudits may be intricately connected by $\piinitcopy$ and $\piinitzero$ clauses, there exists a general form of their state that satisfies these clauses.

\begin{prop} \label{prop:statewitness}
  Consider a clock component that has at least one $\piinitdot$ and one $\Pi_{out}$ clause. Suppose the sub-instance contains $p$ witness qudits connecting to $r$ logical qudits via $\piinitcopy$ clauses. Depending on the instance, $r$ may be smaller or larger than $p$. Furthermore, suppose that a subset $Z$ of the logical qudits are also acted on by $\piinitzero$ clauses. The state of the witness and logical qudits that satisfies the $\piinitcopy$ clauses of the sub-instance (the witness state) is of the form

  \begin{equation} \label{eqn:statewitness}
    \ket{\psi_{\rm wit}} = \sum_{b \in \{0,1\}^p} \alpha_b \ket{f(b)}_L \otimes \ket{b}_W,
  \end{equation}
  \noindent
  where $\sum_{b \in \{0,1\}^p} \norm{\alpha_b}^2 = 1$ and $f: \{0,1\}^p \mapsto \{0,1\}^r$ is a function that depends on the specific arrangement of the $\piinitcopy$ clauses within the component.
\end{prop}

\begin{proof}
  First, let us consider the $\piinitcopy$ clauses only. Observe that, similar to clock components, the logical and witness qudits can be seen as either forming joint or disjoint sets judging by their connections via $\piinitcopy$ clauses. Suppose that there are $m$ such disjoint sets, and let $r_i$ and $p_i$ with $i \in [m]$ be the number of logical and witness qudits of the $i$-th set. Clearly, $r = \sum_{i \in [m]}  r_i$ and $p = \sum_{i \in [m]} p_i$. By generalizing the statements of \cref{lemma:logicalmultiplewitness,lemma:witnessmultiplelogical}, it is evident that the state of any connected set must be of the form $\alpha_0 \ket{0}^{\otimes (r_i + p_i)} + \alpha_1 \ket{1}^{\otimes (r_i + p_i)}$, which we can rewrite as $\sum_{b'_i \in \{0,1\}} \alpha_{b'_i} \ket{b'_i}^{\otimes r_i}_L \otimes \ket{b'_i}^{\otimes p_i}_W$ where we have separated the logical and witness qudits. The overall state is then the $m$-fold tensor product of such states, given by $\sum_{b' \in \{0,1\}^m } \alpha_{b'} \ket{b'_0}^{\otimes r_0} \ldots \ket{b'_{m-1}}^{\otimes r_{m-1}} \otimes \ket{b'_0}^{\otimes p_0} \ldots \ket{b'_{m-1}}^{\otimes p_{m-1}}$ where we let $\alpha_{b'} :=  \alpha_{b'_0} \ldots \alpha_{b'_{m-1}}$. This is almost \cref{eqn:statewitness}, except now $b'$ is a bitstring of length $m$, referring to the number of connected components instead of the number of witness qudits. We can expand the bitstring to be of size $p$, if we let $\alpha_b = 0$ for all $b \in \{0,1\}^p$ when the state $\ket{b} \neq \ket{b'_0}^{\otimes p_0} \ldots \ket{b'_{m-1}}^{\otimes p_{m-1}}$, and $\alpha_b = \alpha_{b'}$ if $\ket{b} = \ket{b'_0}^{\otimes p_0} \ldots \ket{b'_{m-1}}^{\otimes p_{m-1}}$. From these previous statements it is obvious that the state of the logical qudits is related to the state of the witness qudits. For the proofs that follow, it is convenient to write their state simply as $\ket{f(b)}$.

  Now, let us consider the effect of $\piinitzero$ clauses on the logical qudits from $Z$. To satisfy these clauses, we now simply impose the additional restriction that if any of the logical qudits in $Z$ are $\ket{1}$ within $\ket{f(b)}$, we set $\alpha_b = 0$.
\end{proof}
\indent
This concludes the discussion about the satisfiability of clock components. The subsequent discussions about simultaneous propagation clauses and active dots remains the same as before. We now present the algorithm that decides all instances of this problem.

\subsubsection{Algorithm}
The input to the algorithm is the instance $x$ and a classical witness state $\ket{y_x}$ of size equal to the number of witness qudits in the instance. We do not consider a state of the ``free'' logical qudits, since as just shown, their state can be determined by their connections to witness qudits. Even if such a state is provided, the algorithm below will overwrite it, which spares us from having to add additional checks to verify that the states of witness and logical qudits match. The algorithm that decides all instances of the problem is the following.

The algorithm has the same first two steps as in \cref{subsubsection:algorithm}. We now add two new steps in between (2) and (3) to account for the new lemmas above:\footnote{A step to include \cref{lemma:witnessmultiplelogical} is not necessary with the given form of the witness.}

\begin{enumerate}[label=(QCMAa), leftmargin=2.2cm]
  \item For every logical qudit in the instance, check whether it is connected to several witness qudits. If the states of the witness qudits are not the same, reject. Correctness is given by \cref{lemma:logicalmultiplewitness}.
\end{enumerate}

\begin{enumerate}[label=(QCMAb), leftmargin=2.2cm]
  \item For every logical qudit in the instance, check whether it is present in both $\piinitzero$ and $\piinitcopy$ clauses. If there is such a qudit and one of the connected witness qudits is in the state $\ket{1}$, reject. Correctness is given by \cref{lemma:differentinits}.
\end{enumerate}
\noindent
We note that these two instructions account for both scenarios: cases where the clauses stem from the same clock component and those from different components. Normally, the latter cases would be evaluated in step (4); however, in this particular construction, they should be evaluated at this time since the next step initializes the state of all logical qudits to their corresponding values. This slight adjustment makes no difference in the correctness of the algorithm as no sub-instances are accepted in steps (3) or (4) and so it is safe to permute these around.

Step (3) also requires a small modification:

\begin{enumerate}[label=(3)]
  \item \textit{Identify the undefined $\Pi_{prop}$ clauses.} For every logical qudit, check whether it is acted on by at least one $\piinitzero$ or $\piinitcopy$ clause. If it is not, mark all $\Pi_{prop}$ clauses this qudit is part of as undefined.
\end{enumerate}
\noindent

The rest of the algorithm proceeds mostly in the same way as in \cref{subsubsection:algorithm}, the only other remaining difference is on the input of the quantum algorithm in step (6.1). If we assume the TACC has $q$ logical qudits initialized only by $\piinitzero$ clauses, and $p$ logical qudits initialized by $\piinitpair$ clauses, the step in the algorithm is then:

\begin{enumerate}[label=(6.1)]
  \item Initialize a register containing $q+2$ qubits, each initialized to $\ket{0}$ (the two extra qubits are the ancillas used in $\mathcal{C}_{\mathcal{G}_8}$). Initialize a second register of $p$ qubits where the state of each qubit is chosen to match the state of the witness qudits it connects to.
\end{enumerate}
\indent
Finally, we point out that as promised, the witness qudits are not acted on by any of the unitaries of the circuit.

\subsubsection{Runtime}
The runtime of the algorithm can be easily seen to be in polynomial time since the new steps (QCMAa) and (QCMAb) simply iterate over $\mathcal{O}(n)$ logical qudits and for each qudit the algorithm searches over a list of $\mathcal{O}(poly(n))$ clauses.

\subsubsection{Completeness}

Let us demonstrate that if $x$ is a yes-instance of \textsc{Witnessed SLCT-QSAT}, there is a classical basis state $\ket{y_x}$ that makes the algorithm accept with certainty. First, let us address the question about the existence of the classical basis state, and then show that the instance is accepted with certainty if this state is provided.

The key to answering this first question relies on analyzing each sub-instance separately. We will show that for each sub-instance whose TACC has $\piinitcopy$ clauses connecting $p$ witness and $r$ logical qudits, there exists a ($p+r$)-qubit basis state that helps determine the sub-instance is satisfiable. Then, $\ket{y_x}$ will be the union of these basis states for all sub-instances. Although some sub-instances may share logical and witness qudits, the assumption that the instance is satisfiable and \cref{lemma:logicalmultiplewitness,lemma:witnessmultiplelogical,lemma:differentinits} assure us that there is a state of these qudits for which all sub-instances can be simultaneously satisfied.

To begin, recall that there are two types of sub-instances: those that can be decided classically, and those that make use of a quantum sub-routine. Let us discuss each type separately. We make the assumption that all sub-instances below contain at least one $\piinitcopy$ clause, as otherwise they would be no different than those of \cref{section:our} whose satisfiability is independent of a witness.\\

\medbreak
\noindent
\textbf{Classical witness for classical sub-instances} \\
\noindent
By the lemmas of \cref{subsubsection:clock}, we know that each qudit of the instance must serve a single role and all sub-instances must be of one of two types: either they have no ACCs, or they have truncated TACCs with no simultaneous propagation clauses. In addition, the shared logical qudits between sub-instances must agree with the lemmas of \cref{subsubsection:shared}.

For sub-instances without ACCs, the witness and logical qudits are irrelevant as the sub-instance can be satisfied by simply choosing the state of the clock qudits appropriately. Thus, the free qudits can be in any basis state.

For sub-instances with truncated TACCs, there must be a state $\ket{\psi_{\rm wit}}$ of the free qudits such that the truncated history state

\begin{equation} \label{eqn:witnessedhist}
  \ket{\psi_{thist}} = \frac{1}{\sqrt{T+1}} \sum_{t = 0}^T  \big[ U_{t,0} \ldots U_{0,0} \ket{0}^{\otimes q} \otimes\ket{\psi_{\rm wit}} \big] \otimes \ket{\underbrace{d \ldots d}_{t} a_t \underbrace{r \ldots r}_{T-t}}
\end{equation}
\noindent
satisfies all clauses of the TACC. Here, $T < L$ and $\ket{0}^{\otimes q}$ are the logical qudits within the TACC that are present in $\piinitzero$ clauses (including $\piinitzero$ clauses from other sub-instances), but not present in any $\piinitcopy$ clause. On the other hand, $\ket{\psi_{\rm wit}}$ might contain logical qudits that are acted on by $\piinitzero$ clauses and thus constrained to be $\ket{0} $ (these qudits are then not truly ``free'', but grouping them this way helps simplify the arguments below). As stated in \cref{prop:statewitness}, the witness must be of the form $\ket{\psi_{\rm wit}} = \sum_{b \in \{0,1\}^p} \alpha_b \ket{f(b)}_L \otimes \ket{b}_W$ with $\sum_b \norm{\alpha_b}^2 = 1$. We can show that there is a classical witness state using similar arguments to those given in the toy example. First, observe that all $\ket{f(b)} \otimes \ket{b}$ with $\norm{\alpha_b} > 0$ in the decomposition of the witness satisfy the $\piinitzero$ and $\piinitcopy$ clauses of the sub-instance. Then, given that the sub-instance does not contain simultaneous propagation clauses and the sub-instance is truncated, the $\Pi_{prop}$ and $\Pi_{out}$ clauses are satisfied regardless of the initial state. In other words, replacing $\ket{\psi_{\rm wit}}$ with any $\ket{f(b)} \otimes \ket{b}$ (as long as $\norm{\alpha_b} > 0$) in the history state above also leads to a satisfying history state. \\

\medbreak
\noindent
\textbf{Classical witness for quantum sub-instances} \\
\noindent
Recall that quantum sub-instances must have one of the three forms shown in \cref{fig:qinputs}. Let us begin with the third type of sub-instance as these are similar to the truncated classical sub-instances above, except their TACC now contains simultaneous propagation clauses.

Sub-instances of this first type are also satisfied by \cref{eqn:witnessedhist}. By the same argument, any basis state $\ket{f(b)} \otimes \ket{b}$ with $\norm{\alpha_b} > 0$ in the superposition of the witness state satisfies the $\piinitzero$ and $\piinitcopy$ clauses. The $\Pi_{out}$ clauses of the instance are again trivially satisfied, and the lone $\Pi_{prop}$ clauses are satisfied independently from the initial state. Then, it only remains to show that simultaneous propagation clauses can also be satisfied with a classical witness state as input.

\begin{prop} \label{prop:simpropclassical}
  Suppose two $\Pi_{prop}$ clauses with unitaries $U_{t+1,i}$ and $U_{t+1,j}$ within the TACC of a sub-instance are satisfiable, i.e.\ $U_{t+1,i} \ket{\phi_t} = U_{t+1,j} \ket{\phi_t}$ where $\ket{\phi_t} = U_{t,0} \ldots U_{0,0} \ket{\phi_0}$ is the state of the data qudits at time $t$ (see \cref{prop:simultaneousprop}). For sub-instances of {\rm \textsc{Witnessed SLCT-QSAT}}, $\ket{\phi_0} =  \ket{0}^{\otimes q} \otimes \ket{\psi_{\rm wit}}$ where in general ${\ket{\psi_{\rm wit}}} = \sum_{b \in \{0,1\}^p} \alpha_b \ket{f(b)}_L \otimes \ket{b}_W$ (see \cref{prop:statewitness}). Replacing $\ket{\psi_{\rm wit}}$ with any of its basis states $\ket{f(b)} \otimes \ket{b}$ that has $\norm{\alpha_b} > 0$ still results in a state $\ket{\phi_t}$ that satisfies the simultaneous $\Pi_{prop}$ clauses.
\end{prop}

\begin{proof}
  Let us write the intermediate state of the computation as $\ket{\phi_t} = \sum_{b \in \{0,1\}^p} \alpha_b \ket{\xi_b^t} \otimes \ket{b}$, where $\ket{\xi_b^t} = U_{t,0} \ldots U_{0,0} \ket{0}^{\otimes q} \otimes \ket{f(b)}$. Then, by assumption

  \begin{equation*}
    \begin{aligned}
      1 & = \bra{\phi_t} U_{t+1,j}^\dagger U_{t+1,i} \ket{\phi_t}                                                                                        \\
        & = \sum_{b,b' \in \{0,1\}^p} \alpha_b \alpha_{b'}^* \bra{\xi_{b'}^t} \otimes \bra{b'} U_{t+1,j}^\dagger U_{t+1,i} \ket{\xi_b^t} \otimes \ket{b} \\
        & = \sum_{b \in \{0,1\}^p} \norm{\alpha_{b}}^2 \bra{\xi_b^t} U_{t+1,j}^\dagger U_{t+1,i} \ket{\xi_b^t}.
    \end{aligned}
  \end{equation*}
  \noindent
  For this equation to be true, it must be that $\bra{\xi_b^t} U_{t+1,j}^\dagger U_{t+1,i} \ket{\xi_b^t} = 1$ for all $b \in \{0,1\}^p$ with $\norm{\alpha_b} > 0$. Consequently, the state $\ket{\phi'_t} = \ket{0}^{\otimes q} \otimes \ket{f(b)} \otimes \ket{b}$, where $\ket{f(b)} \otimes \ket{b}$ is any basis state of $\ket{\psi_{\rm wit}}$ with $\norm{\alpha_b} > 0$, can satisfy the $\Pi_{prop}$ clauses.
\end{proof}
\noindent
An important observation here is that since any basis state $\ket{f(b)} \otimes \ket{b}$ with $\norm{\alpha_b} > 0$ can help satisfy a pair of $\Pi_{prop}$ clauses, any single state we either pick or receive as input will also help satisfy all other subsequent simultaneous $\Pi_{prop}$ clauses of the TACC.

For the remaining two types of sub-instances, we simply consider the second as it is more general than the first. For this type of sub-instance, the TACC encompasses the whole ACC, and the satisfying state is given by \cref{eqn:witnessedhist}, except now $T = L$. As a consequence, the $\Pi_{out}$ clauses are no longer trivially satisfied. As we have shown previously, any basis state $\ket{f(b)} \otimes \ket{b}$ in the superposition of the quantum witness makes a valid classical witness that leads to a history state that satisfies the $\piinitzero$, $\piinitcopy$, and $\Pi_{prop}$ clauses. To see that this state also satisfies the $\Pi_{out}$ clauses, we can simply generalize the analysis of the toy example. Here, the sub-instances may now have a more erratic arrangement of their $\piinitcopy$ and $\piinitzero$,  but as shown previously, the witness state can be always be expressed as $\ket{\psi_{\rm wit}} = \sum_{b \in \{0,1\}^p} \alpha_b \ket{f(b)} \otimes \ket{b}$. Then, following the same reasoning as in the toy example, it is easy to show that $\bra{\xi_b^L} \Pi^{(1)} \ket{\xi_b^L} = 1$ for all $b$, where now $\ket{\xi_b^L} = U_L \ldots U_1 \ket{0}^{\otimes q} \otimes \ket{f(b)}$. Therefore, any classical witness state $\ket{f(b)} \otimes \ket{b}$ with non-zero amplitude in $\ket{\psi_{\rm wit}}$ also leads to acceptance.\\

\medbreak
\noindent
\textbf{Instance-wide classical witness and perfect completeness} \\
\noindent
We have just shown that if the instance is satisfiable, for each sub-instance with a TACC and ``free'' qudits, there exists a classical state of its free qudits such that the sub-instance can be satisfied. Moreover, there may be more than one such state for each sub-instance. One might then think that the overall classical witness state $\ket{y_x}$ could be the union of any arbitrary combination of these classical witnesses, and while almost correct, there is one caveat: since the sub-instances may not be completely independent and may share logical and witness qudits, it is possible that the state of one qudit leads to a satisfying state in one sub-instance but not in the others. However, since the instance must be satisfiable, there should exist a state of the qudit that can help satisfy all of them. Thus, $\ket{y_x}$ must be some particular combinations of the classical witnesses of each individual sub-instance.

To show perfect completeness, observe that by the correctness of the lemmas in this section and the previous ones, if presented with $\ket{y_x}$, the algorithm---which is based on these lemmas---never wrongfully rejects.

\subsubsection{Soundness}

Now, let us show that when $x$ is a no-instance, the algorithm rejects the instance with high probability for any input classical witness state.

To begin, assume that there's a pair of clauses that cannot be mutually satisfied. If the problem lies in the roles of qudits, the clock component, or in the shared logical qudits between sub-instances where the clauses considered are $\{ \piinitzero, \Pi_{prop}, \Pi_{out}\}$, the algorithm always rejects regardless of the witness state. This again follows from the correctness of the analysis of clock components and shared logical qudits from previous sections of the paper. If the unsatisfiable clauses are those involving the witness and logical qudits, then no matter which classical witness state is input to the algorithm, the algorithm rejects the instance as it is based on \cref{lemma:witnessmultiplelogical,lemma:logicalmultiplewitness,lemma:differentinits} (and the other three which we did not explicitly write down) which cover all possible scenarios.

If $x$ is a no-instance and is not rejected before the quantum subroutine, then the unsatisfiable clauses must lie in one of the quantum sub-instances. We can also assume that the witness state can mutually satisfy the $\piinitcopy$ and $\piinitzero$ clauses as bad witnesses would be caught in steps (QCMAa) and (QCMAb) of the algorithm. Then, either the simultaneous propagation clauses cannot be satisfied, or the final state of the computation results in a state where the logical qudits present in $\Pi_{out}$ clauses are $\ket{0}$. This is the same scenario as in \cref{subsubsection:completenesssoundnessmono}, and since the quantum subroutine here is the same as in that section, the probability that the algorithm accepts the instance is $\leq 1/3$.

\subsection{\QCMA-hard}

Here, we show that \textsc{Witnessed} SLCT-QSAT is $\QCMA$-hard, thus completing the proof of \cref{thm:qcma}.

Let $U_x = U_L \ldots U_1$ with $U_i \in \mathcal{G}_8$ and $L = poly(n)$ be the $\QCMA$ verifier circuit that given an instance $x$ of a problem $A$, decides whether $x \in A_{yes}$ or $x \in A_{no}$. The input to the circuit is the $p$-qubit classical witness state $\ket{y_x}$, and the $q$-qubit ancilla register $\ket{0}^{\otimes q}$, where $p,q$ are two polynomials scaling with $n = \norm{x}$. Denote as $ans$ the logical qubit that when measured provides this decision. The reduction $x \mapsto x'$ requires $2p + q + L+1$ qudits and is the following:

\begin{enumerate}
  \item Choose $2p$ qudits to serve as the logical and witness qudits of the computation, choose other $q$ qudits to serve as the ancilla qudits, and $L+1$ qudits to serve as the clock qudits. For each class of qudit, define an ordering.
  \item For all $i \in [q]$, create a $\piinitzero$ clause acting on $c_0$ and $l_{i}$.
  \item For all $j \in [p]$, create a $\piinitcopy$ clause acting on $c_0$, $l_{q-1+j}$, and $w_j$.
  \item For all $k \in [L]$, create a $\Pi_{prop}$ clause with associated unitary $U_{k+1}$ acting on clock qudits $k$ and $k+1$ (with $k$ being the predecessor of $k+1$) and the two logical qudits that the unitary originally acts on.
  \item Create one $\Pi_{out}$ clause acting on $c_L$ and $l_{ans}$.
\end{enumerate}
\indent
Similarly to the previous problems, the resulting instance is one-dimensional, with no simultaneous propagation clauses and initialized logical qudits (some with $\piinitzero$ clauses and others with $\piinitcopy$).

\subsubsection{Completeness and soundness}

If $x$ is a yes-instance, then the \textsc{Witnessed SLCT-QSAT} instance $x'$ is satisfied with certainty by the history state

\begin{equation*}
  \ket{\psi_{hist}} = \frac{1}{\sqrt{L+1}} \sum_{t=0}^{L} U_t \ldots U_0 \ket{0}^{\otimes q} \otimes \ket{y_x} \otimes \ket{\underbrace{d \ldots d}_t a_t \underbrace{r \ldots r}_{L-t}}.
\end{equation*}
\indent
When $x$ is a no-instance, the only nontrivial change in \cref{eqn:hamilsoundnessfree} happens within $H_{init}$. This term now collects the $(I - \ket{0} \! \bra{0}) \otimes \ket{a} \! \bra{a}$ terms ($q$ of them), the $\ket{?} \! \bra{?} \otimes \ket{a} \! \bra{a}$ terms ($p$ of them), and the $(I-\ket{00}\!\bra{00} - \ket{11}\!\bra{11}) \otimes \ket{a} \! \bra{a}$ terms ($p$ of them). Thus, it is only necessary to show that $H|_{\mathcal{S}_{clock}}$ has no small eigenvalues in both cases where there is at least one logical qudit in the $\ket{?}$ state, and when there's no such state. We note that we cannot simply repeat the argument of Ref.\ \cite{bravyi2006efficient} for the $\QMA$-complete $3$-QSAT problem because the soundness of problem $A$---``for any classical witness state $\ket{y_x}$, $\Pr[U_n \textnormal{ accepts } (x,y_x)] \leq 1/3$''---is restricted to classical states instead of arbitrary quantum states. \\

\medbreak
\noindent
\textbf{At least one undefined logical qudit} \\
\noindent
This case is quite similar to the one in the $\BQP_1$ problem. We still have that the Hamiltonian $H$ within this subspace is $H|_{\mathcal{S}_{clock}} = H_{init} + H_{prop}(T,u)$. Then, it is easy to see that even with the additional terms, $H_{init}$ remains a diagonal matrix and $\gamma(H_{init}) \geq 1$. On the other hand, $H_{prop}$ remains unchanged and so $\gamma(H_{prop}) \geq c/T^2$ for some positive constant $c$. The null spaces of $H_{init}$ and $H_{prop}(T,u)$ are:

\begin{equation*}
  \mathcal{S}_{init} = \mathcal{D}_u \otimes \Span\{\ket{C_1}, \ldots, \ket{C_T}\} \hspace{2em} \textnormal{and} \hspace{2em} \mathcal{S}_{prop} = \mathcal{D}_u \otimes \frac{1}{\sqrt{T+1}} \sum_{t=0}^T \ket{C_t},
\end{equation*}
\noindent
where now $\mathcal{D}_u = (\mathbb{C}^2)^{ \otimes u} \otimes \ket{?}_u \otimes (\mathbb{C}^2)^{\otimes q+p-u-1} \otimes (\mathbb{C}^2)^{\otimes p}$ and $0 \leq u < q + p$ represents the position of the undefined logical qudit. As before, it is straightforward to show that $\alpha = \max_{\ket{\eta} \in \mathcal{S}_{prop}} \bra{\eta} \Pi_{\mathcal{S}_{init}} \ket{\eta} \geq 1-T^{-1}$, and so by the Geometric Lemma, $H|_{\mathcal{S}_{clock(T,u)}} \succeq c(2T^3)^{-1}$. \\

\medbreak
\noindent
\textbf{No undefined logical qudit} \\
\noindent
In this scenario, $H|_{\mathcal{S}_{clock}} = H_{init} + H_{prop} + H_{out}$. Defining $H_1 := H_{init} + H_{out}$ and $H_2 := H_{prop}$ as before, we can find that $\gamma(H_1) \geq 1$, $\gamma(H_2) \geq c/L^2$, and possess null spaces

\begin{equation*}
  \begin{aligned}
    \mathcal{S}_1 = & \ket{0}^{\otimes q} \otimes \Span\{\ket{b} \otimes \ket{b} \mid b \in \{0,1\}^p \} \otimes \ket{C_0} \oplus                         \\
                    & (\mathbb{C}^2)^{\otimes q} \otimes (\mathbb{C}^2)^{\otimes 2p} \otimes \Span\{\ket{C_1}, \ldots, \ket{C_{L-1}}\} \: \oplus          \\
                    & U^\dagger \left[ \ket{1}_{ans} \otimes (\mathbb{C}^2)^{\otimes q+p-1} \right] \otimes (\mathbb{C}^2)^{\otimes p} \otimes \ket{C_L},
  \end{aligned}
\end{equation*}
\noindent
and
\begin{equation*}
  \mathcal{S}_2 = (\mathbb{C}^2)^{\otimes q} \otimes (\mathbb{C}^2)^{\otimes 2p} \otimes \frac{1}{\sqrt{L+1}} \sum_{t=0}^L \ket{C_t}.
\end{equation*}
\indent
For the Geometric Lemma, we then require to calculate the quantity $\alpha = \max_{\ket{\eta} \in \mathcal{S}_2} \bra{\eta} \Pi_{\mathcal{S}_1} \ket{\eta}$. Following the steps of Ref.\ \cite{kitaev2002classical} or Ref.\ \cite{gharibian2015quantum}, and writing projector onto the null space of $\mathcal{S}_1$ as $\Pi_{\mathcal{S}_1} = \Pi_{\mathcal{S}_{1,1}} + \Pi_{\mathcal{S}_{1,2}} + \Pi_{\mathcal{S}_{1,3}}$, we can show that in our case

\begin{equation*}
  \max_{\ket{\eta} \in \mathcal{S}_2} \bra{\eta} \Pi_{\mathcal{S}_{1,2}} \ket{\eta} = \frac{L-1}{L+1},
\end{equation*}
\noindent
while the contribution of the other two terms can be bounded as

\begin{equation*}
  \max_{\ket{\eta} \in \mathcal{S}_2} \bra{\eta} \Pi_{\mathcal{S}_{1,1}} + \Pi_{\mathcal{S}_{1,3}} \ket{\eta} \leq \frac{1}{L+1} \left(1 + \sqrt{\max_{\ket{\beta} \in \mathcal{\tilde{S}}_{1,1}}  \bra{\beta} \Pi_{\mathcal{\tilde{S}}_{1,3}} \ket{\beta}} \right),
\end{equation*}
\noindent
where $\tilde{S}_{1,1} := \ket{0}^{\otimes q} \otimes \Span\{\ket{b} \otimes \ket{b} \mid b \in \{0,1\}^p \}$ and $\tilde{S}_{1,3} := U^\dagger \left[ \ket{1}_{ans} \otimes (\mathbb{C}^2)^{\otimes q-1} \otimes (\mathbb{C}^2)^{\otimes p} \right] \otimes (\mathbb{C}^2)^{\otimes p}$ are the nullspaces $S_{1,1}$ and $S_{1,3}$ without the clock register. Then, as any state $\ket{\beta} \in \mathcal{\tilde{S}}_{1,1}$ can be written as $\ket{\beta} = \ket{0}^{\otimes q} \otimes \sum_{b \in \{0,1\}^p} \alpha_b \ket{b} \otimes \ket{b}$ with $\sum_{b} \norm{\alpha_b}^2 = 1$, we have that

\begin{equation*}
  \begin{aligned}
    \max_{\ket{\beta} \in \mathcal{\tilde{S}}_{1,1}}  \bra{\beta} \Pi_{\mathcal{\tilde{S}}_{1,3}} \ket{\beta} & = \sum_{b,b' \in \{0,1\}^p} \alpha_{b'}^* \alpha_b \bra{0}^{\otimes q} \otimes \bra{b'} \otimes \bra{b'} \left[ U^\dagger (\ket{1}\!\bra{1}_{ans} \otimes I) U \otimes I \right] \ket{0}^{\otimes q} \otimes \ket{b} \otimes \ket{b} \\
                                                                                                              & = \sum_{b,b' \in \{0,1\}^p} \alpha_{b'}^* \alpha_b \bra{0}^{\otimes q} \otimes \bra{b'} \left[ U^\dagger (\ket{1}\!\bra{1}_{ans} \otimes I) U \right] \ket{0}^{\otimes q} \otimes \ket{b} \otimes \braket{b'|b}                      \\
                                                                                                              & = \sum_{b \in \{0,1\}^p} \norm{\alpha_b}^2 \bra{0}^{\otimes q} \otimes \bra{b} \left[ U^\dagger (\ket{1}\!\bra{1}_{ans} \otimes I) U \right] \ket{0}^{\otimes q} \otimes \ket{b}                                                     \\
                                                                                                              & \leq \epsilon \sum_{b \in \{0,1\}^p} \norm{\alpha_b}^2                                                                                                                                                                               \\
                                                                                                              & = \epsilon,
  \end{aligned}
\end{equation*}
\noindent
where in the fourth line we made use of the fact that $A$ is a problem in QCMA and accepts with probability less than $\epsilon = 1/3$ when $x$ is a no-instance. In other words, we observed that $\Pr[U_n \textnormal{ accepts } (x,b)] = \bra{0}^{\otimes q} \otimes \bra{b} U^\dagger (\ket{1}\!\bra{1}_{ans} \otimes I)U \ket{0}^{\otimes q} \otimes \ket{b} \leq \epsilon$ for any classical witness $b \in \{0,1\}^p$. Thus,

\begin{equation*}
  \alpha = \max_{\eta \in \mathcal{S}_2} \bra{\eta} \Pi_{\mathcal{S}_{1,1}} + \Pi_{\mathcal{S}_{1,2}} + \Pi_{\mathcal{S}_{1,3}} \ket{\eta} \leq \frac{L-1}{L+1} + \frac{1 + \sqrt{\epsilon}}{L+1} = 1 - \frac{1+\sqrt{\epsilon}}{L+1},
\end{equation*}
\noindent
and so, plugging this in the Geometric Lemma, we obtain that $H|_{S_{clock}} \succeq 2^{-1} \min\{\gamma(H_1), \gamma(H_2)\} (1 - \alpha) \succeq  c(1 - \sqrt{\epsilon})(2L^3)^{-1}$.

\section{\coRP-complete Problem} \label{section:corp}

By combining elements of the previous constructions, we can generate a problem that is $\coRP$-complete.

\subsection{Construction overview}

In the $\BQP_1$ constructions of \cref{section:monogamy,section:our}, we made use of a quantum circuit in order to decide the satisfiability of ``quantum'' sub-instances. In particular, this circuit verified the satisfiability of the simultaneous $\Pi_{prop}$ clauses and final $\Pi_{out}$ clauses. For the latter, the circuit executed the quantum circuit $U_T \ldots U_1 \ket{0}^{\otimes q}$, measured some of the qubits, and accepted or rejected depending on the measurement outcomes. Intuitively, to generate the $\coRP$ complete problem, we would like to replace the universal quantum circuit by a universal classical \textit{reversible} circuit $R = R_T \ldots R_1$ (reversibility is needed since the best potentially satisfying state is still a \textit{quantum} history state) and introduce randomness into the instance by initializing $p$ ancilla qubits to $\ket{+}$. Then, for these new sub-instances, we could analogously verify the $\Pi_{out}$ clauses by sampling a bitstring $b \in \{0,1\}^p$, evaluating the circuit $R$ on input $(0^q,b)$ and deciding on its satisfiability based on the final state of the bits. While this idea is close to the actual construction, it is not sufficient to decide the satisfiability of the simultaneous propagation clauses.

For the construction to work, we also incorporate elements of the $\QCMA$ problem from the previous section. Namely, we will modify the $\piinitcopy$ clause so it initializes both a witness (now referred to as \textit{auxiliary} qudit as we remove the freedom) and a logical qudit to the maximally entangled state $\ket{\Phi^+}$. This new clause is denoted $\piinitpair$.

At a later time we will demonstrate that this idea does allow verification of the simultaneous propagation clauses. For the moment, let us show in slightly more detail that this construction does allow us to verify the satisfiability of $\Pi_{out}$ clauses. Again consider the toy example of \cref{fig:qcmadrawings} where now the $\piinitcopy$ clauses are replaced by $\piinitpair$ clauses and the unitaries $U_i$ with reversible classical gates $R_i$. If satisfiable, the satisfying state must be the history state

\begin{equation*}
    \begin{aligned}
        \ket{\psi_{hist}} & = \frac{1}{\sqrt{5}} \sum_{t=0}^4 \big[ R_t \ldots R_{0} \ket{00} \otimes \ket{\Phi^+}^{\otimes 2} \big] \otimes \ket{\underbrace{d \ldots d}_t a_t \underbrace{r \ldots r}_{4-t}}  \\
                          & = \frac{1}{\sqrt{5}} \sum_{t=0}^4 \sum_{b \in \{0,1\}^2} \frac{1}{2} \ket{\xi_b^t} \otimes \ket{b}_{Aux} \otimes \ket{\underbrace{d \ldots d}_t a_t \underbrace{r \ldots r}_{4-t}},
    \end{aligned}
\end{equation*}
\noindent
where in the second line we observed that $\ket{\Phi^+}^{\otimes p}$ can be written as $\ket{\Phi^+}^{\otimes{p}} = 2^{-\frac{p}{2}}\sum_{b \in \{0,1\}^p} \ket{b}_L \otimes \ket{b}_{Aux}$ for any $p \in \mathbb{N}$, and defined $\ket{\xi_b^t} := R_t \ldots R_1 \ket{00} \otimes \ket{b}_L$. The $\Pi_{out}$ clause is satisfied if at time $t = 4$, the probability that the second qubit yields outcome ``1'' when measured is $1$. This probability is

\begin{equation*}
    \begin{aligned}
        \Pr(\textnormal{outcome } 1) & = \frac{1}{4} \sum_{b,b' \in \{0,1\}^2} \bra{\xi_{b'}^4} \otimes \bra{b'} \Pi^{(1)} \ket{\xi_b^4} \otimes \ket{b} \\
                                     & = \frac{1}{4} \sum_{b \in \{0,1\}^2} \bra{\xi_b^4} \Pi^{(1)} \ket{\xi_b^4},
    \end{aligned}
\end{equation*}
\noindent
from where it is evident that if the instance is satisfiable, $\bra{\xi_b^4} \Pi^{(1)} \ket{\xi_b^4} = 1$ for all $b \in \{0,1\}^2$. Hence, this shows that it is possible to verify the $\Pi_{out}$ clause by sampling one of the strings $b$, running circuit $R$ on input $(0^2, b)$, and measuring the state of the second qubit.

Another important consideration is that the classical reversible gate set must be chosen with care. In previous constructions, $\mathcal{G}$ was chosen for two reasons: first, so all gates in the set changed the basis states upon application, and second, so the unitary $V(\Pi_i)$, or alternatively the circuit $\mathcal{C}$ of \cref{fig:measurecircuit}, could be implemented perfectly with gates from this set. The first property played an important role in the proof of \cref{lemma:sharedhprop}, and in this construction we still desire to make use of this lemma. However, as we will soon show, circuit $\mathcal{C}$ will no longer be necessary. To make sure that the first property is satisfied, we choose

\begin{equation*}
    \mathcal{G} = \{X, (X \otimes X \otimes X) \textnormal{Toffoli}\},
\end{equation*}
\noindent
which is a universal gate set for reversible classical computation as the Toffoli alone is known to be universal. As a consequence of this choice, the quantum satisfiability problem we define below will have some $5$-local clauses since the unitary gate in the $\Pi_{prop}$ clauses may be a Toffoli. This is the best locality we can achieve as it is also well-known that any universal gate set for reversible quantum computation must include a $3$-bit gate.

Now that we have presented the general idea, let us define the problem more formally.

\subsection{Problem definition}

The problem we claim is $\coRP$-complete is the following:

\begin{defn}[\textsc{Classical SLCT-QSAT}] \label{defn:cslctqsat}
    The problem {\rm \textsc{Classical SLCT-QSAT}} is a QCSP defined on the $8$-dimensional Hilbert space.

    \begin{equation*}
        \mathcal{H} = \Span\{\ket{0_L},\ket{1_L},\ket{?_L}\} \oplus \Span\{\ket{0_{Aux}},\ket{1_{Aux}}\} \oplus \Span\{\ket{r_C},\ket{a_C},\ket{d_C}\},
    \end{equation*}
    which now consists of a logical, an auxiliary, and a clock subspace. Furthermore, here $\mathcal{G} = \{X, (X \otimes X \otimes X) \textnormal{Toffoli}\}$, and $\Pi_{prop}$ now acts on at most $5$ qudits. Besides the operators $O_{init}$, $O_{prop}$ and $O_{out}$ defined in \cref{eqn:Oinitfree,eqn:Opropfree,eqn:Ooutfree}, the problem contains a new $3$-local operator

    \begin{equation*}
        O_{init}^{\ket{\Phi^{\scaleto{+}{3pt}}}} := \Pi_{L,1} + \Pi_{Aux,2} + \Pi_{C,3} + \Pi_{start,3} + \ket{?}\!\bra{?}_1 \otimes \ket{a}\!\bra{a}_3 + (I - \ket{\Phi^+} \! \bra{\Phi^+})_{1,2} \otimes \ket{a} \! \bra{a}_3,
    \end{equation*}
    \noindent
    where $\Pi_{Aux} := I - (\ket{0} \! \bra{0}_{Aux} + \ket{1} \! \bra{1}_{Aux})$ is the projector to the auxiliary subspace. As before, $\piinitpair$ refers to the projector with the same null space as this operator.

    The promise remains the same as before. We are promised that for every instance considered either (1) there exists a {\rm quantum} state $\ket{\psi_{sat}}$ on $n$ $8$-dimensional qudits such that $\Pi_i \! \ket{\psi_{sat}} = 0$ for all $i$, or (2) $\Sigma_i \bra{\psi} \! \Pi_i \! \ket{\psi} \geq 1/poly(n)$ for all {\rm quantum} states $\ket{\psi}$.

    The goal is to output ``YES'' if (1) is true, or output ``NO'' otherwise.
\end{defn}

Before showing that \textsc{Classical SLCT-QSAT} is contained within $\coRP$ and is also $\coRP$-hard, let us make an observation about the problem.

\begin{remark}
    While the projectors of the problem are defined with a restricted set of gates, the algorithm that we use to decide the instance may use any reversible or irreversible classical gate set. Indeed, the circuit $R$ can be perfectly simulated by any other circuit employing a universal gate set with polynomial overhead. As mentioned earlier, this is not a feature present in quantum circuits.

    Thus, the claim is actually that this problem is complete for $\coRP$, rather than $\coRP^\mathcal{G}$, as in the $\BQP_1$ problems.
\end{remark}

\subsection{In \coRP} \label{subsubsection:incorp}

With the introduction of the new $\piinitpair$ clause, we have to verify that it does not break anything from the analysis of clock components and shared logical qudits of the construction of \cref{section:our}.

In the analysis of clock components of \cref{subsubsection:clockfree}, the lemmas were stated in terms of $\Pi_{init}$ and $\Pi_{out}$ clauses. However, the elements of these clauses that mattered for this analysis were the $\Pi_{start}$ and $\Pi_{stop}$ projectors within, which together required that one of the clock qudits in the chain was in an active state. Since $\piinitpair$ also contains the $\Pi_{start}$ projector, for these particular statements, $\piinitzero$ and $\piinitpair$ clause are equivalent. Thus, there are no new cases to consider in the analysis of clock components. On the other hand, there are four new cases to consider concerning the logical qudits of an instance: two for the logical qudits of a single clock component, and two slightly different versions of these where the logical qudit is shared across sub-instances.

\begin{lemma} \label{lemma:piinitplus1}
    Consider a clock component that has at least one $\piinitdot$ and one $\Pi_{out}$ clause. If a logical qudit is acted on by both a $\piinitzero$ and a $\piinitpair$ clause, stemming from the same clock qudit, the instance is unsatisfiable.
\end{lemma}

\begin{proof}
    First, recall that \cref{prop:activetacc} shows that if the instance is satisfiable, both clauses must enforce their constraints on the logical qudit. Satisfying the $\piinitpair$ clause requires the logical and witness qudits to be in the state $\ket{\Phi^+}$. It is easy to see that $(I - \ket{0} \! \bra{0})_1 \ket{\Phi^+} \neq 0$ and so the $\Pi_{init}$ clause cannot be satisfied.
\end{proof}

\begin{lemma} \label{lemma:piinitplus2}
    Consider two clock components, both of which have at least one $\piinitdot$ and one $\Pi_{out}$ clause. If one clock component acts on a logical qudit with a $\piinitzero$ clause, and the other component also acts on this qudit with a $\piinitpair$ clause, the instance is unsatisfiable.
\end{lemma}

\begin{proof}
    Just like in the previous proof, \cref{prop:activetacc} shows that if the instance is satisfiable, both clauses must enforce their constraints on the logical qudit. Now, recall that if the qudit is also present in a $\Pi_{prop}$ clause from any either clock component, the instance is unsatisfiable by \cref{lemma:sharedhprop}. Therefore, the qudit must then remain fixed to one state. Again, the qudit cannot be both $\ket{0}$ and part of a maximally entangled state $\ket{\Phi^+}$, and so the instance is unsatisfiable.
\end{proof}
\indent
One might think that there is a third scenario with a similar statement, which is when the $\piinitzero$ and $\piinitpair$ clause stem from different clock qudits of the same clock component (which may occur given the strangely shaped sub-instances like those of \cref{fig:semilinearsat}). For this case, we simple observe that by the analysis of \cref{subsubsection:clockfree}, if the sub-instance is satisfiable, these two clock qudits can be thought of as forming part of two smaller disjoint clock components. This is the scenario of \cref{lemma:piinitplus2}.

The other two new cases are:

\begin{lemma} \label{lemma:piinitplus3}
    Consider a clock component that has at least one $\piinitdot$ and one $\Pi_{out}$ clause. If two $\piinitpair$ clauses act on the same clock and auxiliary qudits, but act on distinct logical qudits, the instance is unsatisfiable. Similarly, if two $\piinitpair$ clauses act on the same clock and logical qudits, but act on distinct auxiliary qudits, the instance is unsatisfiable.
\end{lemma}

\begin{proof}
    This follows from monogamy of entanglement. First, observe that since the $\piinitpair$ clauses must enforce their constraints on the auxiliary and logical qudits (\cref{prop:activetacc}), the logical qudits cannot be $\ket{?}$ due to the $\ket{?}\!\bra{?} \otimes \ket{a}\!\bra{a}$ projector. Hence, the logical qudits become qubits. Then, in the former scenario, the auxiliary qubit is required to be maximally entangled with two logical qudits. In the latter scenario, it is the logical qubit that is required to be maximally entangled with two auxiliary qubits. Both violate monogamy of entanglement and so the instance is unsatisfiable.
\end{proof}

\begin{lemma} \label{lemma:piinitplus4}
    Consider two clock components, both of which have at least one $\piinitdot$ and one $\Pi_{out}$ clause. If one clock component acts on a logical qudit with a $\piinitpair$ clause, and the other component also acts on it with a $\piinitpair$ clause, the instance is unsatisfiable. Similarly, if this is also the case for an auxiliary qudit, the instance is unsatisfiable.
\end{lemma}

\begin{proof}
    The proof is the same as for \cref{lemma:piinitplus3} if we note that due to \cref{lemma:sharedhprop}---which states that a shared qudit cannot be acted on by any $\Pi_{prop}$ clause---these qudits must always remain fixed to a single state.
\end{proof}
\indent
Now, let us address the simultaneous propagation clauses. In \cref{prop:simultaneousprop}, we showed that $k$ simultaneous $\Pi_{prop}$ clauses were satisfiable iff $U_{t+1,i} \ket{\phi_t}= U_{t+1,j} \ket{\phi_t} \; \; \forall i,j \in \{1,\ldots,k\}$, where $\ket{\phi_t} = U_t \ldots U_1 \ket{\phi_0}$ and $\ket{\phi_0}$ was the initial state of the computation. In this construction we can extend this result by making use of the new clauses of the instance and the reversible nature of the gate set.

\begin{lemma}\label{lemma:simultaneouspropcorp}
    Suppose $k$ distinct $\Pi_{prop}$ clauses act on clock qudits $c_t$ and $c_{t+1}$ with $0 \leq t < T$, and let $R_{t+1,1}, \ldots, R_{t+1,k}$ denote the reversible gates associated with these $k$ clauses. The clauses are satisfiable iff $R_{t+1,i} \ket{\xi_b^t} = R_{t+1,j} \ket{\xi_b^t}$ for all $i,j \in \{1,\ldots,k\}$ and for all $b \in \{0,1\}^p$.
\end{lemma}

\begin{proof}
    Since the classical reversible gates are unitaries, \cref{prop:simultaneousprop} still holds and these simultaneous propagation clauses are satisfiable iff $R_{t+1,i} \ket{\phi_t} = R_{t+1,j} \ket{\phi_t}$ for all $i,j \in \{1,\ldots,k\}$. For an initial state with $p$ $\ket{\Phi^+}$ states, we can write the intermediate state of the computation as $\ket{\phi_t} = 2^{-\frac{p}{2}} \sum_{b \in \{0,1\}^p} \ket{\xi_b^t} \otimes \ket{b}$, and so the clauses are satisfiable iff

    \begin{equation*}
        \begin{aligned}
            1 & = \bra{\phi_t} R^\dagger_{t+1,j} R_{t+1,i} \ket{\phi_t}                                                                            \\
              & = \frac{1}{2^p} \sum_{b,b' \in \{0,1\}^p} \bra{\xi_b^t} \otimes \bra{b'} R^\dagger_{t+1,j} R_{t+1,i} \ket{\xi_b^t} \otimes \ket{b} \\
              & = \frac{1}{2^p} \sum_{b \in \{0,1\}^p} \bra{\xi_b^t} R^\dagger_{t+1,j} R_{t+1,i} \ket{\xi_b^t}.
        \end{aligned}
    \end{equation*}
    \noindent
    Then, it is evident that for the right hand-side to be equal to $1$, $R_{t+1,i} \ket{\xi_b^t} = R_{t+1,j} \ket{\xi_b^t}$ for all basis states $b \in \{0,1\}^p$.
\end{proof}
\indent
In the algorithm of previous constructions, we had to evaluate the satisfiability of the simultaneous propagation clauses by executing a quantum circuit for every pair of clauses. Now, this new lemma will allow us to decide these constraints classically. This will be explained in detail at a later time.

The discussion for active dots remains the same as before.

Finally, let us make a small observation. In \cref{remark:trivial}, we mentioned two ``quantum'' sub-instances that were actually trivially decidable but were intentionally left for the quantum algorithm to decide. Because of the $\piinitpair$ clause, we now have an extra two such sub-instances. These are sub-instances with no undefined or simultaneous propagation clauses, where the logical qudits present in $\Pi_{out}$ clauses are (1) initialized to one half of a Bell pair $\ket{\Phi^+}$ but not acted by any $\Pi_{prop}$ clause, or $(2)$ not initialized. In the first case, the instance is trivially unsatisfiable since its reduced density matrix $\rho = 1/2 (\ket{0}\!\bra{0} + \ket{1} \! \bra{1})$ has support on $\ket{0}$. As before, the second case is trivially satisfiable. All four sub-instances will be decided in the randomized classical algorithm explained below.

\subsubsection{Algorithm}
The first step of the algorithm remains the same as in \cref{subsubsection:algorithm}. The second step requires a small change to indicate that either type of initialization clause is allowed.

\begin{enumerate}[label=(2)]
    \item \textit{Ignore the trivial sub-instances.} For every sub-instance, check whether it contains both $\piinitdot$ and $\Pi_{out}$ clauses. If it does not, ignore all clauses of the sub-instance for the rest of the algorithm. Correctness is given by \cref{lemma:noendpoints}.
\end{enumerate}
\noindent
We now add two new steps between (2) and (3). One to account for \cref{lemma:piinitplus1,lemma:piinitplus2}, and the second to account for \cref{lemma:piinitplus3,lemma:piinitplus4}.

\begin{enumerate}[label=(coRPa), leftmargin=2.2cm]
    \item For every logical qudit in the instance, check whether it is acted on by both a $\piinitzero$ and a $\piinitpair$ clause. If such a qudit exists, reject. Correctness is given by \cref{lemma:piinitplus1,lemma:piinitplus2}.
\end{enumerate}

\begin{enumerate}[label=(coRPb), leftmargin=2.2cm]
    \item For every logical qudit check whether it is connected to several auxiliary qudits. If it is, reject. Similarly, for every auxiliary qudit check whether it is connect to several logical qudits. If it is, reject. Correctness is given by \cref{lemma:piinitplus3,lemma:piinitplus4}.
\end{enumerate}
\noindent
These steps must take place in this specific point in the algorithm---after processing the trivial sub-instances, but before assigning the state of the logical qudits, which occurs in the following step.\footnote{Steps (coRPa) and (coRPb) also process cases related to the shared logical qudits of distinct clock components, which before would only be processed in step (4). Evaluating these sooner does not have any undesired effect.} Afterwards, step (3) requires a small modification to incorporate the $\piinitpair$ clause:

\begin{enumerate}[label=(3)]
    \item \textit{Identify the undefined $\Pi_{prop}$ clauses.} For every logical qudit, check whether it is acted on by at least one $\piinitzero$ or $\piinitpair$ clause. If it is not, mark all $\Pi_{prop}$ clauses this qudit is part of as undefined.
\end{enumerate}
\noindent
Subsequently, steps (4), (5), and (7) are the same as before. In step (6), we replace the quantum algorithm with a classical subroutine acting on ancilla and randomly sampled bits. Let us explain this step in more detail. \\

\medbreak
\noindent
\textbf{Classical subroutine} \\
\noindent
As before, the inputs to this classical subroutine are the sub-instances of \cref{fig:qinputs}, except now the unitaries are reversible classical gates. We let $p$ be the number of logical qudits present in the $\piinitpair$ clauses of a sub-instance and $q$ the number of logical qudits present in $\piinitzero$ clauses. Additionally, let $T$ denote the length of the sub-instance and $k_{t,t+1}$ denote the number of simultaneous $\Pi_{prop}$ clauses acting on clock qudits $c_t$ and $c_{t+1}$.

Ultimately, the goal of the quantum algorithm in previous constructions was to evaluate whether the simultaneous $\Pi_{prop}$ clauses and the final $\Pi_{out}$ clauses could be satisfied, as all other clauses could be determined to be satisfiable. To do so, we discussed an approach based on Bravyi's ideas for $k$-QSAT, which relied on generating the history state and measuring the eigenvalues of all projectors. This method nicely related to the promise of the problem which allowed us to demonstrate the $\BQP_1$ and $\QCMA$ soundness of the algorithm. However, we mentioned that the resulting algorithm was difficult to implement and wasteful. Instead, we calculated the acceptance/rejection probabilities of measuring the eigenvalue of a clause, and developed a more efficient method that reproduced the same statistics. To decide two simultaneous propagation clauses, we applied circuit $\mathcal{C}$ (more accurately $\mathcal{C}_{\mathcal{G}_8}$) shown in \cref{fig:measurecircuit}, and to decide $\Pi_{out}$ clauses, we executed the circuit and at the end of the computation measured the qubits present in these clauses. In this construction, we must rethink the algorithm since we cannot perform Bravyi's method with a classical computer as it requires creating the highly-entangled history state, and also, we cannot the repeat the same algorithm as before because $\mathcal{C}$ is a quantum circuit. We hence design a classical algorithm that yields the same probability as measuring the eigenvalue of a simultaneous $\Pi_{prop}$ clause.

\begin{claim}\label{claim:measuringpropcorp}
    Let $\Pi_{prop}^{(c_t,c_{t+1})}$ be a clause acting on clock qudits $c_t,c_{t+1}$ with associated reversible classical gate $R_{t+1, j}$, where $j \in [k_{t,t+1}]$. Measuring the eigenvalue of projector $\Pi_{prop}^{(c_t,c_{t+1})}$ in the state $\ket{\psi_{phist}}$ is equivalent to the following classical steps:
    \begin{enumerate}
        \item Create the basis state $\ket{\xi_b^t} = R_{t,0} \ldots R_{0,0} \ket{0}^{\otimes q} \otimes \ket{b}$, where $b \in \{0,1\}^p$ is a bitstring sampled uniformly at random.
        \item Flip a coin and compare $R_{t+1,0} \ket{\xi_b^t}$ with $R_{t+1,j} \ket{\xi_b^t}$. Reject if the coin is ``Heads'' and $R_{t+1,0} \ket{\xi_b^t} \neq R_{t+1,j} \ket{\xi_b^t}$.
    \end{enumerate}
\end{claim}

\begin{proof}
    In the proof of \cref{claim:measuringprop}, we showed that the probability of measuring eigenvalue $1$ for this propagation clause was given by $\Pr (\textnormal{outcome } 1) = 2^{-1}(1 - \bra{\phi_t} U^\dagger_{t+1,j} U_{t+1,0} \ket{\phi_t})$. Then, by a similar calculation to that in the proof of \cref{lemma:simultaneouspropcorp}, in this construction, this probability becomes $\Pr (\textnormal{outcome } 1) = 2^{-1}(1 -  2^{-p} \sum_{b \in \{0,1\}^p} \bra{\xi_b^t} R^\dagger_{t+1,j} R_{t+1,0} \ket{\xi_b^t})$.

    To show that the classical steps above also reject the instance with this probability, first observe that $2^{-p} \sum_{b \in \{0,1\}^p} \bra{\xi_b^t} R^\dagger_{t+1,j} R_{t+1,0} \ket{\xi_b^t}$ is the fraction of bitstrings $b$ for which the basis state $\ket{\xi_b^t}$ satisfies $R_{t+1,0} \ket{\xi_b^t} = R_{t+1,j} \ket{\xi_b^t}$. In other words, this is the probability that in the algorithm we sample a bitstring $b$ for which $R_{t+1,0} \ket{\xi_b^t} = R_{t+1,j} \ket{\xi_b^t}$. Denote this event as $A$. Then, the probability the algorithm rejects, i.e.\ observes both ``Heads'' and $A^C$, is $\Pr(\textnormal{Heads} \cap A^C) = \Pr({\rm Heads}) \Pr(A^C) = 2^{-1} (1 - 2^{-p} \sum_{b \in \{0,1\}^p} \bra{\xi_b^t} R^\dagger_{t+1,j} R_{t+1,0} \ket{\xi_b^t})$.\footnote{$A^C$ denotes the event where $A$ does not occur.} Thus, both methods yield the same statistics and are equivalent.
\end{proof}
\indent
Making use of this claim, the subroutine is then the following:

\begin{enumerate}[label=(6)]
    \item For every sub-instance, do the following:
          \begin{enumerate}[label=(6.\arabic*)]
              \item Create the initial state $\ket{\xi_b^0}$ by sampling $p$ bits at random and initializing another $q$ bits to ``0''.
              \item For every $t \in [T]$, do the following:
                    \begin{enumerate}[label=(6.2.\arabic*)]
                        \item If $k_{t,t+1} \geq 1$:
                              \begin{itemize}
                                  \item[--] Choose an index $j \in [k_{t,t+1}]$. Create another copy of the current state of the computation to obtain two copies of $\ket{\xi_b^t}$. Apply $R_{t+1,0}$ to the first, and $R_{t+1,j}$ to the second.
                                  \item[--] Remove the $\Pi_{prop}$ clause with unitary $R_{t+1,j}$ from the instance.
                                  \item[--] Flip a coin. If $R_{t+1,0} \ket{\xi_b^t} \neq R_{t+1,j} \ket{\xi_b^t}$ and the coin is ``Heads'', reject. Otherwise, go back to step (6.1).
                              \end{itemize}
                        \item Otherwise, apply $R_{t+1,0}$ to the state to obtain $\ket{\xi_b^{t+1}}$.
                    \end{enumerate}
              \item \textit{At this step, the algorithm has produced the basis state $\ket{\xi_b^T}$.} For every logical qudit in the instance with a $\Pi_{out}$ clause, check if any of the corresponding logical bits are ``0''. If there is a such bit, reject.
          \end{enumerate}
\end{enumerate}
\indent
We note that although in this construction there is no need to restart the algorithm after checking whether two simultaneous propagation clauses can be simultaneously satisfied, we choose to do so in order to maintain some resemblance with the algorithms of the previous quantum satisfiability problems.

\subsubsection{Runtime}

The small changes in steps (2) and (3) and the new steps (coRPa), (coRPb), and (6) do not affect the algorithm's runtime significantly and so it still runs in polynomial time. In particular, in steps (coRPa) and (coRPb), the algorithm iterates over all $\mathcal{O}(n)$ logical qudits of the instance, and for each qudit, checks whether it is acted on by a $\piinitzero$ or $\piinitpair$ clause. These steps take $\mathcal{O}(poly(n))$ time. For step (6.1), observe that since $p \in \mathcal{O}(n)$, this step takes $\mathcal{O}(n)$ time. Step (6.2) is similar to step (6.2) of the first construction (see \cref{subsection:algorithmmono}), however, instead of executing circuit $C$ to evaluate the satisfiability of the simultaneous propagation clauses, here we create two copies of the state, apply a reversible gate, and compare each pair of bits. These steps clearly also take $\mathcal{O}(n)$ time, and given that there are at most $\mathcal{O}(poly(n))$ simultaneous propagation clauses, step (6.2) takes $\mathcal{O}(poly(n))$ time for the whole sub-instance. In step (6.3), the algorithm iterates over a subset of the logical qudits in the sub-instance and checks whether they are ``0''. This step takes $\mathcal{O}(n)$ time. Overall, the algorithm takes $\mathcal{O}(poly(n))$ time to decide each sub-instance, and given that there are at most $\mathcal{O}(poly(n))$ sub-instances, step (6) also takes $\mathcal{O}(poly(n))$ time.

\subsubsection{Completeness and soundness}

The only difference between this construction and the one of \cref{section:our}, is the new $\piinitpair$ clause. As mentioned in the beginning of \cref{subsubsection:incorp}, these clauses do not introduce new ways to create satisfiable or unsatisfiable instances in the analysis of clock components, but do create four new cases when analyzing the shared logical qudits. The $\coRP$ algorithm mentioned above is an extension of the one from \cref{section:our} and considers the four new cases. Based on the correctness of the algorithm in that previous section and \cref{lemma:piinitplus1,lemma:piinitplus2,lemma:piinitplus3,lemma:piinitplus4}, trivially unsatisfiable and trivially satisfiable sub-instances are verified deterministically and meet the completeness and soundness conditions. These are decided in steps (1)-(5) of the algorithm above.

The three types of non-trivial sub-instances are similar to those shown in \cref{fig:qinputs}. Our task is to demonstrate that if these sub-instances have a quantum satisfying state, we can reach this conclusion classically with prefect completeness and bounded soundness.

Recall that in \cref{section:monogamy}, we showed that the algorithm's acceptance/rejection probabilities (and therefore completeness and soundness) depended only on the outcomes of measuring the eigenvalues of the simultaneous $\Pi_{prop}$ and final $\Pi_{out}$ clauses of each sub-instance. In \cref{claim:measuringpropcorp} we showed that measuring the eigenvalues of the simultaneous propagation clauses could be done with a $\coRP$ circuit, and the same can be shown for the $\Pi_{out}$ clauses by a straightforward generalization of the analysis presented at the beginning of the section for longer chains.\footnote{In contrast to the previous $\QCMA$ problem where there were more differences between the toy instance and arbitrary sub-instances, here the sub-instances maintain the same structure: a logical qudit acted on by a $\piinitpair$ clause cannot be acted by another such clause or a $\piinitzero$ clause.} These ideas are implemented in step (6) of the algorithm above. Therefore, we can reach the same conclusions as in \cref{section:monogamy}: if $x$ is a yes-instance, the algorithm's acceptance probability is $AP=1$, while in the case that $x$ is a no-instance, $AP \leq 1-1/poly(n)$ which can then be improved to $\leq 1/3$ by repeating the algorithm several times and taking the majority vote.

\subsection{\coRP-hard}

\begin{figure}
    \begin{center}
        \begin{tikzpicture}
            \node (circuit1) at (-5.6,0) {
            \Qcircuit @C=1em @R=0.6em {
            & \qw  & \multigate{5}{\hspace{0.5em} Q_x \hspace{0.5em}} & \qw & \qw &\rstick{\hspace{-1em} 0/1?} \\
            & \qvdotss & \nghost{\hspace{0.5em} Q_x \hspace{0.5em}} & \qvdotss &   \\
            & \qw  & \ghost{\hspace{0.5em} Q_x \hspace{0.5em}} & \qw & \qw\\
            & \qw & \ghost{\hspace{0.5em} Q_x \hspace{0.5em}} & \qw & \qw \\
            & \qvdotss & \nghost{\hspace{0.5em} Q_x \hspace{0.5em}} & \qvdotss &   \\
            & \qw & \ghost{\hspace{0.5em} Q_x \hspace{0.5em}} & \qw & \qw
            \inputgroupv{1}{3}{1.2em}{1.4em}{0^n}
            \inputgroupv{4}{6}{1.2em}{1.4em}{r}
            }
            };

            \node (arrow1) at (-3.1,0) {\scalebox{1.5}{\(\Longleftrightarrow\)}};

            \node (circuit2) at (0,0) {
            \Qcircuit @C=1em @R=0.6em {
            & \qw  & \multigate{5}{\hspace{0.5em} R_x \hspace{0.5em}} & \qw & \qw & \rstick{\hspace{-1em} 0/1?} \\
            & \qvdotss & \nghost{\hspace{0.5em} R_x \hspace{0.5em}} & \qvdotss &   \\
            & \qw  & \ghost{\hspace{0.5em} R_x \hspace{0.5em}} & \qw & \qw\\
            & \qw & \ghost{\hspace{0.5em} R_x \hspace{0.5em}} & \qw & \qw \\
            & \qvdotss & \nghost{\hspace{0.5em} R_x \hspace{0.5em}} & \qvdotss &   \\
            & \qw & \ghost{\hspace{0.5em} R_x \hspace{0.5em}} & \qw & \qw
            \inputgroupv{1}{3}{1.2em}{1.4em}{0^q}
            \inputgroupv{4}{6}{1.2em}{1.4em}{r}
            }
            };

            \node (arrow2) at (2.7,0) {\scalebox{1.5}{\(\Longleftrightarrow\)}};

            \node (circuit3) at (6.7,0) {
            \Qcircuit @C=1em @R=0.6em {
            & \qw  & \multigate{5}{\hspace{0.5em} R_x \hspace{0.5em}} & \qw & \meter \\
            & \qvdotss & \nghost{\hspace{0.5em} R_x \hspace{0.5em}} & \qvdotss &   \\
            & \qw  & \ghost{\hspace{0.5em} R_x \hspace{0.5em}} & \qw & \qw\\
            & \qw & \ghost{\hspace{0.5em} R_x \hspace{0.5em}} & \qw & \qw \\
            & \qvdotss & \nghost{\hspace{0.5em} R_x \hspace{0.5em}} & \qvdotss &   \\
            & \qw & \ghost{\hspace{0.5em} R_x \hspace{0.5em}} & \qw & \qw \\
            &&&&\\
            & \qw & \qw & \qw & \qw \\
            &&&&\\
            & \qvdotss & & \qvdotss & \\
            &&&&\\
            & \qw & \qw & \qw & \qw
            \inputgroupv{1}{3}{1.2em}{1.4em}{\ket{0}^{\otimes q} \;\;\;\;\;\;\;}
            \inputgroupv{4}{12}{1.2em}{3.3em}{\ket{\Phi^+}^{\otimes p} \;\;\;\;\;\;\;\;\;}
            }
            };
        \end{tikzpicture}
    \end{center}
    \caption{Circuits illustrating the equivalence between $\coRP$ and $\coRP_q$. Left: A $\coRP$ circuit that decides an instance $x$ of an arbitrary promise problem. $Q_x$ is a deterministic classical circuit acting on $n$ ancilla bits and a random bitstring $r \sim \text{Unif}(\{0,1\}^p)$ of size $p(n)$. Middle: A $\coRP$ circuit where $Q_x$ is replaced by a classical reversible circuit $R_x$ now acting on $q(n) = n + poly(n)$ ancilla bits. Right: a $\coRP_q$ quantum circuit where $R_x$ now acts on $q(n)$ ancilla qubits and on the first subsystem of $p(n)$ maximally entangled states.}
    \label{fig:corpcircs}
\end{figure}
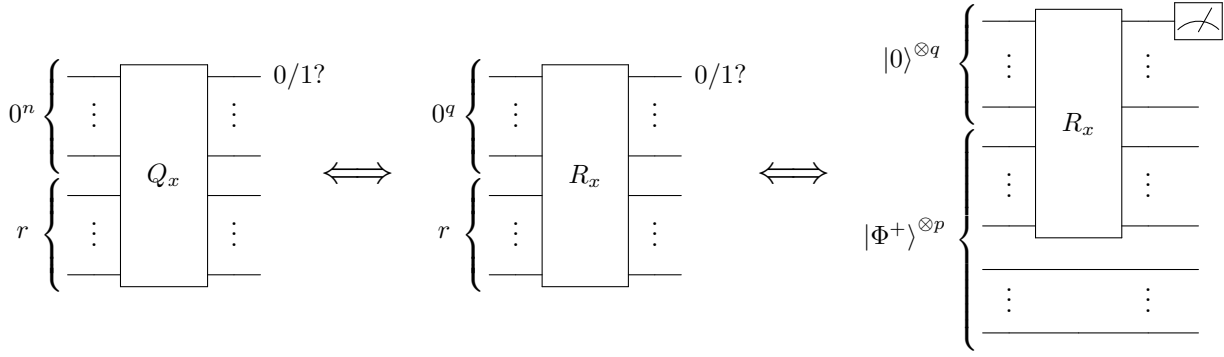

To show that \textsc{Classical SLCT-QSAT} is $\coRP$-hard and complete the proof of \cref{thm:corp}, we use similar ideas to Bravyi \etal \cite{stoquasticLH} when showing that the \textit{stoquastic} LH problem is $\MA$-hard. They show that to encode the evaluation of a classical circuit into a quantum instance, it is useful to first represent the classical circuit as a quantum circuit. To this end, they define $\MA_q$---the quantum version of this class composed purely of reversible classical gates and ancilla qudits $\ket{0}$ and $\ket{+}$---and show that $\MA_q = \MA$. Similarly, we define the complexity class $\coRP_q$ and show that it is equal to $\coRP$. However, instead of using $\ket{+}$ states, we use maximally entangled states $\ket{\Phi^+}$ in preparation for the reduction to a \textsc{Classical} SLCT-QSAT instance.

\begin{defn}[$\coRP_q$]
    A promise problem $A = (A_{yes}, A_{no})$ is in $\coRP_q$ iff there exist polynomials $p,q$ and a uniform family of polynomial-time classical reversible circuits $\{R_n\}$, that take as input a binary string $x$ with $\norm{x} = n$, $q(n)$ ancilla qubits in the state $\ket{0}$, and $p(n)$ Bell states $\ket{\Phi^+}$, such that each $R_n$ acts only on one of the halves of each Bell state, and:
    \begin{itemize}
        \item (Completeness) If $x \in A_{yes}$, then $\Pr \left[ R_n \textnormal{ accepts } x\right] = 1$.
        \item (Soundness) If $x \in A_{no}$, then $\Pr \left[ R_n \textnormal{ accepts } x\right] \leq 1/3$.
    \end{itemize}
\end{defn}

\begin{lemma} \label{lemma:corpq}
    $\coRP_q = \coRP$.
\end{lemma}

\begin{proof}
    $\coRP \subseteq \coRP_q$: Let $A$ be an arbitrary (promise) problem in $\coRP$. In other words, there exists a probabilistic classical algorithm $Q_x = Q_L \ldots Q_0$ with $L = poly(n)$ using $p(n)$ random bits that decides an instance $x$ of $A$ with perfect completeness. As mentioned earlier, any $\coRP$ circuit can be perfectly and efficiently simulated by any other circuit using a universal gate set. Hence, there exists a reversible circuit $R_x = R_{L'} \ldots R_0$ with $L' = poly(L) = poly(n)$, acting on $p(n)$ random bits and an additional $poly(n)$ ancilla bits that is equivalent to $Q_x$. Moreover, we can take $R_x$ to be composed of Toffoli and X gates since this set is both universal and reversible.\footnote{The X gates are necessary because of our choice to have all ancillas initialized to bit $0$. Otherwise, the Toffoli gates would only act trivially.} We let $q(n) = n + poly(n)$ denote the total number of ancilla bits in the circuit. The $\coRP_q$ quantum circuit that decides $A$ is simply $R_x$, except each input random bit is replaced with a two-qubit Bell state $\ket{\Phi^+}$, and each ancilla bit is replaced by an ancilla qubit. These circuits are illustrated in \cref{fig:corpcircs}.

    To demonstrate the $\coRP_q$ circuit meets the completeness and soundness conditions, we again observe that it has the same structure as in previous subsections and the evolution of the input state can be broken down into the evolution of independent bitstrings. The probability the circuit accepts is given by

    \begin{equation} \label{eqn:paccept2}
        \begin{aligned}
            p_{acc} & = 2^{-p} \sum_{b,b' \in \{0,1\}^p} \bra{b'} \otimes \bra{\xi_b^L} \Pi^{(1)}  \ket{\xi_b^L} \otimes \ket{b} \\
                    & = 2^{-p} \sum_{b \in \{0,1\}^p}  \bra{\xi_b^L} \Pi^{(1)}  \ket{\xi_b^L}                                    \\
                    & = 2^{-p} \sum_{b \in \{0,1\}^p}  p_{acc,b},
        \end{aligned}
    \end{equation}
    where $\ket{\xi_b^L}$ is the state at the end of the computation for bitstring $b$ and $p_{acc,b} := \bra{\xi_b^L} \Pi^{(1)}  \ket{\xi_b^L} = \bra{b} \otimes \bra{0}^{\otimes q} R^\dagger_1 \ldots R^\dagger_L \Pi^{(1)}  R_L \ldots R_1 \ket{0}^{\otimes q} \otimes \ket{b}$ is the probability the circuit accepts $(x,b)$. From this equation it becomes apparent that $p_{acc}$ is the fraction of bitstrings $b$ that accept. Then, observe that if $x$ is a yes-instance, the classical circuit accepts the instance with certainty for any random bitstring $b \in \{0,1\}^p$, i.e.\ $p_{acc,b} = 1$ for all $b$. Hence, $p_{acc} = 1$ and the quantum circuit also accepts with certainty. On the other hand, if $x$ is a no-instance, we have that in the original circuit at most a third of the bitstrings lead to wrongful acceptance, implying that $p_{acc} \leq 2^{-p} (2^p/3) = 1/3$.

    $\coRP_q \subseteq \coRP$: Suppose $A$ is a problem in $\coRP_q$, and let the rightmost circuit in \cref{fig:corpcircs} be the $\coRP_q$ circuit decides an instance $x$ of $A$. The classical circuit that also decides this instance is the middle circuit of \cref{fig:corpcircs}, where $r$ is a bitstring sampled uniformly at random. To prove that it satisfies the completeness and soundness conditions, consider the acceptance probability of the original $\coRP_q$ circuit given by \cref{eqn:paccept2}. If $x$ is a yes-instance, this time we have that $p_{acc} = 1$. The r.h.s.\ of the equation can only be $1$ if $p_{acc,b} = 1$ for all $b$. Since this is the probability that the classical algorithm accepts with input bitstring $b$, the classical algorithm also accepts with certainty. On the other hand, if $x$ is a no-instance, $p_{acc} = 2^{-p} \sum_b p_{acc,b} \leq 1/3$. Since this is the fraction of bitstrings that accept, the classical algorithm---which samples one of these random bitstrings---also accepts with probability $\leq 1/3$.
\end{proof}

Now, we can more comfortably show how to encode a $\coRP$ circuit into an instance of our problem. Let $R = R_L \ldots R_1$ be the $\coRP_q$ circuit that given an instance $x$ of problem $A$, decides whether $x \in A_{yes}$ or $x \in A_{no}$. The input to the circuit are two ancilla registers $\ket{0}^{\otimes q}$ and $\ket{\Phi^+}^{\otimes p}$ and let $ans$ be the ancilla qubit that provides the decision. The reduction $x \mapsto x'$ then requires $q + p + (L + 1)$ qudits, each of dimension $6$, and is the following:

\begin{enumerate}
    \item Choose $q + p$ qudits to serve as the logical qudits of the computation, and $L+1$ to function as the clock qudits. Define an ordering for each group of qudits.
    \item For all $i \in [q]$, create a $\piinitzero$ clause acting on $c_0$ and $l_i$.
    \item For all $j \in [p]$, create a $\piinitpair$ clause acting on $c_0$ and $l_{q-1+j}$.
    \item For all $k \in [L]$, create a $\Pi_{prop}$ clause with associated unitary $R_{k+1}$ where the two clock qudits of the clause are $c_k$ and $c_{k+1}$ (with $c_k$ being the predecessor of $c_{k+1}$) and the two logical qudits are those that the gate originally acts on.
    \item Create one $\Pi_{out}$ clause acting on $c_L$ and $l_{ans}$.
\end{enumerate}
\noindent
The resulting instance $x'$ consists of a one-dimensional clock chain of length $L$ ($L+1$ clock qudits) with a unique direction, no simultaneous propagation clauses, and all logical qudits initialized to either $\ket{0}$ or $\ket{\Phi^+}$.

\subsubsection{Completeness and soundness}

In the case that $x \in A_{yes}$, $x'$ is satisfied by the history state

\begin{equation*}
    \ket{\psi_{hist}} = \frac{1}{\sqrt{L+1}} \sum_{t = 0}^L \left[ R_{t} \ldots R_{0} \ket{0}^{\otimes q} \ket{\Phi^+}^{\otimes p} \right] \otimes \ket{\underbrace{d \ldots d}_{t} a_t \underbrace{r \ldots r}_{L-t}},
\end{equation*}
\noindent
and so completeness is preserved.

When $x \in A_{no}$, the only nontrivial change in \cref{eqn:hamilsoundnessfree} is that $H_{init}$ now collects both the $(I - \ket{0} \! \bra{0}) \otimes \ket{a} \! \bra{a}$ terms ($q$ of them) and the $(I - \ket{\Phi^+} \! \bra{\Phi^+}) \otimes \ket{a} \! \bra{a}$ terms ($p$ of them). Again, we must only show that $H|_{\mathcal{S}_{init}}$ has no small eigenvalues in both scenarios where one of the logical qudits in the TACC is in the undefined state, and when there's no such logical qudit. Given that we have shown several similar arguments before, we only provide the key points. \\

\medbreak
\noindent
\textbf{At least one undefined logical qudit} \\
\noindent
Here, $H|_{\mathcal{S}_{init}} = H_{init} + H_{prop}(T,u)$. While $H_{init}$ is no longer diagonal, the $(I - \ket{0} \! \bra{0}) \otimes \ket{a} \! \bra{a}$ terms commute with the $(I - \ket{\Phi^+} \! \bra{\Phi^+}) \otimes \ket{a} \! \bra{a}$ terms. Analyzing each set of terms separately, we can observe that $\gamma(H_{init}) \geq 1$. The rest follows in the same way as in the corresponding section of the $\QCMA$ problem. \\

\medbreak
\noindent
\textbf{No undefined logical qudit} \\
\noindent
As before, $H|_{\mathcal{S}_{clock}} = H_1 + H_2$ where $H_1 = H_{init} + H_{prop}$ and $H_2 = H_{out}$. In this construction, $\gamma(H_{1}) \geq 1$, $\gamma(H_2) \geq c/T^2$, and their null spaces are given by

\begin{equation*}
    \begin{aligned}
        \mathcal{S}_1 = & \ket{0}^{\otimes q} \otimes \left( 2^{-p/2} \sum_{b \in \{0,1\}^p} \ket{b} \otimes \ket{b} \right) \otimes \ket{C_0} \oplus  \\
                        & (\mathbb{C}^2)^{\otimes q} \otimes (\mathbb{C}^2)^{\otimes 2p} \otimes \Span\{ \ket{C_1}, \ldots, \ket{C_{L-1}} \} \: \oplus \\
                        & U^\dagger [\ket{1}_{ans} \otimes (\mathbb{C}^2)^{\otimes q+p-1}] \otimes (\mathbb{C}^2)^{\otimes p} \otimes \ket{C_L},
    \end{aligned}
\end{equation*}
\noindent
and
\begin{equation*}
    \mathcal{S}_2 = (\mathbb{C}^2)^{\otimes q} \otimes (\mathbb{C}^2)^{\otimes 2p} \otimes \frac{1}{\sqrt{L+1}} \sum_{t=0}^L \ket{C_t}.
\end{equation*}
\noindent
We can again show that $\max_{\ket{\eta} \in \mathcal{S}_2} \bra{\eta} \Pi_{\mathcal{S}_{1,2}} \ket{\eta} = L-1/L+1$, but now
\begin{equation*}
    \max_{ \ket{\beta} \in \mathcal{\tilde{S}}_{1,1}} \bra{\beta} \Pi_{\mathcal{\tilde{S}}_{1,3}} \ket{\beta} = \frac{1}{2^p} \bra{0}^{\otimes q} \otimes \bra{b} [U^\dagger (\ket{1}\!\bra{1}_{ans} \otimes I)U] \ket{0}^{\otimes q} \otimes \ket{b}.
\end{equation*}
\noindent
This latter quantity represents the fraction of bitstrings that lead to acceptance, and so by the soundness property of the $\coRP$ problem $A$, $\max_{\ket{\beta} \in \mathcal{\tilde{S}}_{1,1}} \bra{\beta} \Pi_{\mathcal{\tilde{S}}_{1,3}} \ket{\beta} \leq \epsilon$ with $\epsilon = 1/3$. Plugging this in the Geometric Lemma, we again obtain that $H|_{\mathcal{S}_{clock}} \succeq c(1 - \sqrt{\epsilon})/(2L)^{-3}$, concluding that $H$ has no small eigenvalues.

\section{Universality of Qubits for QCSPs} \label{section:universality}

In the previous sections, we showed that there are three QSAT problems acting on qudits that are complete $\BQP_1^{\mathcal{G}_8}$, $\QCMA$, and $\coRP$. Here, we refine these statements and show that there are QSAT problems on qubits (albeit with higher locality) that are also complete for these classes. To achieve this, we show that any QCSP on qudits can be reduced to another QSAT problem on qubits using little computational power. We note that this section, apart from some changes in the exposition, stems directly from Meiburg's previous iteration of the paper.

At first glance, this statement may seem trivial as operations on qubits are universal for quantum computation, i.e.\ we can emulate a $d$-qudit with a $\lceil \log_2(d)\rceil$ qubits and carry out unitaries on those qubits. However, there is a nuance around reductions on QCSPs that has not come up before, as previous QCSPs were defined on qubits. To see this, consider one of our QSAT problems, say SLCT-QSAT. It is evident that any satisfiable/unsatisfiable instance of SLCT-QSAT when expressed as qubits keeps the same satisfiability status. However, it is not clear if all input instance generated with these qubit clauses are contained within this class. For a successful reduction, we must have both.

To see this, for a simple example, let us begin by representing the basis clock states using qubits: $\ket{r} = \ket{00}$, $\ket{a} = \ket{01}$, and $\ket{d} = \ket{10}$. The $\Pi_{start} = \ket{r} \! \bra{r}$ clause within $\mathcal{O}_{init}$ can now be written as $\Pi_{start} = \ket{00} \! \bra{00} + \ket{11} \! \bra{11}$, where the last term is to prevent the fourth basis state, which did not exist before. The clause $(\ket{00} \! \bra{00} + \ket{11} \! \bra{1 1})_{1,2} + (\ket{00} \! \bra{00} + \ket{11} \! \bra{1 1})_{2,3}$ acting on three qubits is now valid. The clauses can be satisfied by the state $\ket{0_1 0_2 1_3}$, but now there is some ambiguity on the state that they represent: either we have $\ket{r}$ on qubits 1 and 2, or $\ket{a}$ on qubits 2 and 3. In general, doing this for all of our clauses and considering all input instances that may occur adds a significant level of complexity to the problem, making it difficult to determine if it remains in $\BQP_1$. Moreover, we remark that this is not only particular to our QSAT problems, but in fact applies to all QCSPs and classical CSPs defined on qudits or non-Boolean variables! In general, the issue is that we cannot ensure that the new qubit clauses are applied to qubits in a consistent fashion based on its parent qudit problem. For example, a qubit clause might treat a particular qubit as ``qubit 1'' of a previous $d$-qudit, while another clause might refer to the same qubit as ``qubit 2''. Moreover, the qubit clauses could also ``mix and match'', combining ``qubit 1'' from one previous $d$-qudit with ``qubit 2'' from another $d$-qudit (as in the example with our problem). Overall, these lead to constraints that were unrealizable in the parent qudit problem.

Here, we show that with a different reduction, we can guarantee that an satisfiable/unsatisfiable instance on qubits maps to one on qudits with the same satisfiability status. This is something that is not known to be possible classically! In the quantum world, we can fix the issues mentioned above by again using monogamy of entanglement to bind together our constituent qubits into ordered, entangled larger systems. Each clause in the resulting problem is given a projector that forces this particular ordering of qubits, and any two clauses that try to use the same qubits in multiple ways are frustrated. Formally,

\begin{thm} [\cref{thm:dtob}; formal]
  \label{thm:q2formal}
  For any QCSP $\mathcal{C}$ on $d$-qudits, there is another QCSP $\mathcal{C}'$ on qubits, and $\AC^0$ circuits $f$ and $g$, such that $f$ reduces $\mathcal{C}$ to $\mathcal{C'}$, and $g$ reduces $\mathcal{C}'$ to $\mathcal{C}$. If $\mathcal{C}$ is $k$-local, then $\mathcal{C'}$ can be chosen to be $4\cdot 2^{\lceil \log_2\left(\lceil \log_2(d)\rceil \right) \rceil}k $ local (that is, $O(\log(d))$ times larger.)
\end{thm}

\begin{proof}
  First, we will show that any problem in $d$-qudits can be reduced to one on 4-qudits. After that, we will reduce to qubits, and finalize by showing that the reduction is in $\AC^0$.

  We will view a 4-qudit as the product of a ``data'' qubit and an ``entanglement'' qubit. A $d$-qudit will be replaced by $n=\,\lceil \log_2(d)\rceil$ many 4-qudits, and the state of the $d$-qudit will be encoded in the product of all the data qubits. If $d < 2^{\lceil \log_2(d)\rceil} \,= 2^n$, that is, if $d$ is not exactly a power of 2, we will have a Hamiltonian term $T_1$ in our clauses to ensure that the last $2^n - d$ states are not used. Acting on entanglement subspaces of the 4-qudits, consider a term $T_2$ whose null space consists of just the vector

  \begin{equation*}
    \left[1\otimes R_X(\theta) \otimes R_X(2\theta) \otimes \ldots \otimes R_X((n-1)\theta)\right]\frac{\ket{0}^n + \ket{1}^n}{\sqrt{2}},
  \end{equation*}
  \noindent
  where $R_X(\phi)$ is the single-qubit rotation around the X-axis by an angle $\phi$, and here we let $\theta = \frac{\pi}{2n}$. This is a kind of GHZ state, which uses a slightly different pair of basis states (instead of just $\ket{0}$ and $\ket{1}$) on each separate qubit. Any bipartition of this state is impure, but since $T_2$ has a one-dimensional nullspace, it cannot be satisfied by any mixed state. Thus the sum of two $T_2$ on any two overlapping but distinct sets of 4-qudits will be frustrated. If two copies of $T_2$ act on the same 4-qudits in a different order, they will apply the wrong angles $R_X(k\theta)$ at those sites, and the ground states will not align---also leading to frustration.

  Each clause $H$ of $\mathcal{C}$ is mapped to a new clause $H'$ that acts as $H$ on the data subspaces of each set of 4-qudits; that has $T_1 = 1 - \sum_i^d \ket{i} \! \bra{i}$ on each clumping of 4-qudits, to ensure that only the first $d$ states are used; and $T_2$ on each clumping of 4-qudits, to ensure that they will only be clustered to each other and in a particular order.

  Then, we want to reduce this from 4-qudits to qubits. Consider the following Hamiltonian on 4 qubits:

  \begin{equation*}
    H_{4\to 2} = 1 - \ket{\psi_1}\bra{\psi_1} - \ket{\psi_2}\bra{\psi_2}  - \ket{\psi_3}\bra{\psi_3}  - \ket{\psi_4}\bra{\psi_4},
  \end{equation*}
  \noindent
  where

  \begin{equation*}
    \begin{aligned}
      \ket{\psi_1} & = \frac{1}{2}\left(\frac{3}{5}\ket{0000} - \frac{4}{5}\ket{0001} + \ket{0100} + \ket{1010} + \frac{8}{17}\ket{1100} + \frac{15}{17}\ket{1111}\right)      \\
      \ket{\psi_2} & = \frac{1}{2}\left(\frac{4}{5}\ket{0000} + \frac{3}{5}\ket{0001} - \ket{0110} + \ket{1001} + \frac{20}{29}\ket{1101} + \frac{21}{29}\ket{1110}\right)     \\
      \ket{\psi_3} & = \frac{1}{2}\left(\frac{5}{13}\ket{0010} + \frac{12}{13}\ket{0011} - \ket{0111} + \ket{1000} - \frac{21}{29}\ket{1101} + \frac{20}{29}\ket{1110}\right)  \\
      \ket{\psi_4} & = \frac{1}{2}\left(\frac{-12}{13}\ket{0010} + \frac{5}{13}\ket{0011} - \ket{0101} + \ket{1011} - \frac{15}{17}\ket{1100} + \frac{8}{17}\ket{1111}\right).
    \end{aligned}
  \end{equation*}
  \noindent
  Each $\ket{\psi_i}$ is orthonormal, so $H_{4\to 2}$ has a null space of dimension four. By inspecting the 840 distinct ways to apply two copies of $H_{4\to 2}$ to seven qubits, it can be checked that each sum will have a ground state above zero, except for the case where they are applied in the same way. This is a kind of ``uniqueness'' property that could be very loosely interpreted as monogamy for whole subspaces, instead of just one state.

  By counting dimensions, one can check that this property is generic: it would hold almost always for any four random vectors. Unfortunately, for any simple and clean expressions one would write down, it would lack this property by one symmetry or another. This is why the simplest construction readily available, given above, is actually quite ugly.

  Given a problem on 4-qudits, we can replace each 4-qudit with a collection of four qubits. Each clause is modified to act on the $\ket{\psi_i}$ basis of qubits instead of $\ket{i}$ basis of the 4-qudits. Then, for each 4-qudit that the clause acted on, we add a copy of $H_{4\to 2}$. The above uniqueness property ensures that no other clause can act on the same qubits in any other order, or mix them with any other set of qubits.

  In the $\mathcal{C}'$ QCSP, any problem where qubits are mixed or applied in inconsistent orders, can immediately be rejected. Some qubits may not be acted on by any clause, and so not correspond to a $d$-qudit in $\mathcal{C}$, but then those qubits can simply be ignored. This leaves us with only correctly grouped qubits, in a certain subspace, that thus function equivalently to the $d$-qudits.

  Combined, this gives a faithful reduction from $d$-qudits to qubits, and back again. It turns a $k$-local Hamiltonian into a $4\lceil \log_2(d)\rceil k$ local Hamiltonian. It remains to check the complexity of the reductions.

  In the above description, the expansion factor is $4\lceil \log_2(d)\rceil$. To get a low circuit complexity, we want the expansion to be a power of two. Thus we round this up to $4\cdot 2^{\lceil \log_2\left(\lceil \log_2(d)\rceil \right) \rceil}$, which will denote $x$; the QCSP $\mathcal{C}'$ is augmented through just adding more qubits to increase the subspace dimension, and then $T_1$ is modified again to prevent those states from being occupied.

  An instance of a $k$-local QCSP $\mathcal{C}$ can be written down as a list of integers, each given by an integer clause type, and $k$ many integers representing the qudits they act on. The clause types of $\mathcal{C'}$ are in one-to-one correspondence with the clause types of $\mathcal{C}$, so those integers remain unchanged. Each clause acting on qudit $i$ now instead acts on qudits $[xi,xi+1,\dots xi+(x-1)]$ in that order. Thus a reduction circuit $f$ only needs to be able to replicate an integer several times, multiply by a constant power of 2, $x$, and add a number $i \in [0,x)$. This is in $AC^0$ (in fact it requires no gates at all).

  For a circuit $g$ to convert back from $\mathcal{C'}$ to $\mathcal{C}$, we need to map qubit numbers back to qudits numbers, and check that no qubits are used in inconsistent fashion. For each collection of qubits $[a_1,a_2,\dots a_x]$ that a clause in $\mathcal{C}'$ is acting on, we can map that to the $d$-qudit number $a_1$. Thus, many qudit numbers will go unused, but this does not affect the correctness: as long as all collections use the same numbers in the same order, they will all be mapped to $a_1$. To check that all qudits are used in a consistent fashion, we need to check for each pair of collections $([a_1,\dots a_x], [b_1,\dots b_x])$ that they do not use qubits in inconsistently. Logically, this reads:

  \begin{equation*}
    \big(((a_i = b_i) \wedge (a_j = b_j)) \vee ((a_i \neq b_i) \wedge (a_j \neq b_j))\big) \wedge (a_i \neq b_j) \wedge (a_j \neq b_i)
  \end{equation*}
  \noindent
  and this must be checked for every collection, for every $i$ and $j$, and then combined by an unbounded fan-in AND. The integer equalities $a_i = b_i$ and $a_i \neq b_i$ can be evaluated with unbounded fan-in AND and OR respectively. This check is all in $\AC^0$. If the check fails, the circuit outputs some fixed clause(s) that are unsatisfiable, otherwise it outputs a repetition of the first clause. This $\AC^0$ circuit checks that the qudits are used consistently, and if they are, gives an equivalent instance in the original language.
\end{proof}
\indent
If we do not care about having the reductions in $\AC^0$, and instead allow $\P$-reductions, then $4\lceil \log_2(d)\rceil k$ locality suffices. This reduction is optimal within a factor of 4, in the sense that encoding one $d$-qudit in several qubits requires at least $\lceil \log_2(d)\rceil$-many qubits.

Now, we can conclude:
\begin{coro}[\cref{coro:dtob}]
  \cref{thm:q2formal}, combined with the {\rm QSAT} problems from \cref{section:our,section:qcma,section:corp} imply:
  \begin{enumerate}
    \item There is a $\BQP_1^{\mathcal{G}_8}$-complete {\rm QSAT} problem on qubits with $48$-local interactions.
    \item There is a $\QCMA$-complete {\rm QSAT} problem on qubits with $48$-local interactions.
    \item There is a $\coRP$-complete {\rm QSAT} problem on qubits with $60$-local interactions.
  \end{enumerate}
\end{coro}

\section{Direct Sum and Direct Product Classes} \label{section:sumproduct}

There is a notion of {\em direct sum} (denoted by ``$\oplus$'') and {\em direct product} (denoted by ``$\otimes$'') on CSP languages. To be clear, this is an operation on a CSP, not on instances. For example, we can talk about the languages $\textsc{3-Colorable}\oplus\textsc{4-SAT}$ and $\textsc{3-Colorable}\otimes\textsc{4-SAT}$. This notion extends in a natural way to the quantum setting. This notion appears quite natural\iffalse\footnote{We justify the claim of naturalness by two criteria. First, some colleagues of the authors were able to immediately guess what we meant by "direct product of constraint problems" without other explanation. Second, asking large language models to define such a term, they provided the same definition that we had. However, both the human and machine denied having seen such a concept before.}\fi, but we were unable to find any sources actually discussing this precise notion---possibly because the classical theory is not as exciting, for reasons we will also discuss. Since sum and product QCSPs inherit completeness properties from their constituents, we will be able to construct QCSPs complete for additional classes in this way. We start with the classical case, since the quantum case follows almost identically.

\subsection{Direct product of constraint problems}

Let $L_1$ and $L_2$ be two CSPs, with domains $D_1$ and $D_2$, and allowed constraints $C_1$ and $C_2$.

\begin{defn}[Direct Product of Constraint Problems] \label{defn:directprod}
    Given the CSPs $L_1$ and $L_2$, their {\em direct product} $L_1 \otimes L_2$ is a CSP whose domain is the Cartesian product $D_1 \times D_2$. Each constraint $c_i \in C_1$ (resp.\ $C_2$) of locality $k$ leads to a constraint $c_i'$ in $L_1 \otimes L_2$, also of locality $k$, as follows. A tuple $(v_1,v_2,\dots, v_k) \in (D_1 \times D_2)^{k}$, where each entry $v_i = (v_{i,1}, v_{i,2})$, belongs to $c_i'$ if the tuple $(v_{i,1},\dots, v_{k,1})$ belongs to $c_i$. Each constraint in $L_1 \otimes L_2$ arises this way from a constraint in $L_1$ or $L_2$.
\end{defn}
As a direct product, we have {\em projection} functions $\pi_1$ and $\pi_2$ that map a constraint problem $\mathcal{P}$ in $L_1 \otimes L_2$ to problems $\mathcal{P}_1$ of $L_1$ and $\mathcal{P}_2$ of $L_2$. There are also projections that take variable assignments on $\mathcal{P}$ to variable assignments on $\mathcal{P}_1$ or $\mathcal{P}_2$, by projecting each variable component-wise to $D_1$ or $D_2$. The problem $\mathcal{P}$ is then satisfiable if and only if the projections $\mathcal{P}_1$ and $\mathcal{P}_2$ are both satisfiable. This immediately leads to the following properties.

\begin{thm}[Containment of Direct Products]\label{thm:classical-prod-contain}
    Let $\mathcal{C}$ be a complexity class that is closed under intersection, and powerful enough to compute projections of a direct product CSP. (For instance, any class closed under reductions by local functions.\footnote{Here by {\em local function} we mean an operation which at each symbol of input outputs zero or one symbols of output, using only the last $k$ symbols of input---for some fixed $k$. There is a similar notion of local functions for circuits, where each bit of output only depends on finitely many bits of input, but we include geometric locality. This is an extremely weak notion of reduction, much weaker than $AC^0$-reducibility. The fact that this reduction is weak means that our theorems apply to essentially all complexity classes of interest.}) If the satisfiability problems for CSPs $L_1$ and $L_2$ are both contained in $\mathcal{C}$, then the direct product $L_1 \otimes L_2$ is also contained in $\mathcal{C}$.
\end{thm}

\begin{proof}
    Given an instance $\mathcal{P}$ of $L_1 \otimes L_2$, the projection $\mathcal{P}_1$ can be computed, and then solved, in $\mathcal{C}$. This process forms one language $\mathcal{L}_1$ in $\mathcal{C}$. Likewise, there is a language $\mathcal{L}_2$ in $\mathcal{C}$ for the $L_2$ projection. Then the intersection of $\mathcal{L}_1$ and $\mathcal{L}_2$ is also in $\mathcal{C}$.
\end{proof}
Being closed under intersection is a very mild requirement, essentially equivalent to being able to compute a logical AND. It is satisfied by all classes discussed in this paper, with simple constructions. It is well-known that all classes discussed up to this point satisfy this requirement. An example of a nontrivial case is the closure of PP under intersection \cite{Beigel1991}. Similarly, closure under local functions is also a weak requirement, essentially expressing irrelevance up to local changes in problem encoding.
\par The direct product construction also works the other way, giving hardness:

\begin{thm}[Hardness of Direct Products]\label{thm:classical-prod-hard}
    Let $\mathcal{C}$ be a complexity class, so that the satisfiability problem for the CSP $L_1$ is hard for $\mathcal{C}$ under $M$-reductions, where $M$ is a class of functions closed under composition with computing projections of a CSP. (It suffices that $M$ is closed under composition with local functions.) Then, for any CSP $L_2$, the direct product $L_1 \otimes L_2$ is also hard for $\mathcal{C}$.
\end{thm}

\begin{proof}
    The hardness of $L_1$ for $\mathcal{C}$ under $M$-reductions means that for any language $\mathcal{L} \in \mathcal{C}$ there is a function $f \in M$ so that $f(s)$ describes a satisfiable instance of $L_1$ iff $s \in \mathcal{L}$. We can map from the instance of $L_1$ to an instance of $L_1 \otimes L_2$ with the canonical injection $\iota_1$. This produces an instance of $L_1 \otimes L_2$ which is satisfiable iff $s \in \mathcal{L}$ which is a function also in $M$, and so the composition $\iota_1 \circ f$ is also in $M$, and so $L_1 \otimes L_2$ is also hard for $\mathcal{C}$.
\end{proof}
Again, the requirement on the reduction is very mild; for instance, completeness under $P$- and $L$-reductions, the two most common notions of completeness, both satisfy this property. These two theorems together lead to a simple result about completeness of direct products.

\begin{defn}[Pairwise Intersection of Classes] \label{defn:PI}
    If $\mathcal{C}_1$ and $\mathcal{C}_2$ are two complexity classes (any sets of languages), then $\mathsf{PI}(\mathcal{C}_1,\mathcal{C}_2)$ is the class that denotes the {\em pairwise intersection} of $\mathcal{C}_1$ and $\mathcal{C}_2$. In other words, it is the class of languages that can be written as the intersection, i.e.\ the logical AND, of a language in $\mathcal{C}_1$ and a language in $\mathcal{C}_2$.
\end{defn}

\begin{thm}[Completeness of Direct Products]\label{thm:classical-prod-complete}
    Let $M$ be a set of functions closed under composition with local functions, and closed under concatenations (i.e.\ if some $f,\, g : \Sigma_1^* \to \Sigma_2^*$ are each in $M$, then $h : x \to f(x)g(x)$ is as well). Let $L_1$ be a CSP complete under $M$-reductions for a class $\mathcal{C}_1$, and likewise $L_2$ be complete for $\mathcal{C}_2$. Assume that each of $\mathcal{C}_1$ and $\mathcal{C}_2$ are closed under reductions by local functions, closed under intersections, and contain the language $\textsc{All}$ of all strings, $\Sigma^*$. Then, the direct product $L_1 \otimes L_2$ is complete under $M$-reductions for $\mathsf{PI}(\mathcal{C}_1,\mathcal{C}_2)$.
\end{thm}

\begin{proof}
    First we show containment, then hardness.

    \textbf{Containment:} Since $\textsc{All}$ is in $\mathcal{C}_2$, every language $\mathcal{L}_1$ in $\mathcal{C}_1$ is also in $\mathsf{PI}(\mathcal{C}_1,\mathcal{C}_2)$, as $\mathcal{L}_1 = \mathcal{L}_1 \cap \textsc{All}$. Likewise, every language in $\mathcal{L}_2$ is in $\mathsf{PI}(\mathcal{C}_1,\mathcal{C}_2)$.

    \par Since $\mathcal{C}_1$ and $\mathcal{C}_2$ are each individually closed under intersection, the class $\mathsf{PI}(\mathcal{C}_1,\mathcal{C}_2)$ will be closed under intersection as well. Since $\mathcal{C}_1$ and $\mathcal{C}_2$ are each closed under local functions, $\mathsf{PI}(\mathcal{C}_1,\mathcal{C}_2)$ is closed under local functions as well. As $\mathsf{PI}(\mathcal{C}_1,\mathcal{C}_2)$ satisfies the requirements of \cref{thm:classical-prod-contain}, and the problems $L_1$ and $L_2$ are both contained in $\mathsf{PI}(\mathcal{C}_1,\mathcal{C}_2)$, by this theorem, $L_1 \otimes L_2$ is contained in $\mathsf{PI}(\mathcal{C}_1,\mathcal{C}_2)$.

    \textbf{Hardness:} Any language $\mathcal{L}$ in $\mathsf{PI}(\mathcal{C}_1,\mathcal{C}_2)$ can be written as $\mathcal{L}_1 \cap \mathcal{L}_2$. Since $L_1$ is complete (and thus hard) for $\mathcal{C}_1$, and $L_2$ is complete for $\mathcal{C}_2$, by \cref{thm:classical-prod-hard}, their direct product $L_1 \otimes L_2$ is hard for both $\mathcal{C}_1$ and $\mathcal{C}_2$. Accordingly, there are $M$-computable functions mapping $\mathcal{L}_1$ and $\mathcal{L}_2$ to CSP instances that are satisfiable iff the original strings are in the language. By concatenating these instances (as lists of constraints on numbered variables), we get a new instance that is satisfiable iff both are, i.e.\ only if the string is in both $\mathcal{L}_1$ and $\mathcal{L}_2$. Since $M$ is closed under concatenation, this is an $M$-computable reduction, and $L_1 \otimes L_2$ is hard for $\mathsf{PI}(\mathcal{C}_1,\mathcal{C}_2)$.
\end{proof}

The assumptions on $M$ in \cref{thm:classical-prod-complete} are again very mild, satisfied for $AC^0$-reductions or even weaker. This essentially says that if we have constraint problems complete for ``reasonable" classes $\mathcal{C}_1$ and $\mathcal{C}_2$, then there is a constraint problem complete for $\mathsf{PI}(\mathcal{C}_1,\mathcal{C}_2)$.

\par When is this interesting? If $\mathcal{C}_1 \subseteq \mathcal{C}_2$, then $\mathsf{PI}(\mathcal{C}_1,\mathcal{C}_2) = \mathcal{C}_2$. This is the case with $\P$ and $\NP$, the two most widely studied classes, so there is nothing new. For instance, the direct product of 2-SAT and 3-SAT gives another $\NP$-complete CSP. Schaefer's Dichotomy Theorem states that all Boolean CSPs are either in co-NLOGTIME, $\mathsf{L}$-complete, $\NL$-complete, $\oplus \mathsf{L}$-complete, $\P$-complete or $\NP$-complete under $AC^0$ reductions \cite{Allender2009}. With exception of $\NL$ and $\oplus \mathsf{L}$, all possible pairs from this list have an obvious containment relation, so the only nontrivial consequence would be that there exists a CSP, on a domain of size four, that is $\mathsf{PI}(\oplus \mathsf{L},\NL)$-complete under $AC^0$ reductions. But under coarser reductions such as $P$-reductions, all CSPs collapse to $\P$ or $\NP$, so we would have nothing to say. This will become more interesting in the quantum case!

\subsection{Direct sum of constraint problems}

If direct products let us express (informally) a ``two-input logical AND" of two CSPs, then direct sums let us express ``unbounded-fanin AND of fanin-2 ORs".

\begin{defn}[Direct Sum of Constraint Problems] \label{defn:directsum}
    Given the CSPs $L_1$ and $L_2$, their {\em direct sum} $L_1 \oplus L_2$ is a CSP whose domain is the disjoint union $D_1 \cupdot D_2$. Each constraint in $L_1 \oplus L_2$ is either of the form
    $c_i \cup \big(D_2^k\big)$, where $c_i \in C_1$ is a constraint of locality $k$; or it is $c_i \cup \big(D_1^k\big)$ for some $c_i \in C_2$.
\end{defn}

Informally, a constraint $c_i \in C_1$ becomes a constraint in $L_1 \oplus L_2$, stating ``all variables in the support belong to $D_2$, or they are all belong to $D_1$, and in the latter case they form a satisfying assignment to $c_i$''. Given an instance $\mathcal{P}$ of $L_1 \oplus L_2$ with a single connected component,\footnote{A {\em connected component} of a CSP is a connected component in the graph for that CSP, where the vertices are variables, and there is an edge between variables if they share a constraint.} any satisfying assignment must put all variables in $D_1$ or put all variables in $D_2$. Each option can be inspected and reduces to solving one instance of $L_1$ or $L_2$. If either one is satisfiable, the component is as well. Solving the whole problem then amounts to doing this check for each connected component. This leads to the following property:

\begin{thm}[Containment of Direct Sums]\label{thm:classical-sum-contain}
    Let $\mathcal{C}$ be a complexity class that is closed under union and delimited concatenation. Furthermore assume it is powerful enough to compute projections of a direct sum CSP, and powerful enough to separate a CSP into its connected components. If the satisfiability problems for CSPs $L_1$ and $L_2$ are both contained in $\mathcal{C}$, then the direct sum $L_1 \oplus L_2$ is also contained in $\mathcal{C}$.
\end{thm}

\begin{proof}
    To solve an instance of $L_1 \oplus L_2$, the problem is first separated into its connected components. These become separate problems, but can be treated as one, concatenated and appropriately delimited. (In terms of strings, the problem would be written $\textsc{Component1}|\textsc{Component2}|\textsc{Component3}$, etc.) This is solvable if each component is solvable individually. Each component $C \in L_1 \oplus L_2$ is solved by computing the projections $\mathcal{P}_1 \in L_1$ and $\mathcal{P}_2 \in L_2$. The component is satisfiable if either projection is satisfiable, and since $\mathcal{C}$ is closed under union, it suffices that each of $L_1$ and $L_2$ are in $\mathcal{C}$.
\end{proof}
The first several requirements on $\mathcal{C}$ are again mild. Closure under delimited concatenation here means that if $L$ is a language in $\mathcal{C}$, then $L(dL)^*$ is also in $\mathcal{C}$, where $d$ is a fresh symbol not in the alphabet of $L$ and ${}^*$ is the Kleene star as before. In other words, take any combination of one or more strings from $L$, and concatenate them with a fresh symbol $d$ in between, segmenting them; the language of all such concatenations must also belong to $\mathcal{C}$. Intuitively, we can recognize this language by splitting at delimiter, checking ``each part'' of an input string individually and accepting only if they are all individually accepted. The assertion that $\mathcal{C}$ is closed under this operation means that machines recognizing $\mathcal{C}$ are able to carry out such a split-and-check operation.

\par On the other hand, computing connected components is slightly nontrivial, but as shown in Ref.\ \cite{Reingold2008}, it can be done in $\mathsf{L}$. Otherwise, these closure properties are satisfied by every ``reasonable" class above $\mathsf{L}$.

\par One might expect that, by analogy with the direct product, the direct sum should then be complete for the pairwise union of two classes. This would be the case if we only ever had to worry about problems that form a single connected component. By embedding a constraint problem from each of $L_1$ and $L_2$ on the same set of variables, we get a union-like problem. But by placing each constraint problem on disjoint sets of variables, we get an intersection-like problem, that must satisfy each. This motivates the following definition:

\begin{defn}[Star of Pairwise Unions of Classes] \label{defn:SoPU}
    If $\mathcal{C}_1$ and $\mathcal{C}_2$ are two complexity classes, then their {\em star of pairwise unions}, denoted $\mathsf{SoPU}(\mathcal{C}_1,\mathcal{C}_2)$, is a complexity class defined as follows: for each language $L_1 \in \mathcal{C}_1$ and $L_2 \in \mathcal{C}_2$, let $d$ be a fresh symbol that is not in the alphabet of $L_1$ or $L_2$. Then, the language $(dL_1|dL_2)^*$ is in $\mathsf{SoPU}(\mathcal{C}_1,\mathcal{C}_2)$. $\mathsf{SoPU}(\mathcal{C}_1,\mathcal{C}_2)$ is the closure of all such languages under $\mathsf{L}$ (logspace reductions).
\end{defn}
This definition merits a brief explanation. For a pair of languages $L_1$ and $L_2$, what do the strings in the language $L \equiv (dL_1|dL_2)^*$ look like? Given an input string like $d01001110d101101d101001$, it will belong to $L$ if and only if each of $\{01001110,101101,101001\}$ belongs to either $L_1$ or $L_2$. So, with an oracle for $L_1$ or $L_2$, the task becomes segmenting out these $s$-delimited substrings, querying both oracles to see if each segment passes, and accepting the string only if each segment is accepted. If $C$ is a complexity class powerful enough to recognize $L_1$ and $L_2$ itself, and do the segmenting, then $\mathsf{SoPU}(\mathcal{C}_1,\mathcal{C}_2) \in C$.

We now formalize and proof these facts.

\begin{thm}[Hardness of Direct Sums]\label{thm:classical-sum-hard}
    Let $\mathcal{C}$ be a complexity class and $L_1$ be a CSP, such that $L_1\mathsf{-SAT}$ is hard for $\mathcal{C}$ under $M$-reductions, where $M$ is a class of functions closed under composition with $AC^0$. Let $L_2$ be any CSP with at least one unsatisfiable instance. Then the direct sum $L_1 \oplus L_2$ is also hard for $\mathcal{C}$.
\end{thm}
\begin{proof}
    The hardness of $L_1$ for $\mathcal{C}$ under $M$-reductions means that for any language $\mathcal{L} \in \mathcal{C}$ there is a function $f \in M$ so that $f(s)$ describes a satisfiable instance of $L_1$ iff $s \in \mathcal{L}$. We can map from the $f(s)$ instance of $L_1$ to an instance $s'$ of $L_1 \oplus L_2$ with the canonical injection $\iota_1$. This produces an instance of $L_1 \oplus L_2$ which is satisfiable whenever $s$ is, by mapping a satisfying assignment along $\iota_1$; but it is also trivially satisfiable by assigning variables any value in the domain of $L_2$.

    So, let $U$ be any unsatisfiable instance of $L_2$, and let $v_U$ be the number of variables in that instance. Modify $s'$ by adding dummy variables until the number is a multiple of $v_U$. Arbitrarily partitioning the variables of $s'$ into sets of size $v_U$, add constraints representing $v_{s'}/v_U$ many copies of $U$ to the problem.

    The resulting problem $s'$ is an instance of $L_1 \oplus L_2$. Since each variable belongs to one of the copies of $U$, there cannot be a satisfying instance of $s'$ that assigns any variable a value in the domain of $L_2$. The dummy variables can take any value in the domain of $L_1$, while the original variables produce a satisfying instance of $s$.

    In this way, we've produced an instance of $L_1 \oplus L_2$ which is satisfiable iff $s \in \mathcal{L}$. The initial mapping $f : \mathcal{L} \to L_1$ is in $M$, and the embedding from $L_1$ to $L_1 \oplus L_2$ is in $AC^0$, so the composition is also in $M$. Since this mapping is complete, sound, and in $M$ for any $\mathcal{L} \in \mathcal{C}$, we have proved that $L_1 \oplus L_2$ is also hard for $\mathcal{C}$.
\end{proof}

\begin{thm}[Completeness of Direct Sums]
    Let be $M$ a set of functions closed under composition with logspace-computable functions (such as the set of logspace functions themselves, $\mathsf{FL}$). Let $L_1$ be a CSP complete under $M$-reductions for a class $\mathcal{C}_1$, and likewise $L_2$ be $M$-complete for $\mathcal{C}_2$. Assume that each of $\mathcal{C}_1$ and $\mathcal{C}_2$ are closed under $M$-reductions, and contain the language $\textsc{None}$ of no strings, $\emptyset$. Then, the direct sum $L_1 \oplus L_2$ is complete under $M$-reductions for $\mathsf{SoPU}(\mathcal{C}_1,\mathcal{C}_2)$.
\end{thm}

\begin{proof}
    First we show containment, then hardness.

    \textbf{Containment:} We might seek to apply \cref{thm:classical-sum-contain} to show that, since $L_1$ and $L_2$ are each in $\mathsf{SoPU}(\mathcal{C}_1,\mathcal{C}_2)$, their direct sum is too. This approach is limited by the fact that $\mathsf{SoPU}$ is not generally going to be closed under union, a necessary property for \cref{thm:classical-sum-contain}. But the same algorithm described in the proof of \cref{thm:classical-sum-contain} essentially applies.

    To solve an instance of $L_1 \oplus L_2$ in $\mathsf{SoPU}(\mathcal{C}_1,\mathcal{C}_2)$, first separate the CSP into its connected components, and join them in a delimited fashion, e.g. $s\textsc{Component1}s\textsc{Component2}s\textsc{Component3}$. The problem is solvable if each component is solvable individually. Each component $c_i$ is an instance of $L_1 \oplus L_2$ and has projections $p_{i,1} \in L_1$ and $p_{i,2} \in L_2$. The component is satisfiable if either projection is. So reduce this (in logspace) to the string $sp_{1,1}tp_{1,2}sp_{2,1}tp_{2,2}sp_{3,1}tp_{3,2}\dots$, where $t$ is another new fresh symbol. The resulting language is not the star-of-union of the languages $L_1$ and $L_2$, but of two other languages $L_1' := L_1t\Sigma^*$ and $L_2' := \Sigma^*tL_2$. Since $L_1'$ and $L_2'$ are equivalent to $L_1$ and $L_2$ under logspace reductions, and $L_1$ and $L_2$ belong to the classes $\mathcal{C}_1$ and $\mathcal{C}_2$ which are closed under logspace reductions, $L_1'$ and $L_2'$ belong to $\mathcal{C}_1$ and $\mathcal{C}_2$ as well, so this language belong to $\mathsf{SoPU}(\mathcal{C}_1,\mathcal{C}_2)$.

    \textbf{Hardness:} Any language $\mathcal{L}$ in $\mathsf{SoPU}(\mathcal{C}_1,\mathcal{C}_2)$ is equivalent under an $\mathsf{L}$ reduction to $(s\mathcal{L}_1|s\mathcal{L}_2)^*$. Since $L_1$ is complete for $\mathcal{C}_1$ and $\mathcal{L}_1 \in \mathcal{C}_1$, we can map from $\mathcal{L}_1$ to equivalent instances of $L_1$, and likewise for $L_2$. So given a string for which we want to test membership in $(s\mathcal{L}_1|s\mathcal{L}_2)^*$, we use the following logspace reduction to $L_1 \oplus L_2$:

    Segment the string by $s$. Each segment will produce a disconnected part of the resulting CSP. One string segment $x$ then gets mapped to instances $p_1$ of $L_1$ and $p_2$ of $L_2$. Separate $p_1$ and $p_2$ into their connected components, labelled $p_{1,i}$ and $p_{2,j}$. Then, for each pair of components $(i,j)$, output a set of $\max(|p_{1,i}|,\,|p_{2,j}|)$ many fresh variables, and map the constraints from $p_{1,i}$ and $p_{2,j}$ onto these variables. The resulting CSP will be satisfiable iff the string belonged to $(s\mathcal{L}_1|s\mathcal{L}_2)^* = (s(\mathcal{L}_1|\mathcal{L}_2))^*$: each individual segment $x \overset{?}\in (\mathcal{L}_1|\mathcal{L}_2)$ should be accepted only if at least one of $p_{1}$ or $p_{2}$ are satisfiable. $p_1$ is satisfiable iff all $p_{1,i}$ are satisfiable, and likewise for $p_2$ and $p_{2,j}$. Each output component $(i,j)$ is satisfiable only if $p_{1,i}$ or $p_{2,j}$ are, and taking the intersection of these conditions, we see that this is exactly equivalent to the original problem.
\end{proof}

\subsection{Quantum sums and products}
The above constructions transfer in a very natural way to the quantum setting. Now, instead of domains that are a Cartesian product or disjoint union, the Hilbert spaces are a tensor product or direct sum. The clauses are accordingly built as tensor products and direct sums.

\begin{defn}[Direct Product of Quantum Constraint Problems] \label{defn:prodqcsps}
    Given the QCSPs $L_1$ and $L_2$, their {\em direct product} $L_1 \otimes L_2$ is a QCSP whose domain is the tensor product Hilbert space $D_1 \otimes D_2$. Each operator $H_i$ in $L_1$ leads to an operator $H_i \otimes I$, a tensor product with the identity, and likewise for $L_2$.
\end{defn}

\begin{defn}[Direct Sum of Quantum Constraint Problems]  \label{defn:sumqcsps}
    Given the QCSPs $L_1$ and $L_2$, their {\em direct sum} $L_1 \cupdot L_2$ is a QCSP whose domain is the direct sum Hilbert space $D_1 \oplus D_2$. Each operator $H_i$ in $L_1$ leads to an operator $H_i \oplus 0$, a direct sum with the 0 operator, requiring that a frustration-free state lies in the nullspace of $H_i$ {\em or} the right half of the direct sum (or a linear combination). Likewise for operators in $L_2$.
\end{defn}
These have the same essential properties as the direct product and sum for classical CSPs, that we can produce product and sum instances that are satisfiable iff both (resp.\ either) of the original instances were satisfiable.

One might mistakenly think that, given all this discussion of languages and strings of symbols, we have to talk about quantum states and concatenations of strings of qubits. However, this is not the case: the strings of symbols are just the encoding of the constraints, which are classical data even for a QCSP. The only new property to check for the quantum case is that we can take tensor products of satisfying states and get a satisfying state, or embed a state into the direct sum to get a satisfying state; and the appropriate converse properties. These follow directly from the definition of tensor products and direct sums.

\subsection{Basic properties}
Above, we proved that $\mathsf{PI}(A,B)$ and $\mathsf{SoPU}(A,B)$ have a complete CSP, as long as the original classes $A$ and $B$ also have one. Motivated by this, we give some basic properties of these classes.

\begin{thm}\label{thm:pi-contains}
    If the class $B$ includes the language $\textsc{All}$ of all strings, then $A \subseteq \mathsf{PI}(A,B)$.
\end{thm}
\begin{proof}Since $\textsc{All} \in B$, for any $a \in A$, we have $a = \{a\} \cap \textsc{All} \in \mathsf{PI}(A,B)$.\end{proof}

\begin{thm}\label{thm:sopu-contains}
    If the class $B$ includes the language $\textsc{None}$ of no strings, then $A \subseteq \mathsf{SoPU}(A,B)$.
\end{thm}
\begin{proof}
    Take any $L_1 \in A$. Simplify $(s(L_1|\emptyset))^* = (sL_1)^*$, and then we can do a reduction in logspace to only look at the first $s$-delimited subsequence to get $L_1 \in \mathsf{SoPU}(A,B)$
\end{proof}

\begin{thm}\label{thm:pi-sopu-respect}
    $\mathsf{PI}$ and $\mathsf{SoPU}$ respect the inclusion order of complexity classes. That is, $A \subseteq C$ and $B \subseteq D$ implies $\mathsf{PI}(A,B) \subseteq \mathsf{PI}(C,D)$ and $\mathsf{SoPU}(A,B) \subseteq \mathsf{SoPU}(C,D)$.
\end{thm}
\begin{proof}Immediate from the definitions.\end{proof}

The $\mathsf{SoPU}$ construction generally leads to a more powerful class than the $\mathsf{PI}$ construction.

\begin{thm}\label{thm:pi-subseteq-sopu}
    If classes $A$ and $B$ are closed under reduction by local functions, then $\mathsf{PI}(A,B) \subseteq \mathsf{SoPU}(A,B)$.
\end{thm}

\begin{proof}
    A language in $\mathsf{PI}(A,B)$ is of the form $L_A \cap L_B$. Since $A$ and $B$ are closed under reduction by local functions, the languages $0L_A$ and $1L_B$---the languages where we add one additional 0 or 1 to the start of all strings---are also in $A$ and $B$ respectively. Applying the $\mathsf{SoPU}$ construction to this pair of languages, we get the language $(s0L_A|s1L_B)^*$. Then, we can apply the logspace reduction (while staying in $\mathsf{SoPU}(A,B)$) of mapping any string $x$ (over an alphabet without $s$) to the string $s0xs1x$, so that $x$ is only accepted if it belongs to both $L_A$ and $L_B$. Thus we have a mechanism to recognize $L_A \cap L_B$ in $\mathsf{SoPU}(A,B)$.
\end{proof}

$\mathsf{PI}$ and $\mathsf{SoPU}$ do not increase the power of classes by combining with something weaker. Formally:

\begin{thm}\label{thm:pi-sopu-dni}
    If $A$ is closed under intersection and $B \subseteq A$, then $\mathsf{PI}(A,B) \subseteq A$. If $A$ is closed under logspace reductions, unions, and delimited concatenation, then $\mathsf{SoPU}(A,B) \subseteq A$.
\end{thm}

\begin{proof}The first is immediate from the definition. The second comes from applying each closure property to $(s\mathcal{L}_A|s\mathcal{L}_B)^*$ to see that it is still in $A$.
\end{proof}

Combining \cref{thm:pi-sopu-dni} with \cref{thm:pi-contains,thm:sopu-contains}, these imply that for well-behaved classes where $B \subseteq A$, $\mathsf{PI}(A,B) = \mathsf{SoPU}(A,B) = A$. This explains why these constructions have not been of much interest until the quantum setting. In the classical setting, the only two complexity classes (up to $\mathsf{P}$-reduction) that arise from constraint problems are $\mathsf{P}$ and $\mathsf{NP}$. Both are closed under all the properties discussed above (logspace, unions, intersections, and delimited concatenation), and $\mathsf{P} \subseteq \mathsf{NP}$. In particular:

\begin{coro}\label{thm:pi-sopu-p-np}
    $\mathsf{PI}(\mathsf{P},\mathsf{NP}) = \mathsf{SoPU}(\mathsf{P},\mathsf{NP}) = \mathsf{NP}$.
\end{coro}

\subsection{New complete classes}

The 7 classes discussed so far as having complete constraint problems are $\mathsf{P}$, $\mathsf{coRP}$, $\mathsf{BQP}_1$, $\mathsf{NP}$, $\mathsf{QCMA}_1$, and $\mathsf{QMA}_1$. Recognize that all of these have all of the closure properties discussed in this section so far: union, intersection, logspace reductions, and delimited concatenation; and they all include the trivial problems $\textsc{ALL}$ and $\textsc{NONE}$.

Among these seven classes, most pairs $\{A,B\}$ have $A \subseteq B$, in which case $\mathsf{PI}(A,B) = \mathsf{SoPU}(A,B) = B$. However, there are three pairs that are not known to contain each other: $\mathsf{coRP} \overset{?}{\subseteq} \mathsf{NP}$, $\mathsf{BQP}_1 \overset{?}{\subseteq} \mathsf{NP}$, and $\mathsf{BQP}_1 \overset{?}{\subseteq} \mathsf{MA}$. Each of these pairs leads to two new classes $\mathsf{PI}(A,B)$ and $\mathsf{SoPU}(A,B)$, that are not obviously equal to some other known class, for a total of six more complexity classes with complete quantum constraint problems. The resulting thirteen classes are shown in \cref{fig:inclusions}.

Notably, $\coRP \overset{?}{\subseteq} \NP$ involves only classical classes, and accordingly one could hope that the $\mathsf{PI}(\coRP, \NP)$ and $\mathsf{SoPU}(\coRP, \NP)$ would have more existing theory already developed around them. There are at least two interesting cases where a collapse can occur:
\begin{enumerate}
    \item If $\P=\RP$ (derandomization), then $\coRP = \P \subseteq \NP$ and so
          \begin{equation*}
              \NP \subseteq \mathsf{PI}(\coRP, \NP) \subseteq \mathsf{SoPU}(\coRP, \NP) = \mathsf{SoPU}(\NP, \NP) = \NP.
          \end{equation*}
          \noindent
          An even weaker version of derandomization where $\NP = \MA$ would also lead to a collapse, because $\coRP \subseteq \MA$:
          \begin{equation*}
              \NP \subseteq \mathsf{PI}(\coRP, \NP) \subseteq \mathsf{SoPU}(\coRP, \NP) \subseteq \mathsf{SoPU}(\MA, \NP) = \mathsf{SoPU}(\NP, \NP) = \NP.
          \end{equation*}
    \item Similarly, if $\NP = \coNP$ (concise refutations), then $\coRP \subseteq \coNP = \NP$ and we have a collapse.
\end{enumerate}

The complexity class $\DP = \mathsf{PI}(\NP, \coNP)$ has been studied before, and forms the second layer of the {\em boolean hierarchy} $\BH$: $\DP = \BH_2$ \cite{Cai1988,Papadimitriou1984}. \cref{thm:pi-sopu-respect} tells us that $\mathsf{PI}(\coRP, \NP) \subseteq \DP$ as a result. No inclusions are known between $\DP$ and $\MA$.

$\mathsf{SoPU}(\coRP, \NP)$ does not obviously lie in the Boolean hierarchy. If it was a simple pairwise union, it would lie in $\BH_3$, the pairwise union of $\DP$ and $\NP$. But the long list of checks to do means that we cannot obviously condense down the queries. Queries to an $\NP$ oracle do not need to depend on each other adaptively, though, so we do know that $\mathsf{SoPU}(\coRP, \NP)$ is contained in $\P^{\vert\vert \NP} = \P^{\NP[\log]}$, studied in \cite{Buss1991,Hemachandra1989}; problems that can be solved by a $\P$ machine with polynomially many nonadaptive queries, or logarithmically many adaptive queries. This would generally be seen as a harder class than $\MA$, but since $\MA \overset{?}{\subseteq} \Delta_2\P = \P^\NP$ is not even known, this is not implied by $\mathsf{SoPU}(\coRP, \NP) \subseteq \MA$.

For the quantum classes that involve $\BQP_1$, $\NP$, or $\MA$, it seems difficult to say anything at all about---besides the fact that they lie above $\BQP_1$ and below $\QCMA$.

\section*{Acknowledgements}

We thank Alex Grilo and Dorian Rudolph for helpful discussions, and the authors of Ref.\ \cite{cade2024glh} for the inspiration of \cref{fig:projection}. RRC was supported by projects APVV-22-0570 (DeQHOST) and VEGA No.\ 2/0164/25 (QUAS). AM was supported by the Natural Sciences and Engineering Research Council of Canada through grant number RGPIN-2019-04198. IQC and the Perimeter Institute are supported in part by the Government of Canada through ISED and the Province of Ontario.

\newpage

\printbibliography

\begin{appendices}
  \section{Proofs of Lemmas: Section~\ref{section:monogamy}} \label[appendix]{appendix:monogamy}

\singletypequditmono*

\begin{proof}
    Let $\ket{\psi}$ and $\ket{\phi}$ be states in the null space of $\Pi_\alpha$ and $\Pi_\alpha'$ respectively. Since $\Pi_\alpha$ and $\Pi_{\alpha'}$ act on orthogonal subspaces, it follows that $\ket{\psi}$ and $\ket{\phi}$ are orthogonal and $\ket{\psi} \neq \ket{\phi}$. Therefore, there is no state that can simultaneously satisfy both $\Pi_\alpha$ and $\Pi_{\alpha'}$, making the instance unsatisfiable.
\end{proof}

\subsection{Clock component}

\linearchainmono*

\begin{proof}
    Recall that any clause from $\{\Pi_{init}, \Pi_{prop}, \Pi_{out}\}$ acting on a clock qudit requires one of the clock qudit's auxiliary subspaces to form a Bell pair. A clock qudit is only equipped with two such subspaces ($C\!A$ and $C\!B$) and so a clock qudit that is connected to more than two clock or endpoint qudits necessarily entails a violation of monogamy of entanglement.
\end{proof}

\directionmono*

\begin{proof}
    Recall that satisfying a $\Pi_{prop}$ clause between two clock qudits $c_0$ and $c_1$ where $c_0$ is the predecessor of $c_1$ requires forming a Bell pair between the $C\!B$ subspace of $c_0$ and the $C\!A$ subspace of $c_1$. Then, in the case where a clock qudit has two predecessors, satisfying both $\Pi_{prop}$ clauses demands that its $C\!A$ subspace is maximally entangled with the $C\!B$ subspace of two different clock qudits. This is a violation of monogamy of monogamy of entanglement. Similarly, if a clock qudit has two successors, its $C\!B$ subspace would also be required to be maximally entangled with two different clock qudits.
\end{proof}

\uniqueclockmono*

\begin{proof}
    Satisfying a single $\Pi_{init}$ or $\Pi_{out}$ clause requires the endpoints' single subspace $EC$ to be maximally entangled with one of the subspaces of the clock qudit. Then, any two clauses that require entangling $EC$ to two different subspaces (even if its the same $C\!A$ or $C\!B$ of two distinct clock particles) would entail a violation of monogamy of entanglement.
\end{proof}

\uniqueendpointmono*

\begin{proof}
    Satisfying a $\Pi_{init}$ clause implies that the $C\!A$ subspace of the clock qudit present in the clause is maximally entangled with the $EC$ subspace of the endpoint qubit. Similarly, $\Pi_{prop}$ clauses require that the $C\!B$ subspace of the predecessor is maximally entangled with the $C\!A$ subspace of its successor. Then, in a one-dimensional chain, there is only a single free $C\!A$ subspace on clock qudit $c_0$. Therefore, any $\Pi_{init}$ clause that acts on a clock qudit other than $c_0$ violates the monogamy of entanglement conditions. The same argument applies to $\Pi_{out}$, the $C\!B$ subspace, and $c_L$.
\end{proof}

\noendpointsmono*

\begin{proof}
    If the clock component has no $\Pi_{init}$ ($\Pi_{out}$), the assignment where all clock qudits are set to $\ket{r}$ ($\ket{d}$) is a satisfying assignment. The $\Pi_{init}$ and $\Pi_{out}$ clauses contain the $\Pi_{start}$ and $\Pi_{stop}$ projectors forbidding clock qudits to be in the inactive states $\ket{r}$ and $\ket{d}$. Then, if one clause is missing in the clock component, the clock qudits can be set to the inactive state that is not forbidden. Furthermore, this assignment does not violate any of the $\Pi_{clock}$ clauses (whether they are $\Pi_{clock,D}$ or $\Pi_{clock,?}$) since these allow both the successor and predecessor to be in the same inactive state. Additionally, the logical qudits of the $\Pi_{prop}$ clauses remain unconstrained since there is not active state to demand otherwise.

    Clearly, if the instance has neither of these two clauses, then either assignment of the clock qudits suffices.
\end{proof}
\noindent

\activetaccmono*

\begin{proof}
    For the sake of contradiction, assume that a satisfying state cannot have clock qudit $c_0$ in the state $\ket{a}$. Let us show that all other assignments of the clock qudits result in either a violation of the assumption or a violation of a clause in the sub-instance.

    To begin, observe that $c_0$ cannot be $\ket{r}$ as this would violate $\Pi_{start}$. $c_0$ must then be permanently fixed to $\ket{d}$. Now, observe that if $\Pi_{prop}^{(c_0, c_1)}$ acts on an undefined logical qudit, this results in a violation of the $\Pi_{clock,?}$ clause within, since $\ket{d_0 \textnormal{x}_1}$ is not in the null space of this clause for any $\textnormal{x} \in \{r,a,d\}$ (see \cref{eqn:nullspacehprop}). Thus, suppose this first $\Pi_{prop}$ clause is well-defined. In this case, the $\Pi_{clock,D}$ clause only allows $c_1$ to be either $\ket{d}$ or $\ket{a}$. If we let $c_1$ be $\ket{a}$, satisfying $\Pi_{prop}^{(c_0, c_1)}$ requires the active spot to shift towards the left (see \cref{eqn:propsat}), in which case $c_0$ becomes $\ket{a}$, contradicting the assumption. Therefore $c_1$ must also be $\ket{d}$. By induction, these steps show that the only state of the clock that can satisfy the $\Pi_{prop}$ (and $\Pi_{clock}$ clauses within) must be $\ket{d_0 d_1 \ldots d_L}$. However, this state violates the $\Pi_{stop}$ clause acting on clock qudit $c_L$. Therefore, there is no state that avoids contradicting the assumption or a violation of one the clauses.
\end{proof}

\uniqueactivetaccmono*

\begin{proof}
    Observe that since all $\Pi_{prop}$ clauses within the chain are oriented towards $c_T$, the only assignment of the clock qudits that agrees with the $\Pi_{clock,D}$ clauses is that where qudits $c_i$ with $i < a$ are in the state $\ket{d}$ and qudits $c_i$ with $i > a$ are in the state $\ket{r}$. Therefore, if $c_b$ is also in the state $\ket{a}$, there must be a violation of a $\Pi_{clock,D}$ clause.
\end{proof}

\truncated*

\begin{proof}
    To show that this state is a satisfying assignment of the sub-instance, we first consider the satisfying assignment of the TACC. Afterwards, we determine which states the other qudits should be in order to satisfy the remaining clauses.

    According to \cref{prop:activetaccmono} and \cref{prop:uniqueactivetaccmono}, if a satisfying assignment of the TACC exists, then it must must contain exactly one clock qudit in an active state at any given time. These conditions mirror those of an active clock chain devoid of undefined clauses, which we know is most likely satisfied by a history state. The only uncertainty is whether this state also satisfies the final $\Pi_{out}$ clauses. For the TACC considered here, there is no such uncertainty as the $\Pi_{out}$ clauses on qudit $c_L$ are never part the chain. Thus, this TACC is satisfied by the history state expressing the computation of circuit $U_{T} \ldots U_1$.

    For the satisfiability of the clauses outside the TACC, observe that the undefined $\Pi_{prop}^{(c_T, c_{T+1})}$ clause can be satisfied alongside the TACC by fixing the state of the successor qudit to $\ket{r}$. Similarly, the $\Pi_{prop}$ clauses acting on the clock qudits that follow ($c_{T+1}, \ldots, c_L$) can also be satisfied by fixing the clock qudits to state $\ket{r}$, regardless of whether they are well-defined or undefined. As a consequence, the clauses outside the TACC do not exert any constraints on the logical qudits. Finally, observe that since the logical qudits in the set $S \setminus D$ do not form part of clauses within the TACC, they may be in any state $\ket{\phi}$. \cref{eqn:truncatedhist} is the state just described.

    The instance is trivially satisfiable because the $\Pi_{out}$ clauses lie outside the TACC and hence do not demand that the logical qudits they act on are in the state $\ket{1}$ at the end of the computation. Consequently, there is no need for a quantum computer to verify whether the circuit produces such results.
\end{proof}

\subsection{Logical Qudits}

\sharedhprop*

\begin{proof}
    To begin, recall that the unique satisfying state of a TACC is the history state corresponding to the circuit expressed by the chain. This state is pure. If we suppose that the chain has length $T \geq 1$, the history state is given by $\ket{\psi_{hist}} = \frac{1}{\sqrt{T+1}} \sum_{t = 0}^T \ket{\psi_t} \otimes \ket{C_t}$, where we use the shorthand $\ket{C_t}$ to represent the legal states of the clock register, and $\ket{\psi_t} := U_t \ldots U_0 \ket{0}^{\otimes q}$. Notice that due to our choice of universal gate set, when a logical qudit is acted on by a unitary for the first time, the state of the qudit necessarily changes. That is, if we suppose a unitary acts on a fresh logical qudit at time $\tau \in [T]$, then $\ket{\psi_{\tau+1}} \neq \ket{\psi_\tau}$, meaning that at time $\tau+1$, the logical and clock register become entangled.

    Now, consider an instance with two disjoint clock components $C_1$ and $C_2$, each with a respective set of logical qudits $L_1$ and $L_2$ with a non-empty intersection $L_3$. Furthermore, suppose that $C_1 \cup L_1$ is the component that surely acts with a unitary on a shared logical qudit. For the sake of contradiction, assume that $C_1 \cup L_1$ and $C_2 \cup L_2$ have satisfying history states. Under these conditions, $L_3$ must be entangled with $C_1$, implying that any subsystem including only one of $L_3$ or $C_1$ will inevitably be impure. $C_2 \cup L_2$ includes $L_3$ but not $C_1$, leading to the conclusion that this solution must be impure. This is a contradiction.
\end{proof}

\mutualboth*

\begin{proof}
    Satisfying the $\Pi_{init}$ clause demands that the logical qudit must be in the state $\ket{0}$ at the beginning of the computation of the first component. As there cannot be any unitaries acting on the qudit, the logical qudit should remain fixed to this state. Therefore, the logical qudit cannot satisfy the $\Pi_{out}$ clause requiring that it is either $\ket{1}$ or $\ket{?}$.
\end{proof}

\mutualsingle*

\begin{proof}
    In the former case, the $\Pi_{init}$ clauses can be satisfied by setting the clock qudit to the state $\ket{0}$. In the latter case, the $\Pi_{out}$ clauses can be satisfied by setting the clock qudit to either $\ket{1}$ or $\ket{?}$.
\end{proof}

\subsection{Simultaneous Propagation}

\simultaneousprop*

\begin{proof}
    First, recall that as shown in \cref{eqn:propsat}, a $\Pi_{prop}$ clause with associated unitary $U$ can be satisfied by three distinct states: either $\ket{\phi} \otimes \ket{rr}$, $\ket{\phi} \otimes \ket{dd}$, or $2^{-1/2}(\ket{\phi} \otimes \ket{a d} + U \ket{\phi} \otimes \ket{da})$. Here, $\ket{\phi}$ is an arbitrary state of the two logical qudits. As we have established in \cref{prop:activetaccmono}, all clock qudits of the TACC must be active at some point in time and so here we must consider this latter case.

    $(\Rightarrow)$: Assume the $k$ $\Pi_{prop}$ clauses are satisfiable. Then, as just mentioned, the $i$-th $\Pi_{prop}$ clause must be satisfied by a state of the form $2^{-1/2}(\ket{\phi_t} \otimes \ket{a_t d_{t+1}} + U_{t+1,i} \ket{\phi_t} \otimes \ket{d_{t} a_{t+1}})$. To simultaneously these clauses, it is then evident that $U_{t+1,i} \ket{\phi_t} = U_{t+1,j} \ket{\phi_t}$ for all $i,j \in [k]$.

    $(\Leftarrow)$: Assume that $U_{t+1,i} \ket{\phi_t} = U_{t+1,j} \ket{\phi_t}$ for all $i,j \in [k]$. In other words, $U_{t+1,i} \ket{\phi_t} = \ket{\phi_{t+1}}$ for all $i \in [k]$. Then, the state $2^{-1/2}(\ket{\phi_t} \otimes \ket{a_t r_{t+1}} + \ket{\phi_{t+1}} \otimes \ket{d_t a_{t+1}})$ satisfies all $k$ $\Pi_{prop}$ clauses.
\end{proof}

\measuringprop*

\begin{proof}
    Before starting the proof, there is a clarification that needs to be made. The claim states that we measure the eigenvalue of projector $\Pi_{prop}$ acting on clock qudits $c_t$ and $c_{t+1}$, yet, it has never been explicitly defined. Really, we mean that we measure the eigenvalue of the $\Pi_{work}$ projector within the $O_{prop}$ operator. This is the only projector of interest as all other projectors have already been checked by the classical algorithm. Indeed, the privileged history state $\ket{\psi_{phist}}$ has a correct assignment of the clock, logical, and endpoint qudits, all of its auxiliary subspaces are entangled accordingly, all of the clauses within the chain are well-defined, and the clock register represents a legal state of the clock at all times. $\Pi_{work}$ is the only term left that has not been verified.

    To simplify things further, observe that since the projector is only concerned with two clock qudits $\Pi_{work} \ket{\psi_{phist}} = \Pi_{work} \ket{\theta_{t,t+1}}$, where $\ket{\theta_{t,t+1}} := 2^{-1/2} (\ket{\phi_t} \otimes \ket{C_t} + U_{t+1,j} \ket{\phi_t} \otimes \ket{C_{t+1}})$. Thus, we only care about measuring the eigenvalue of $\Pi_{work}$ in the state $\ket{\theta_{t,t+1}}$.

    To begin the proof, recall that the probability of measuring eigenvalue $1$ of a projector $\Pi_i$ on a state $\ket{\psi}$ is given by $\bra{\psi} \Pi_i \ket{\psi}$ (see \cref{eqn:probmeasure}). In our case, this probability is given by $\bra{\theta_{t,t+1}} \Pi_{work} \ket{\theta_{t,t+1}}$. First, note that the action of the projector on the state is given by

    \begin{equation*}
        \begin{aligned}
            \Pi_{work} \ket{\theta_{t,t+1}} & = \frac{1}{2\sqrt{2}} \left[\ket{\phi_t} \otimes \ket{C_t} + U_{t+1,0} \ket{\phi_t} \otimes \ket{C_{t+1}} - U_{t+1,j} \ket{\phi_t} \otimes \ket{C_{t+1}} - U_{t+1,j}^\dagger U_{t+1,0} \ket{\phi_t} \otimes \ket{C_t} \right] \\
                                            & = \frac{1}{2\sqrt{2}} \left[ (I - U_{t+1,j}^\dagger U_{t+1,0}) \ket{\phi_t} \otimes \ket{C_t} + (U_{t+1,0} - U_{t+1,j}) \ket{\phi_t} \otimes \ket{C_{t+1}} \right],
        \end{aligned}
    \end{equation*}
    \noindent
    and so the probability of measuring eigenvalue $1$ is
    \begin{equation*}
        \begin{aligned}
            \bra{\theta_{t,t+1}} \Pi_{work} \ket{\theta_{t,t+1}} & = \frac{1}{4} \left[  \bra{\phi_t} I - U_{t+1,j}^\dagger U_{t+1,0} \ket{\phi_t} + \bra{\phi_t} I - U_{t+1,0}^\dagger U_{t+1,j} \ket{\phi_t} \right] \\
                                                                 & = \frac{1}{2} \Re( \bra{\phi_t} I - U_{t+1,j}^\dagger U_{t+1,0} \ket{\phi_t})                                                                       \\
                                                                 & = \frac{1}{2} - \frac{1}{2} \Re(\bra{\phi_t} U_{t+1,j}^\dagger U_{t+1,0} \ket{\phi_t}).
        \end{aligned}
    \end{equation*}
    On the other hand, the action of circuit $\mathcal{C}$ on the state $\ket{\phi_t}$ yields

    \begin{equation*}
        \begin{aligned}
            \mathcal{C} \ket{\phi_t} & = (XH \otimes I) \left[ \frac{\ket{0} \otimes \ket{\phi_t} - \ket{1} \otimes U_{t+1,j}^\dagger U_{t+1,0} \ket{\phi_t}}{\sqrt{2}} \right]                     \\
                                     & = \frac{1}{2}\left[ \ket{1} \otimes (I - U_{t+1,j}^\dagger U_{t+1,0}) \ket{\phi_t} + \ket{0} \otimes (I + U_{t+1,j}^\dagger U_{t+1,0}) \ket{\phi_t} \right],
        \end{aligned}
    \end{equation*}
    \noindent
    from which it follows that the probability of measuring outcome ``1'' is

    \begin{equation*}
        \begin{aligned}
            \Pr(\textnormal{outcome 1}) & = \bra{\phi_t} \mathcal{C}^\dagger (\ket{1} \! \bra{1} \otimes I) \mathcal{C} \ket{\phi_t}                                                    \\
                                        & = \frac{1}{4} \left[ \bra{\phi_t} (I - U_{t+1,0}^\dagger U_{t+1,j}) (I - U_{t+1,j}^\dagger U_{t+1,0}) \ket{\phi_t} \right]                    \\
                                        & = \frac{1}{4} \left[  2 - \bra{\phi_t}U_{t+1,0}^\dagger U_{t+1,j}\ket{\phi_t} - \bra{\phi_t} U_{t+1,j}^\dagger U_{t+1,0} \ket{\phi_t} \right] \\
                                        & = \frac{1}{2} - \frac{1}{2} \Re(\bra{\phi_t} U_{t+1,j}^\dagger U_{t+1,0} \ket{\phi_t}).
        \end{aligned}
    \end{equation*}
    \noindent
    Since both methods produce the same statistics, we can say that they are equivalent.

    Although the Hadamard gates originally belong to $\mathcal{G}$, the controlled version of $U_{t+1,0}$, $U_{t+1,j}$ and the Pauli-X gate do not. However, they can be decomposed perfectly into gates of this set. Indeed, observe that in the computational basis, the $3$- or $2$-local controlled unitaries can be written as $\ket{0} \! \bra{0} \otimes I + \ket{1} \! \bra{1} \otimes U$, and the $X$ gate as $\ket{0}\!\bra{1} + \ket{1}\!\bra{0}$. In this form it becomes clear that the matrix elements are of the form $\frac{1}{2^s}(a + ib + \sqrt{2}c + i\sqrt{2}d)$ with $s \in \mathbb{N}$ and $a,b,c,d \in \mathbb{Z}$. Therefore, by the result of Giles and Selinger \cite{giles2013exact}, these gates can be implemented perfectly with a constant number of gates from the set $\mathcal{G}_8 = \{H,T,\textnormal{CNOT}\}$ if given access to an additional ancilla qubit.
\end{proof}

\section{Proofs of Lemmas: Section~\ref{section:our}} \label[appendix]{appendix:nomonogamy}

\singletypequdit*

\begin{proof}
    The proof of \cref{lemma:singletypequditmono} is still valid for this statement if we simply omit $\Pi_E$.
\end{proof}

\noendpoints*

\begin{proof}
    Since the proof of \cref{lemma:noendpointsmono} is agnostic to the direction of the $\Pi_{prop}$ clauses, it is still valid for this statement.
\end{proof}

\activetacc*

\begin{proof}
    The proof is similar to that of \cref{prop:activetaccmono} except that here the $\Pi_{prop}$ clauses within the chain may point in different directions, and there might be other $\Pi_{init}$ clauses in the chain. As we will see, the latter does not influence the arguments below in any way.

    For the sake of contradiction, assume the satisfying state cannot have clock qudit $c_0$ in the state $\ket{a}$. We now show that all other assignments of the clock qudits result in either a violation of the assumption or a violation of a clause in the sub-instance.

    To begin, observe that $c_0$ cannot be $\ket{r}$ as this would violate $\Pi_{start}$. Thus, $c_0$ must be permanently fixed to $\ket{d}$. Next, let us show that under this assumption, there can be no undefined clause.

    Suppose that $\Pi_{prop}^{(c_t, c_{t+1})}$ with $0 \leq t < T$ is the first undefined clause of the chain and points in an arbitrary direction. As we have just established, if $c_t = c_0$, it cannot be in the active state. Then, for a different $c_t$, there must be a chain of clock qudits $c_0, \ldots, c_{t}$ connected by well-defined $\Pi_{prop}$ clauses. Observe that if any of these qudits are in an active state, satisfying the $\Pi_{prop}$ clauses (regardless of their direction) requires this state to shift towards the left until $c_0$ becomes $\ket{a}$, contradicting the initial assumption. This chain hence cannot have any active states, and moreover they must all be $\ket{d}$ in order to satisfy the $\Pi_{clock,D}$ clauses. However, this assignment includes $c_t$, and as seen in \cref{eqn:nullspacehprop}, $\ket{d_t \textnormal{x}_{t+1}}$ does not satisfy $\Pi_{clock,?}^{(c_t,c_{t+1})}$ for any state of the qudit $c_{t+1}$. This shows that $c_0, \ldots, c_L$ cannot have an undefined clause.

    Then, suppose the chain only contains well defined clauses. Given that the ACC has at least one $\Pi_{init}$ and a one $\Pi_{out}$ clause, there must be an active state somewhere. Observe that because the clauses are well-defined, to satisfy the $\Pi_{prop}$ clauses, this state is again free to shift towards the left until $c_0$ becomes active. Therefore, we can conclude that the assumption cannot hold.

    We remark that the extra $\Pi_{init}$ clauses in the interior of the chain can only set the constraints that their clock qudits cannot be in state $\ket{r}$. This does not influence any of the arguments above.
\end{proof}

\hinitdirections*

\begin{proof}
    To show this, first recall that by \cref{prop:activetacc} this clock qudit must be in state $\ket{a}$ at some point in time. Then, as shown in \cref{eqn:nullspacehprop}, there is no state that satisfies an undefined clause if the successor qudit is in state $\ket{a}$.

    On the other hand, if the $\Pi_{prop}$ clause is well-defined, satisfying this clause requires the active state to shift towards the right. Then, to keep a valid state of the clock, $c_0$ must become $\ket{r}$ which is forbidden by the $\Pi_{start}$ clause.
\end{proof}

\begin{prop} \label{prop:uniquesucc}
    Let $c_0, \ldots, c_T$ with $1 < T \leq L$ be the clock qudits of a TACC of non-zero length. If a qudit $c_j$ with $0 < j < T$ has more than one well-defined successor or predecessor, the instance is unsatisfiable.
\end{prop}

\begin{proof}
    \cref{fig:uniquesucc} illustrates these two cases. Since $c_j$ is part of a TACC, there will be a time where it is  $\ket{a}$. Without loss of generality, let us assume that this is the initial state of the clock. Then, in the case where the qudit has two successors, the only valid starting state of the clock is $\ket{\ldots r a r \ldots }$. Then, observe that the state that satisfies $\Pi_{prop}^{(1)}$, given by $\ket{\dots r a r \ldots } + \ket{\ldots a d r \ldots }$, violates $\Pi_{clock}^{(2)}$ since the second basis state has a $\ket{d}$ preceded by $\ket{r}$. For the case where the qudit has two predecessors, the valid starting state of the clock is $\ket{\dots d a d \ldots}$. In a similar way as before, the state that satisfies $\Pi_{prop}^{(1)}$ is $\ket{\dots d a d \ldots} + \ket{\ldots a r d \ldots}$, which has a troublesome term where $\ket{r}$ is preceded by $\ket{d}$.
\end{proof}

\singlewelldefined*

\begin{proof}
    \cref{prop:activetacc} states that this clock qudit must be in the state $\ket{a}$ at some point in time. Satisfying the well-defined $\Pi_{prop}$ clause then requires the clock qudit to become $\ket{d}$. If we suppose the second $\Pi_{prop}$ clause is undefined, then this already presents a violation of the undefined clause since as can be seen in \cref{eqn:nullspacehprop}, there is no satisfying state where the clock qudit is in this state. On the other hand, if the second $\Pi_{prop}$ clause is also well-defined, the clock qudit has two successors and can be shown to be unsatisfiable by \cref{prop:uniquesucc}.
\end{proof}

\uniqueactivetacc*

\begin{proof}
    Depending on how the $\Pi_{prop}$ clauses are oriented in the instance, there are many cases in which such a state already violates numerous $\Pi_{clock}$ clauses. Let us assume that the $\Pi_{prop}$ clauses are oriented in such a way that the state does not immediately violate any $\Pi_{clock}$ clause.\footnote{An example of such instances can be seen in \cref{fig:uniquesucc}. Suppose the TACC only consists of these three clock qudits and let the left-most and right-most qudits be in the state $\ket{a}$. Then, $\ket{ada}$ (for the left image) and $\ket{ara}$ (for the right image) do not immediately violate any $\Pi_{clock}$ clause of the instance.} Then, given that the $\Pi_{prop}$ clauses within the chain are well-defined, satisfying the $\Pi_{prop}$ clauses in between qudits $c_a$ and $c_b$ by shifting the left active spot towards $c_b$ implies that at some moment, we must reach a basis state of the form $\ket{\ldots a_{b-1} a_b \ldots}$. This state violates $\Pi_{clock}^{(c_{b-1}, c_b)}$.
\end{proof}

\begin{figure}
    \centering
    \includegraphics[width=0.9\textwidth]{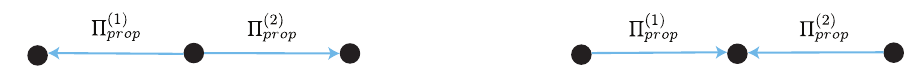}
    \caption{Left: A clock qudit in a chain with two successors. Right: A clock qudit in a chain with two predecessors.}
    \label{fig:uniquesucc}
\end{figure}

\linearchain*

\begin{proof}
    Since $c_t$ represents a qudit within the TACC that is not situated at the ends of the chain, it must be connected to two other qudits via well-defined $\Pi_{prop}$ clauses. As stated by the previous lemma, one of these qudits must serve as the predecessor, and the other as the successor. Then, depending on the direction of the third clause, $c_t$ results in having either two successors or two predecessors. To show that this unsatisfiable, we must consider the cases where the clause could be well-defined or undefined. \cref{prop:uniquesucc} states that the former case is unsatisfiable. The latter requires a deeper analysis.

    Let $c_{t'}$ refer to the other clock qudit of this undefined clause. Without loss of generality, assume that the starting state of the clock is when $c_{t+1}$ is in state $\ket{a}$. Then, the state of the clock that satisfies the $\Pi_{clock}$ clauses is $\ket{ \ldots d_{t} a_{t+1} r_{t+1} \ldots \textnormal{x}_{t'}}$ where we have appended the state of qudit $c_{t'}$ at the end of the clock register. This already poses a problem since $\ket{d_{t} \textnormal{x}_{t'}}$ is not in the null space of the undefined $\Pi_{prop}^{(c_t,c_{t'})}$ clause for any $\textnormal{x} \in \{r,a,d\}$.
\end{proof}

\direction*

\begin{proof}
    To begin, let us assume that $\Pi_{prop}^{(c_0,c_1)}$ is one of the clauses that points towards $c_0$. As shown in \cref{prop:activetacc}, $c_0$ must be $\ket{a}$ at some point in time, meaning that $\ket{ad\ldots}$ must be a basis state in the superposition. Satisfying this clause requires the active state to shift towards the right, implying that $c_0$ must become $\ket{r}$ in order to maintain a valid state of the clock. This clock state violates the $\Pi_{start}$ clause. Therefore, $\Pi_{prop}^{(c_0,c_1)}$ clause must point away from $c_0$.

    Now, suppose that $\Pi_{prop}^{(c_t,c_{t+1})}$ with $1 < t \leq T$ is the first clause that points towards $c_0$.
    Observe that since $\Pi_{prop}^{(c_0,c_1)}$ must point away from $c_0$, and $\Pi_{prop}^{(c_t,c_{t+1})}$ is the first clause of the chain pointing towards this qudit, $c_t$ must have two predecessors. By \cref{prop:uniquesucc}, the instance is then unsatisfiable.
\end{proof}

\outsidetacc*

\begin{proof}
    Case (1) is trivial. By assumption, $c_0$ is present in a well-defined clause of the TACC that points away from it. This is the same setting as in \cref{lemma:singlewelldefined}, which concludes that the instance is unsatisfiable.

    Consider case (2) and let $c_{T+1}$ be the other qudit of the external $\Pi_{prop}$ clause. Here, we have to show that for both possible directions of the clause, the instance is unsatisfiable. (I) Suppose that $\Pi_{prop}^{(c_T,c_{T+1})}$ is well-defined and points towards $c_T$. Recall that since the clauses of the TACC are all well-defined, we can assume that the clock begins in state $\ket{d_0 \ldots d_{T-1} a_T d_{T+1}}$. Then, regressing this clock to satisfy $\Pi_{prop}^{(c_{T-1}, c_T)}$ yields $\ket{d_0 \ldots d_{T-2} a_{T-1} r_T d_{T+1}}$, which violates $\Pi_{clock}$ as there is a $\ket{r}$ preceded by a $\ket{d}$. (II) If the clause points the other way, the state that satisfies $\Pi_{clock}^{(c_T, c_{T+1})}$ becomes $\ket{d_0 \ldots d_{T-1} d_T a_{T+1}}$, which violates the $\Pi_{stop}$ clause on $c_T$.

    Consider case (3). \cref{prop:activetacc} and \cref{prop:uniqueactivetacc} established that exactly one of the qudits in the TACC must be in an active state. Then, given that all $\Pi_{prop}$ clauses within the TACC are well-defined, satisfying these clauses requires the active state to switch places towards the right until $c_T$ eventually becomes $\ket{a}$. As shown in \cref{eqn:nullspacehprop}, there is no state that satisfies an undefined clause if the successor qudit is in state $\ket{a}$. Therefore, the instance must be unsatisfiable.
\end{proof}

\outsidedot*

\begin{proof}
    The first case was already covered in \cref{lemma:hinitdirections}. For the second case, recall that $c_d$ must be $\ket{a}$ at all times. Satisfying the $\Pi_{prop}$ clause requires the active state to shift towards the other clock qudit, leaving $c_d$ to be in the state $\ket{d}$. If the active dot occurs because there is both a $\Pi_{init}$ and $\Pi_{out}$ clause acting on the qudit, having $c_d$ as $\ket{d}$ causes a violation of the $\Pi_{stop}$ clause. On the other hand, if the active dot occurs because the ACC is immediately truncated by an undefined clause, then this undefined $\Pi_{prop}$ clause is infringed since, as shown in \cref{eqn:nullspacehprop}, there is no state in the null space of the clause that has a $\ket{d}$ in the clock register.
\end{proof}

\section{Choice in Logical Qudits?} \label[appendix]{appendix:choice}

Now, why should we set the state of logical qudits in the instance since the very beginning (globally) and not for each individual sub-instance? To observe why this is, consider an instance where two clock components $A$ and $B$ share a single logical qudit $l_x$. Suppose $A$ acts on $l_x$ only with a $\Pi_{init}$ clause. Then, suppose that $B$ does not act on it with a $\Pi_{init}$ clause, but acts on it with a $\Pi_{prop}$ clauses stemming from its ACC. Let us look at the instance from the point of view of the individual clock components (unaware that $l_x$ is also shared with another system). For component $A$, \cref{prop:activetaccmono} demonstrates that this qudit must be in the state $\ket{0}$ at some point in time. Moreover, as there are no other clauses that change its state, we conclude that it must be fixed to $\ket{0}$. On the other hand, for component $B$, this qudit is unconstrained so it could be anything. Because of this freedom, it labels it as $\ket{?}$ so the component is trivial. Then, in a future step of the analysis of the instance where the shared logical qudits are analyzed, it becomes clear that there is an apparent conflict with this qudit as it cannot be in both states. As system $A$ actually constrains this qudit, to satisfy this constraint, this qudit should be set to $\ket{0}$. This is more complicated than if we had originally considered both components as one and determined that the logical qudit should be set to $\ket{0}$ and then proceeded to evaluate the individual components.
\end{appendices}

\end{document}